\documentclass[conference]{IEEEtran}

\newif\ifdraft
\draftfalse 

\usepackage{changes}
\usepackage[users=R, \ifdraft \else hide \fi]{uchanges}
\ChangesSetUserColor{R}{blue}
\usepackage{cite}
\usepackage{amsmath,amssymb,amsfonts}
\usepackage{algorithmic}
\usepackage{graphicx}
\usepackage{textcomp}
\usepackage{xcolor}
\def\BibTeX{{\rm B\kern-.05e m{\sc i\kern-.025em b}\kern-.08em
    T\kern-.1667em\lower.7ex\hbox{E}\kern-.125emX}}
    
\usepackage[utf8]{inputenc}
\usepackage{graphicx}
\usepackage{hyperref}
\usepackage{float}

\usepackage{etoolbox}

\usepackage{mathptmx}
\usepackage{amsmath}
\usepackage{amssymb}
\usepackage{amsthm}

\usepackage{diagbox}
\usepackage{extarrows}
\usepackage{stmaryrd}
\usepackage{tabularx}
\usepackage{cleveref}

\usepackage{tikz}
\usetikzlibrary{arrows}
\usetikzlibrary{positioning}

\usepackage{enumitem}

\apptocmd{\thebibliography}{\raggedright}{}{}

\newtheorem{theorem}{Theorem}[section]
\newtheorem{lemma}[theorem]{Lemma}

\newtheorem{corollary}{Corollary}
\newtheorem{definition}{Definition}[section]
\newtheorem{example}{Example}[section]
\newtheorem{remark}{Remark}

\newcommand{\srp}{$\textsf{CR}$}
\newcommand{\srs}{$\textsf{SR}_{\subseteq}$}
\newcommand{\SR}{\textsf{SR}}
\newcommand{\sNE}{\textsf{sNE}}

\newcommand{\U}{\mathbb{U}}
\newcommand{\mS}{\mathcal{S}}
\newcommand{\R}{\mathbb{R}}
\newcommand{\mH}{\mathcal{H}}
\newcommand{\T}{\mathcal{T}}

\newcommand{\orcidID}[1]{\href{http://orcid.org/#1}{\raisebox{-1.25pt}{\includegraphics{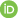}}}}

\begin{document}

\title{Towards a Game-Theoretic Security Analysis of Off-Chain Protocols}
	
	\author{\IEEEauthorblockN{Sophie Rain \orcidID{0000-0002-8940-4989}}
\IEEEauthorblockA{
\textit{TU Wien}, Austria}

\and

\IEEEauthorblockN{Georgia Avarikioti}
\IEEEauthorblockA{
\textit{TU Wien}, Austria }

\and

\IEEEauthorblockN{Laura Kov\'acs \orcidID{0000-0003-0845-5811}}
\IEEEauthorblockA{
\textit{TU Wien}, Austria }

\and

\IEEEauthorblockN{Matteo Maffei}
\IEEEauthorblockA{ \textit{Christian Doppler Lab Blockchain} \\ \textit{Technologies for the Internet of
Things}\\
\textit{TU Wien}, Austria }}


\maketitle

	\begin{abstract}
		Off-chain protocols constitute one of the most promising approaches to solve the inherent scalability issue of blockchain technologies. The core idea is to let parties  transact on-chain only once to establish a channel between them, leveraging later on the resulting channel paths to perform arbitrarily many peer-to-peer transactions off-chain. While significant progress has been made in terms of proof techniques for off-chain protocols, existing approaches do not capture the game-theoretic incentives at the core of their design, which led  to overlooking significant attack vectors like the Wormhole attack in the past. 

In this work we take a first step towards a principled game-theoretic \NEWR{security} analysis of off-chain protocols by introducing the first game-theoretic model that is expressive enough to reason about their security. We advocate the use of Extensive Form Games (EFGs) 
and introduce two instances of EFGs to capture security properties of the closing and the routing of the Lightning Network. Specifically, we  model the closing protocol, which relies on punishment mechanisms to disincentivize parties to upload old channel states on-chain. Moreover, we model the routing protocol, thereby formally characterizing the Wormhole attack, a vulnerability that undermines the   fee-based incentive mechanism underlying the Lightning Network. 

This work is the extended version of our CSF 2023 paper ''Towards a Game-Theoretic Security Analysis of Off-Chain Protocols''.
	\end{abstract}


	\section{Introduction}

Blockchain technologies are emerging as a revolutionary paradigm to perform secure  decentralized financial applications. Nevertheless, a widespread adoption of  cryptocurrencies, such as Bitcoin~\cite{nakamoto2008bitcoin} and Ethereum~\cite{wood2014ethereum}, is severely hindered by their inherent limitations on transaction throughput~\cite{ScalingBC,ScalingBC2}. For instance, while Bitcoin can support tens of transactions per second and the  confirmation time is about an hour,  traditional credit networks like Visa can comfortably handle up to 47,000 transactions per second.

Off-chain protocols~\cite{gudgeon2020sok} are recognized  as one of the most promising scalability solutions,  achieving a seemingly contradictory property: the bulk of transactions is performed off-chain, and yet in a secure fashion. The idea is to leverage the blockchain only in case of disputes, resorting otherwise to off-chain, peer-to-peer transactions. Bitcoin's Lightning Network \cite{lightning} is the most widely adopted  off-chain instantiation, hosting  at the time of writing  bitcoins worth more than 170M USD, in a total of more than 27,000 nodes and more than 76,000 channels. 
In a nutshell, parties deposit money in a shared address, called channel, and can later on perform arbitrarily many off-chain transactions with each other by redistributing the deposit on the channel. In the end, the channel can be closed and the latest state (i.e., deposit distribution) is posted on-chain.
Off-chain transactions
are not limited to the end-point of the channel, but they
can be routed over paths of channels (so-called multi-hop payments). Besides such payment channel networks, an entire ecosystem of  off-chain protocols~\cite{gudgeon2020sok} (virtual channels, watchtowers,  payment-channel hubs, state channels, side-chains, etc.) is under development for Bitcoin~\cite{AMHL,GenChannels,AMEEFRHM21,AumayrMKM21,avarikioti2019brick,HTLC}, Ethereum~\cite{DziembowskiEFM19,DziembowskiEFHH19,McCorryBBM019,avarikioti2020cerberus}, as well as other cryptocurrencies~\cite{ThyagarajanMSS20}. 

The cryptographic protocols underlying these off-chain constructions are rather sophisticated and, most importantly, rely on game-theoretic arguments to discourage malicious behavior. For instance, the Lightning Network relies on a punishment mechanism to disincentivize parties to publish  old states on-chain and on an unlocking mechanism where parties first pay a neighbor and then retrieve the paid amount from the other to ensure the atomicity of multi-hop payments (i.e., either all channels are consistently updated or none is).

Off-chain protocols are typically subject to rigorous security analyses, which however  concentrate on  cryptographic properties and do not capture the game-theoretic ones. In particular, most protocols are proven secure in the Universal Composability framework~\cite{TCC:CDPW07}, proving that the cryptographic realization simulates  the ideal functionality.  This framework, however,  was developed to reason about security in the classical  honest/Byzantine setting: in particular,  the ideal functionality has to model all possible parties' behavior, rational and irrational, otherwise it would not be simulatable, but reasoning on whether or not certain behavior is rational is  outside of the model and thus left to informal arguments. This is not just a theoretical issue, but a practical one, as there is the risk to let attacks pass undetected: for instance, the Wormhole attack 
~\cite{AMHL} constitutes a rational behavior in the Lightning Network, which is thus admitted in any faithful model thereof although  it undermines its  incentive mechanism. The first step towards closing this  gap in cryptographic proofs is to come up with a  \emph{faithful game-theoretic model for off-chain protocols} in order to reason about security in the presence of rational parties.  We address this challenge in this paper, advocating the use of Extensive Form Games (EFGs) for the game-theoretic security analysis of off-chain protocols. In particular, we introduce two instances of EFGs to model the closing and the routing of the Lightning Network.

\subsection{Related Work} 
A  game-theoretic model for  off-chain protocols is initiated and introduced in~\cite{CITE}. This work \REPLACER{is however limited in a number of ways}{suffers, however, from several limitations, which make it unsuitable to conduct faithful security analyses}.
Firstly, the game model considers only honest closing of channels, i.e., all deviations --   such as posting an old state --  are ignored: this makes it impossible to reason about the security of basic  channel operations. Secondly, the pay-offs are represented as constants, which neglects the dependency of the channel's balance on its security properties. Further,  fees are not considered at all, thereby ignoring their impact on Lightning protocols. For instance, the routing game  to model the security of multi-hop payments fails to capture already identified attacks in payment channel networks, like the Wormhole attack~\cite{AMHL} that targets  the fee distribution among players. \NEWR{Additionally, Lightning is vulnerable to the Griefing attack~\cite{griefing}, where a significant amount of money is locked.}
In our work, we overcome the aforementioned limitations, by defining a stronger closing phase model, by aligning the utilities to the monetary outcome, by considering all possible deviations of parties during closing, and by revising the relevant security properties. 
We  demonstrate the importance of precision in game-theoretic protocol models by modeling the Wormhole attack\NEWR{, as well as the Griefing attack}.


Our work further complements other game-theoretic advancements in the area, most prominently the following lines of research. 

\paragraph{Incentivizing Watchtowers}
A major drawback of payment channel protocols is that channel participants must frequently be online and watch the blockchain to prevent cheating. To alleviate this issue, the parties can employ third parties, or so-called watchtowers, to act on their behalf in case their counterparty misbehaves.  Correctly aligning the incentives of watchtowers to yield a secure payment channel protocol is, however, challenging. This is the main focus of several  works~\cite{McCorryBBM019,avarikioti2018towards,avarikioti2020cerberus,avarikioti2019brick}. 
As their objective is to incentivize external parties, their models does not apply in our work. 

\paragraph{Payment Channel Network Creation Games}
Avarikioti et al.~\cite{avarikioti2020ride,avarikioti2019payment} study payment channel networks as network creation games. Their goal is to determine which channels a rational node  should establish to maximize its profit.
Ersoy et al.~\cite{ersoy2020profit} undertake a similar task; they formulate the same problem 
as an optimization problem, show it is NP-hard, and present a greedy algorithm to approximate it.
Similarly to our work, all these works assume rational participants. However, we aim to model the security of the protocols, in contrast to these works that study the network creation problem graph-theoretically.

\paragraph{Blockchains with Rational Players}
Blockchains incentivize miners to participate in the network via monetary rewards~\cite{nakamoto2008bitcoin}. Therefore, analyzing blockchains under the lens of rational participants is critical for the security of the consensus layer. There are multiple works in this direction: 
Badertscher et al.~\cite{badertscher2018but} present a rational analysis of the Bitcoin protocol. 
Eyal and Garay~\cite{eyal2014majority} introduce an attack on the Nakamoto consensus, effectively demonstrating that rational miners will not faithfully follow the Bitcoin protocol. This attack is generalized in~\cite{kwon2017selfish,sapirshtein2016optimal}.
Consequently, Kiayias et al.~\cite{kiayias2016blockchain} analyze how miners can deviate from the protocol to optimize their expected outcome. 
Later, Chen et al.~\cite{chen2019axiomatic} investigate the reward allocation schemes in longest-chain protocols and identify Bitcoin's allocation rule as the only one that satisfies a specific set of desired properties.
On a different note, several works study the dynamics of mining pools from a game-theoretic perspective~\cite{eyal2015miner,teutsch2016cryptocurrencies} or introduce network attacks that may increase the profit of rational miners~\cite{heilman2015eclipse,nayak2016stubborn}.
An overview of game-theoretic works on blockchain protocols can be found in~\cite{SurveyOnBC}.

All these works, however, focus on the consensus layer (Layer-1) of block\-chains and as both the goals and assumptions are different from the application layer (Layer-2), the models introduced there cannot be employed for our purposes.
For instance,  payment channel protocols occur off-chain and thus game-based  cryptographic assumptions of the blockchain do not apply. 
In addition, consensus  protocols investigate the  expected reward of miners which is a probabilistic problem, whereas we ask if any honest player could lose money, which depends on the behavior of the other players and is fundamentally deterministic.

 Game-based definitions have also been proposed for the security analysis of smart contracts~\cite{QuantAnalysisSC,ProbSC}. These models, however, target an on-chain setting and are thus not suitable to reason about the specifics of off-chain constructions (e.g., closing games, routing games, etc.).

\subsection{Our Contributions}
In this work, we take the first steps towards closing the gap between security and game-theoretic analysis of off-chain protocols. Specifically, we introduce the first game-theoretic models that are expressive enough to reason about the security of off-chain protocols. 
We model off-chain protocols as games and then analyze whether or not certain security properties are satisfied. 
The design  of our models is driven by two principles: (a)~all possible actions should be represented and (b)~the utility function should mirror the  monetary outcome realistically. 
We aim to ensure that \textit{honest participants do not suffer any damage \ref{P1}}, whereas \textit{deviating from the protocol yields a worse outcome for the adversary \ref{P2}} We will use weak immunity (\Cref{def:wi}) to implement \ref{P1}, and collusion resilience (\Cref{def:CR}) together with practicality (\Cref{def:practicalEFG}) for \ref{P2}. While we believe that our approach  of implementing principles (a) and (b) is easily extensible to other off-chain protocols, in this work we focus on  the Bitcoin Lightning Network, which constitutes the most widely adopted off-chain protocol.
Our technical contributions can be  summarized as follows:
\begin{itemize}
    \item We refine existing game-theoretical concepts in order to reason about the security of off-chain protocols  (\Cref{sec:theory}).
    \item  We introduce the Closing Game $G_c$, the first game-theoretic security model that accurately captures the closing phase of Lightning channels, encapsulating   arbitrary deviations  from the protocol specification (Section~\ref{sec:models}).
    \item We perform a detailed security analysis of $G_c$, formalize folklore security corner cases of Lightning, 
    and present the strategy that rational parties should follow to close their channels in order to maximize their expected outcome relative to the current and previous distribution states (\Cref{sec:sec}).
    \item We identify limitations in prior work \cite{CITE} on game-based modeling of multi-hop payments, putting forward a new game-based definition that is precise enough to cover the Wormhole \NEWR{and the Griefing} attack (\Cref{sec:refine}).
    \NEWR{We further show how to model  Fulgor protocol~\cite{10.1145/3133956.3134096}, a variant of Lightning's routing that prevents the Wormhole attack.} 
    Our formalization leverages  game theory concepts introduced in \Cref{sec:theory} and \Cref{sec:models}, thereby demonstrating       \NEWR{the theoretical expressiveness of our framework to  analyze complex protocols}. \REMOVER{inspired by real-world scenarios}. 
\end{itemize}

\NEWR{In conclusion, our work brings game-theoretical foundations to enforce security of off-chain protocols, by providing a rigorous analysis over security properties expressed through formal requirements over game strategies. We believe, the provided rigor in our paper opens up new venues for automating security analysis via game-theoretic arguments, a challenge which we aim to tackle in future work.\footnote{We refer the interested reader to the appendix for the complete definitions and proofs.} } 

	\section{Background and Preliminaries}\label{sec:prelim}

\subsection{Payment Channel Networks}\label{sec:pcn}
A payment channel~\cite{GenChannels} can be seen as an escrow (or multi-signature), into which two  parties Alice $A$ and Bob $B$  transfer their initial coins with the guarantee that their coins are not locked forever and the agreed balance can be withdrawn at any time.
After that, $A$ and $B$  can pay each other off-chain by  signing and exchanging messages that reflect the updated  balances in the escrow. These signatures can be used at any time to close the channel and distribute the coins on-chain according to the last channel state. In order to discourage parties from posting  an old state on-chain, a punishment mechanism is in place. 
In particular, in Lightning~\cite{lightning}, once $A$ closes the channel, she has to wait a mutually agreed time before getting her coins. Meanwhile, $B$ has the opportunity to withdraw all the coins in the channel (by posting a so-called revocation transaction), including the ones assigned to $A$, if the state posted on-chain by $A$ is not the last one they mutually agreed on. Such a punishment mechanism is of game-theoretic nature: parties can indeed post an old state on-chain, yet they are discouraged to do so.

\NEWR{
In particular, Lightning payment channels operate 
as follows: First, Alice and Bob  create a funding transaction where they input their respective coins; the funding transaction has a single output that can only be spent if both $A$ and $B$ provide their signature (2-out-of-2 multi-signature).
Then, the two parties create the first commitment transaction, i.e., a transaction that spends the output of the funding transaction and returns the initial coins to both parties. In other words, the input of the commitment transaction is the output of the funding transaction while the output of the first commitment transaction is two-fold: the first output returns the coins to $A$ and the second output to $B$. However, the commitment transaction each party holds is not the same. Specifically, the commitment transaction of $A$ has an additional spending condition, a timelock $t$ that signifies the revocation period and is pre-agreed between the two parties; in $A$'s commitment transaction $B$'s output is spendable immediately. Symmetrically, in $B$'s commitment transaction $B$'s output has a timelock $t$ while $A$'s output is spendable immediately. Note that a timelock $t$ is a condition that allows the coins of the output to be spent on-chain only after time $t$ has elapsed from the publication of the transaction.
After $A$ and $B$ sign and exchange the respective first commitment transactions, they proceed to signing the funding transaction and publishing it on-chain. This order is important to avoid hostage situations\footnote{If the funding transaction is published on-chain before the first commitment transactions are signed,  a party may hold the other hostage since none of the parties can close the channel unilaterally but only in collaboration.}.
As soon as the funding transaction is securely published on-chain, $A$ and $B$ can transact off-chain by creating every time a new commitment transaction that depicts the current balance of the joint capital among the two parties. Every time a new commitment transaction is created, the parties reveal a secret to their counterparty that allows their counterparty to spend their own coins immediately  (e.g., $A$ can spend $B$'s coins from the previous commitment) if the previous commitment transaction appears on chain (revocation transaction).
To close a Lightning channel, the two parties can either collaborate and spend the output of the funding transaction, or each of them can close the channel unilaterally by publishing the last commitment transaction.
 Since the commitment transactions each party hold have a timelock, in case of cheating, i.e., publication of a previous commitment transaction on-chain, the counterpart can immediately spend the cheating party's coins, claiming all the coins of the channel, thus punish the cheating party for misbehaving. 
}

 Technically, $A$ and $B$ do not just lock their initial funds but also a certain small amount which will be used as a fee for the closing transaction of the channel. Note that every on-chain transaction requires such a fee $f$. The fee for the opening transaction is paid upon the opening of the channel and is thus irrelevant to our consideration. However, in case $A$ posts an old state on-chain and $B$ performs the revocation transaction -- which is an on-chain transaction -- to prove it, $B$ has to carry the additional transaction fee alone. These facts have an important impact on our game-theoretic models.
 
 
 
 \NEWR{In the following, we refer to \emph{honest closing} when a party unilaterally closes the channel by posting the last commitment transaction or when the parties  close collaboratively, where both parties sign to spend the funding transaction output directly.}

    Off-chain transactions are not limited to the end-points of a channel, as they can be performed whenever sender and receiver are connected by a path of channels with enough capacity. The cryptographic approach to do so exploits hash-time-locked-contracts (HTLC) \cite{HTLC}.  Assume players $A$ and $B$ do not share a channel. Instead,  $A$ has a channel with $E_1$;  $E_1$ has a channel with $I$;  $I$ has a channel with $E_2$;  and $E_2$ has a channel with $B$, as illustrated in \Cref{tbl:honestrouting}. Player $A$ can now send an amount $m$ to player $B$ via the intermediaries $E_1$, $I$ and $E_2$, where each intermediary charges a fee $f$ for the routing service, hence  $A$ should pay $m+3f$. The core idea is that $A$ pays $E_1$, $E_1$ pays $I$, and so forth, until $B$ gets paid.

\begin{figure}[tb!]
	\centering
\begin{tikzpicture}[scale= 0.8, ->]
	\node (1) at (0,0) {$A$}; 
	\node (2) at (2.5,0) {$E_1$};
	\node (3) at (5,0) {$I$};
	\node (4) at (7.5,0) {$E_2$};
	\node (5) at (10,0) {$B$};
	\path (5) edge[out=100, in=80, distance= 1.5cm] node [below] {\color{red} \small $y$} node [ below, pos=0.3] {1.} (1); 
	\path (1) edge node [above] {\footnotesize $(m+3f,\textcolor{red}{y},t_1)$} node [ below, pos=0.3] {2.}(2); 
	\path (2) edge node [above] {\footnotesize $(m+2f,\textcolor{red}{y}, t_2)$} node [ below, pos=0.3] {3.} (3); 
	\path (3) edge node [above] {\footnotesize $(m+f,\textcolor{red}{y}, t_3)$} node [ below, pos=0.3] {4.} (4); 
	\path (4) edge node [above] {\footnotesize $(m,\textcolor{red}{y}, t_4)$} node [ below, pos=0.3] {5.} (5); 
	\path (5) edge[out=240, in=300, distance= 0.8cm] node [below] {\color{cyan} \small $x$}    node [ above, pos=0.3] {6.} (4);
	\path (4) edge[out=240, in=300, distance= 0.8cm] node [below] {\color{cyan} \small $x$} node [ above, pos=0.3] {7.} (3); 
	\path (3) edge[out=240, in=300, distance= 0.8cm] node [below] {\color{cyan}  \small $x$}    node [ above, pos=0.3] {8.} (2); 
	\path (2) edge[out=240, in=300, distance= 0.8cm] node [below] {\color{cyan}  \small $x$}    node [ above, pos=0.3] {9.} (1); 
\end{tikzpicture}
\vspace{-0.2cm}
\caption{Routing in Lightning using HTLCs.}
\label{tbl:honestrouting}
\end{figure}
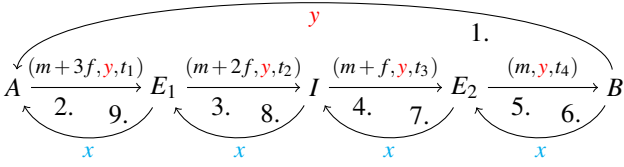

A key security property in multi-hop payments is \emph{atomicity}: either all payments are successful, and the deposit in each channel is updated accordingly, or none is.
To achieve this property, the Lighting protocol proceeds as follows. 
First, the receiver $B$ generates a secret $x$ and sends its hash $h(x)=y$ to the sender $A$ (see action 1 in \Cref{tbl:honestrouting}). 
Then $A$ creates an HTLC for $E_1$, where she locks $m+3f$ with lock $y$ and time-out $t_1$. That means only $E_1$ can claim the money and only by providing a value whose hash is $y$ within time $t_1$ (action 2 in \Cref{tbl:honestrouting}). Although $E_1$ does not know such a value yet and can therefore not unlock, $E_1$ can nevertheless proceed by creating another HTLC for $I$ also locked with $y$ and a time-out $t_2$ (action 3 in \Cref{tbl:honestrouting}). Thereafter, $I$ and $E_2$ continue in the same way (actions 4 -- 5 in \Cref{tbl:honestrouting}). Actions 1 -- 5 of \Cref{tbl:honestrouting} are called the {\it locking phase}. \NEWR{Note that in order to allow everybody to unlock their HTLCs in the subsequent steps, the time-outs have to be decreasing $t_1>t_2>t_3>t_4$.} Once $B$ receives the conditional payment, he can reveal $x$ to $E_2$ and the conditional payment is unlocked (action 6 in \Cref{tbl:honestrouting}). The others can now unlock the HTLCs one after the other from right to left (actions 7 -- 9 in  \Cref{tbl:honestrouting}), which is called the {\it unlocking phase}. Finally, $A$ paid $m+3f$, $B$ received $m$ and each intermediary was rewarded with $f$.

We note that atomicity is achieved by a game-theoretic argument: 
intermediaries can, in principle, stop the protocol either in the locking phase or in the unlocking phase. In the former, they would lose the transaction fee $f$, while in the latter, they would lose the payment amount $m$, $m+f$, $m+2f$ respectively. Thus, they are incentivized to act once they have committed to participate.

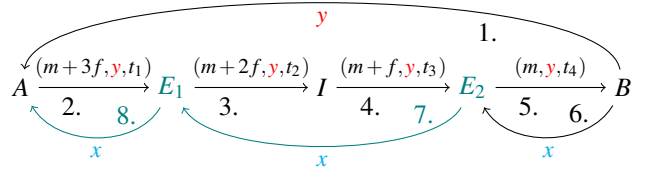
\begin{figure}[tb!]
	\centering
\begin{tikzpicture}[scale= 0.8, ->] 
	\node (1) at (0,0) {$A$}; 
	\node (2) at (2.5,0) {\textcolor{teal}{$E_1$}};
	\node (3) at (5,0) {$I$};
	\node (4) at (7.5,0) {\textcolor{teal}{$E_2$}};
	\node (5) at (10,0) {$B$};
	\path (5) edge[out=100, in=80, distance= 1.5cm] node [below] {\color{red} \small $y$} node [ below, pos=0.3] {1.} (1); 
	\path (1) edge node [above] {\footnotesize $(m+3f,\textcolor{red}{y},t_1)$} node [ below, pos=0.3] {2.}(2); 
	\path (2) edge node [above] {\footnotesize$(m+2f,\textcolor{red}{y},t_2)$} node [ below, pos=0.3] {3.} (3); 
	\path (3) edge node [above] {\footnotesize $(m+f,\textcolor{red}{y},t_3)$} node [ below, pos=0.3] {4.} (4); 
	\path (4) edge node [above] {\footnotesize $(m,\textcolor{red}{y},t_4)$} node [ below, pos=0.3] {5.} (5); 
	\path (5) edge[out=240, in=300, distance= 0.8cm] node [below] {\color{cyan} \small $x$} node [ above, pos=0.3] {6.} (4);
	\color{teal}	\path (4) edge[out=240, in=300, distance= 1cm] node [below] {\color{cyan} \small $x$} node [ above, pos=0.2] {7.} (2); 
	\path (2) edge[out=240, in=300, distance= 0.8cm] node [below] {\color{cyan}  \small $x$} node [ above, pos=0.3] {8.} (1); 
\end{tikzpicture}
\vspace{-0.2cm}
\caption{Wormhole Attack in Lightning.}
\label{tbl:wormhole}
\end{figure}

\paragraph{The Wormhole Attack}
The aforementioned routing protocol is proven to be vulnerable to the \emph{Wormhole attack}~\cite{AMHL}, which is depicted in \Cref{tbl:wormhole}.
The attack is as follows: $E_1$ and $E_2$ collude, and bypass $I$ in the unlocking phase, thus stealing $I$'s participation reward $f$.  Until actions 6 in \Cref{tbl:honestrouting} and \Cref{tbl:wormhole}, the behavior is identical. Then, $E_2$, knowing $x$, forwards $x$ to $E_1$ (offline) instead of unlocking the HTLC from $I$ (action 7 in \Cref{tbl:wormhole}). This way, $E_1$ can unlock $A$'s HTLC and claim the money (action 8), but $I$  will never be able to unlock. After a certain time the remaining HTLCs time-out and the locked money returns to the creators.

 Therefore, the parties $A$ and $B$ are not affected. However,  $E_1$ and $E_2$ collectively earn $3f$ instead of the $2f$ they deserve, stealing the fee $f$ from $I$, who locked resources in the locking phase of the protocol. 
This attack undermines the incentive  of intermediaries to route payments. 

\paragraph{The Griefing Attack}
It describes the scenario  when a player, assume $B$ for simplicity, ignores the proposed payment and refuses to proceed~\cite{griefing}. This way,  money is locked in the conditional payments for a considerable amount of time. While~\cite{CITE-RELATED-WORK} studies
the Griefing attack through probabilistic modelling and~\cite{DBLP:journals/corr/abs-2005-09327} provides mitigation techniques,    to the best of our knowledge there is no  formal security analysis of this attack at present.  Our work addresses this limitation and shows that  Lightning's routing module is indeed susceptible to the Griefing attack.

\NEWR{In the sequel we consider the behavior as illustrated in \Cref{tbl:honestrouting} as the only \emph{honest routing} behavior.}

\subsection{Game-Theoretic Definitions}
\NEWR{We now introduce the  game-theoretic concepts relevant for our formalization.} We denote real numbers by $\R$ and tuples as $\sigma=(\sigma_1,...,\sigma_n)$. \NEWR{We write $\sigma[\sigma'_i/\sigma_i]$ to denote the tuple resulting from substituting $\sigma_i$ by $\sigma'_i$ in $\sigma$, that is $\sigma[\sigma'_i/\sigma_i]=(\sigma_1,...,\sigma_{i-1},\sigma'_i,\sigma_{i+1},...,\sigma_n)$.} We understand games as static objects in which finitely many players can choose finitely many times from a finite set of actions. A game yields a certain positive or negative utility for each player.
We briefly overview the very common Normal Form Games, \NEWR{ also called Strategic Games}~\cite{GameTheoryBook}, in which each player chooses an action \emph{only once}, called strategy. 
\begin{definition}[Normal Form Game -- NFG] \label{def:nfg}
A \emph{Normal Form Game (NFG)} is a tuple $\Gamma=(N,\mS, u)$, where $N$ is the set of game \emph{players}, $\mS=\vartimes_{p \in N} \mS_{p}$ the set of  \emph{joint strategies} $\sigma$ and $u$ the \emph{utility function}:
\begin{itemize}
\item $\mS_p$ is the non-empty set of strategies player $p$ can choose from. Thus, a joint strategy $\sigma \in \mS$ is a tuple of strategies $\sigma=(\sigma_{p_1},...,\sigma_{p_{|N|}})$, with $\sigma_{p_i}\in \mS_{p_i}$.
\item $u=(u_{p_1},\dots,u_{p_n})$, where $u_{p_i}: \mS \to \R$ assigns player $p_i$ its utility for every joint strategy $\sigma \in \mS$.
\end{itemize}
\end{definition}
\REMOVER{An example of an NFG is in \Cref{tbl:nfg}.}
\NEWR{In what follows we fix an arbitrary game $\Gamma$ and give all definitions relative to it. }
To formalize an optimal outcome on game strategies, we use Nash Equilibria. 

\begin{definition}[Nash Equilibrium -- NE]\label{def:NE}
A \emph{Nash Equilibrium} is a joint strategy $\sigma \in \mS$ s.t.\ no player \REPLACER{$p$}{$p_i$} can increase their utility by unilaterally deviating from $\sigma=(\sigma_{p_1},...,\sigma_{p_{|N|}})$,  \REMOVER{That is, choosing a strategy different from $\sigma_{p}$.} Formally, 
\begin{equation} \forall p \in N \; \forall \sigma'_p \in \mS_p:\; u_p(\sigma) \geq u_p(\;\sigma[\sigma'_p/\sigma_p]\;)\;.\end{equation}
\end{definition}

Another important concept is \emph{weakly dominated strategies}, expressing the strategies a rational player would not play since they yield worse utilities.
\begin{definition}[Weakly Dominated Strategy] \label{def:wds}
 A strategy $\sigma_p^d \in \mS_p$ of player $p$ is called \emph{weakly dominated} by strategy $\sigma'_p \in \mS_p$, if it always yields a utility at most as good as $\sigma'_p$ and a strictly worse utility at least once:
\begin{align}
	&\forall\, \sigma \in \mS:\; u_p(\;\sigma[\sigma^d_p/\sigma_p]\;) \leq u_p(\;\sigma[\sigma'_p/\sigma_p]\;) \; \text{and} \\
	& \exists\, \sigma \in \mS:\; u_p(\;\sigma[\sigma^d_p/\sigma_p]\;) < u_p(\;\sigma[\sigma'_p/\sigma_p]\;) \;.
\end{align} 
\end{definition}

\NEWR{
\begin{example}\label{ex:oldclosing}
Consider the NFG $\Gamma_C$ in \Cref{tbl:oldclosing}, which  was introduced in~\cite{CITE} to model closing. In this game $\Gamma_C$, there are two players $N=\{A,B\}$ and each players can choose from the same strategy set $\mS_A=\mS_B=\{U,C,\mathfrak{I}\}$.  
Here, strategy $U$ captures unilateral closing, that is publishing the latest state on-chain. Further, strategy $C$ corresponds for closing collaboratively, that is publishing a mutually signed transaction. Finally, strategy  $\mathfrak{I}$ stands for ignoring, that is doing nothing. The utility for each joint strategy is given in \Cref{tbl:oldclosing}, where player $A$'s strategies are listed in the left column of \Cref{tbl:oldclosing} and the strategies of $B$ are given in the top row of \Cref{tbl:oldclosing}. 

Applying \Cref{def:NE}, the joint strategies $(C,C)$, $(U,\mathfrak{I})$ and $(\mathfrak{I},U)$  are Nash Equilibria: for each of these joint strategies,  neither $A$ nor $B$ can deviate in order to increase their  own utility. Comparing the second and the third row of \Cref{tbl:oldclosing}, we see that $A$'s utility is always as least as good in the second row as it is in the third row. Hence, strategy $C$ weakly dominates strategy $\mathfrak{I}$ for player $A$, by \Cref{def:wds}; the same property  also holds  for player $B$. By comparing the other pairs of rows/columns of  \Cref{tbl:oldclosing}, we see that there is no other weak dominance in $\Gamma_C$. 

\begin{table}	
	\centering	
	\caption{NFG $\Gamma_C$ \NEWR{with players $A,B$}.}
	\begin{tabular}{|c||r|c|c|c|}
	    \hline 
	\diagbox[width=5em,height=2em]{$~~~~A$}{$~~B$}	&{$U$} & \textcolor{black}{$C$} & $\mathfrak{I}$ \\
		\hline\hline 
		$U$  & $(1/2,{1/2})$ & $(0,{1})$ & $(0,1)$ \\
		\hline
		\textcolor{black}{$C$} & $(1,{0})$ & $(1,{1})$ & $(-1,{-1})$ \\
		\hline
		\textcolor{black}{$\mathfrak{I}$} & $({1},{0})$ & $({-1},{-1})$ & $({-1},-1)$ \\
		\hline
	\end{tabular}
	\label{tbl:oldclosing}
\end{table}

\end{example}}

\REMOVER{To formalize strategies where players make multiple choices one after the other, we next consider Extensive Form Games (EFGs), which extend NFGs as follows. 

\begin{definition}[Extensive Form Game -- EFG]
An \emph{Extensive Form Game (EFG)} is a tuple $\Gamma=(N,\mH,P,u)$, where $N$ and $u$  are as in NFGs. The set  $\mH$ captures \emph{game  histories}, $\T \subseteq \mH$ is the set of \emph{terminal histories} and $P$ denotes the \emph{next player function}, satisfying the following properties.
\begin{itemize}
\item The set $\mH$ of histories is a set of sequences of actions with  \begin{enumerate} 
\item $\emptyset \in \mH$; 
\item if the action sequence $(a_k)_{k=1}^K \in \mH$ and $L<K$, then also $(a_k)_{k=1}^L \in \mH$; 
\item a history is terminal $(a_k)_{k=1}^K \in \T$, if there is no action $a_{K+1}$ with  $(a_k)_{k=1}^{K+1} \in \mH$.
\end{enumerate} 
\item The next player function $P$ 
\begin{enumerate}
    \item assigns the next player $p \in N$ to every non-terminal history  $(a_k)_{k=1}^K \in \mH\setminus \T$, that is  $P((a_k)_{k=1}^K) = p$; 
\item after a non-terminal history $h= (a_k)_{k=1}^K \in \mH$, it is player $P(h)$'s turn to choose an action from the action set $A(h)=\{a: (h,a) \in \mH\}$.
\end{enumerate}
\end{itemize}
A \emph{strategy} of  player $p$ is a function $\sigma_p$ mapping every $h \in \mH$ with $P(h)=p$ to an action from $A(h)$. The utilities of all joint strategies with terminal history $h$ are the same.
\end{definition}

Note that the set of terminal histories $\T$ is well-defined by $\mH$ and does therefore not occur in the tuple $\Gamma$. 
EFGs can be well-represented via trees.
\begin{definition}[EFG as tree]
Given an EFG $\Gamma=(N,\mH,P,u)$, then the following tree $G=(V,E)$ represents $\Gamma$.
\begin{itemize}
    \item For every history $h$ exists exactly one node $v_h \in V$. 
    It is labeled $P(h)$, the next player, if $h$ is not terminal ($h \notin \T$), and $u(\sigma)$, the joint utility of playing a game with history $h$, if $h \in \T$ and the joint strategy $\sigma$ yields history $h$. 
    \item Two nodes $v_h$, $v_{h'}\in \mH$ are connected via an oriented edge $(v_h,v_{h'}) \in E$ iff $h'=(h,a)$. This edge is labeled $a$.
\end{itemize}
The graph $G$ is a tree with root $v_{\emptyset}$ representing the empty history, which is labeled with the first player, and with one leaf $v_t$ per terminal history $t \in \T$.  
\end{definition}

Note that (i) a history $h$ is a path starting from the root, and it is terminal if and only if it ends in a leaf; (ii) each internal tree node models the turn of an EFG player $p$ after a non-terminal history $h$. The outgoing edges represent the action set $A(h)$. Further, (iii) each leaf represents the joint utility of every joint strategy, whose  history (path) leads to it; (iv) an EFG strategy $\sigma$ is not a history (path) but a set of edges that contains precisely one outgoing edge per internal node.

\begin{example}
The tree in \Cref{tbl:efg} results from the game $\Gamma_E=(N,\mH, P,u)$, where $N=\{A,B\}$ and the set of histories is  $\mH=\{\emptyset,(a), (b), (b,c), (b,d), (b,d,e), (b,d,f), (b,d,f,g),$  $(b,d,f,h)\}$. The next player function $P$ assigns player $A$ after histories $\emptyset$ and $(b,d)$, and player $B$ after $(b)$ and $(b,d,f)$. Finally, the utility function $u$ assigns  joint utility $(2,2)$ to strategies that yield history $(a)$, utility $(3,1)$ for strategies with history $(b,c)$, utility $(1,1)$ for strategies with history $(b,d,e)$, and $(0,2)$ for strategies resulting in $(b,d,f,h)$.

A strategy $\sigma$ in $\Gamma_E$ is for example: $A$ chooses $a$ after history $\emptyset$, and $f$ after $(b,d)$. $B$ takes $c$ after $(b)$ and $g$ after $(b,d,f)$.
 \begin{figure}
	\centering
\begin{tikzpicture}[scale= 0.9,->,>=stealth',auto,node distance=3cm, el/.style = {inner sep=4pt, align=center, sloped}]
	\node (1) at (0,0) {$A$};
	\node (2) at (1.5,-0.75) {$B$} ; 
	\node (3) at (-1.5,-0.75) {$(2,2)$};
	\node (4) at (3,-1.5) {$A$} ;
	\node (5) at (0, -1.5) {$(3,1)$};
	\node (6) at (4.5,-2.25) {$B$} ; 
	\node (7) at (1.5,-2.25) {$(1,1)$};
	\node (8) at (6,-3) {$(0,2)$} ;
	\node (9) at (3, -3) {$(0,1)$};

	\path (1) edge node[ above] { $b$} (2);
	\path (1) edge node[above, pos=0.7] {$a$} (3); 
	\path (2) edge node [above] {$d$} (4);
	\path (2) edge node[ above, pos=0.7] {$c$} (5);
	\path (4) edge node[above] {$f$} (6);
	\path (4) edge node[ above, pos=0.7] {$e$} (7); 
	\path (6) edge node[above] {$h$} (8);
	\path (6) edge node[above, pos=0.7] {$g$} (9);
\end{tikzpicture}
\caption{Game $\Gamma_{E}$.}
\end{figure}
\end{example}

In the context of EFGs, the concept  of    
\emph{Nash Equilibria}  remains as given in \Cref{def:NE}. In addition to Nash Equilibria, another important concept for EFGs is the \emph{Subgame Perfect Equilibrium}, characterizing the strategies played in practice by rational parties. To this end, we first introduce subgames of EFGs.

\begin{definition}[Subgame of  EFG]
The \emph{subgame} of an EFG $\Gamma=(N,\mH,P,u)$ that follows history $h\in \mH$ is the  EFG $\Gamma(h)=(N,\mH_{|h}, P_{|h}, u_{|h})$ s.t.\   for every sequence $h' \in \mH_{|h}$ we have $(h,h') \in \mH$,  $P_{|h}(h'):= P(h,h')$ and $u_{|h}(h')=u(h,h')$.
\end{definition}

\begin{definition}[Subgame Perfect Equilibrium]
A \emph{subgame perfect equilibrium} is a joint strategy $\sigma \in \mS$, s.t.\ $\sigma_{|h}$ is a Nash Equilibrium of the subgame $\Gamma(h)$, for every $h \in \mH$.
\end{definition}
}

 \subsection{Game-Theoretic Security Properties of Off-Chain Protocols}\label{sec:prelin:NFGSec}
 We now present existing game-theoretic concepts~\cite{CITE,GameTheoryBook} implying security properties of off-chain protocols. 
 In Section~\ref{sec:theory}, we extend these concepts towards  \NEWR{another type of games, called Extensive Form Games}, enabling our security analysis in  Section~\ref{sec:models}.  We focus on two security properties ensuring that \ref{P1} honest players do not suffer damage, and \ref{P2} subgroups of rational players do not deviate from a respective strategy. A protocol is compliant to these properties, if the strategy implementing the intended behavior satisfies them; \NEWR{we call such a  strategy an \emph{honest strategy}}. 
 \vspace{0.2cm}
 
\begin{enumerate}[label={(P\arabic*})]
    \item \label{P1} {{\it  No Honest Loss.}} As the utility function of a game is supposed to display the monetary and intrinsic value of a certain joint strategy, property~\ref{P1} is expressed using \emph{weak immune strategies} defined next. 
\end{enumerate}

\begin{definition}[Weak Immunity] \label{def:wi}
	\NEWR{A joint strategy $\sigma \in \mS$ in an NFG $\Gamma$ is called \emph{weak immune}, if every player $p$ that follows $\sigma$ gets utility at least 0, regardless of how the other players behave:}
	\begin{equation}  \forall p \in N\;\;\forall \sigma'\in \mS:\quad u_p(\; \sigma'[\sigma_p/\sigma'_p]\;)\geq 0 \;. \end{equation}
\end{definition}

\NEWR{\begin{example}
In the game  $\Gamma_C$ of \Cref{tbl:oldclosing}, the only weakly immune strategy is $(U,U)$. This is the case, because as long as $A$ chooses $U$, player $B$ can take any strategy and $A$ will never get negative utility (similarly, vice-versa). 
\end{example}}

\vspace{0.2cm}

\begin{enumerate}[label={(P\arabic*})]
\setcounter{enumi}{1}
    \item \label{P2}  {{\it No Deviation.}}
 Even though the concept of Nash Equilibria seems to be a good candidate to ensure \ref{P2} at first glance, they have two crucial shortcomings. First, a Nash Equilibrium only ensures that a single player cannot profit from deviating, but does not imply that two or more players cannot do so. Second, there might be Nash Equilibria, which are weakly dominated by another strategy for a specific player. Such Nash Equilibria will therefore not be played by rational parties and hence should not be considered to satisfy \ref{P2}.  
\end{enumerate}

The solution proposed for NFGs in~\cite{CITE} is to consider strategies $\sigma$ compliant to \ref{P2}, if they are both \emph{strongly resilient} (fixing the former shortcoming) and \emph{practical} (fixing the latter) as defined subsequently.\vspace{0.2cm}

Strong resilience extends Nash Equilibria by considering deviations of multiple players.

\begin{definition}[Strong Resilience -- \SR]\label{def:sr}
	A joint strategy $\sigma \in \mS$ in an \NEWR{NFG} $\Gamma$ is  \emph{strongly resilient (\SR)}  if no proper subgroup of players $S:=\{s_1,...,s_j\}$ has an incentive in deviating: 
\begin{equation}
\begin{aligned}
	{\color{teal} \forall S \subset N}\quad & \forall \sigma'_{s_i} \in \mS_{s_i} \quad {\color{teal} \forall p \in S}:\\
	&u_p(\sigma) \geq u_p(\;\sigma[\sigma'_{s_1}/\sigma_{s_1},...,\sigma'_{s_j}/\sigma_{s_j}]\;) \;.
	\end{aligned}
	\end{equation}
	\end{definition} 
	
\NEWR{	
We note that in  games with two players (i.e. two-player games), 	 strong resilience and Nash Equilibrium  are identical. As such, in $\Gamma_C$ from \Cref{tbl:oldclosing}, the joint  strategies $(C,C)$, $(U,\mathfrak{I})$ and $(\mathfrak{I},U)$ of Example~\ref{ex:oldclosing} are also strongly resilient.}

To define practicality \NEWR{of a strategy}, we first introduce the concept of  \emph{iterated deletion of weakly dominated strategies (IDWDS)}.

\begin{definition}[Iterated Deletion of Weakly Dominated Strategies -- IDWDS] \label{def:idwds}
The \emph{iterated deletion of weakly dominated strategies} (IDWDS)  of an \NEWR{NFG} $\Gamma$ is defined as iteratively rewriting $\Gamma$ by omitting \emph{all} weakly dominated strategies of all players. This is repeated until no strategy is weakly dominated any more. The resulting game $\Gamma'$ is thus a subgame of $\Gamma$.
\end{definition}

Note that, when IDWDS is applied to a game $\Gamma$, then every Nash Equilibrium of the resulting game $\Gamma'$ is also a Nash Equilibrium of $\Gamma$. \NEWR{Since all weakly dominated strategies of every player are removed at each step, the generated game is unique. Details and proofs can be found in~\cite{GameTheoryBook}}. \vspace{0.1cm}

\NEWR{We now define practical strategies}, in order to  ensure that no single strategy is weakly dominated at any iteration.

\begin{definition}[Practicality]\label{def:nfg:practical}
	A strategy is \emph{practical} if it is a Nash Equilibrium of the \NEWR{NFG} $\Gamma'$ after iterated deletion of weakly dominated strategies.
\end{definition}

\NEWR{
\begin{example}
Let us consider $\Gamma_C$ from \Cref{tbl:oldclosing}.  We know from \Cref{ex:oldclosing} that  only $\mathfrak{I}$ is weakly dominated for both $A$ and $B$.  Therefore, according to \Cref{def:idwds}, strategy $\mathfrak{I}$ has to be removed from both player's strategy set. This yields  the game $\Gamma'_C$ as listed \Cref{tbl:idwds}.

\begin{table}	
	\centering	
	\caption{NFG $\Gamma'_C$ obtained from IDWDS over \Cref{tbl:oldclosing}.}

	\begin{tabular}{|c||r|c|c|}
	    \hline 
	\diagbox[width=5em,height=2em]{$~~~~A$}{$~~B$}		&{$U$} & \textcolor{black}{$C$}  \\
		\hline
		\hline
		$U$  & $(1/2,{1/2})$ & $(0,{1})$  \\
		\hline
		\textcolor{black}{$C$} & $(1,{0})$ & $(1,{1})$  \\
		\hline
	\end{tabular}
	\label{tbl:idwds}
\end{table}
Note that there are no weakly dominated strategies in $\Gamma'_C$ . Thus, any Nash Equilibrium  of $\Gamma'_C$ is also practical strategy of $\Gamma_C$. By comparing  utilities, we derive that the only Nash Equilibrium  of $\Gamma'_C$ is the joint strategy $(C,C)$.
\end{example}}

An alternative approach for expressing \ref{P2} is by requiring a strategy $\sigma$ to be both a \emph{strong Nash Equilibrium} (a property similar to \SR{}) and practical, instead of \SR{} and practical.

	\begin{definition}[Strong Nash Equilibrium -- \sNE] \label{def:sne}
	A joint strategy $\sigma$ is a \emph{strong Nash Equilibrium (\sNE)} if for every group of deviating players $S:=\{s_1,...,s_j\}$ and all possible deviations $\sigma'_{s_i} \in \mS_{s_i}$, $i\in \{1,...,j\}$ at least one player $p \in S$ has no incentive to participate, that is
\begin{equation}
\begin{aligned}
	 \forall {\color{teal} S \subseteq N,\, S \neq \emptyset} \quad & \forall \sigma'_{s_i} \in \mS_{s_i} \quad {\color{teal} \exists p \in S}:\\ 
	 &u_p(\sigma) \geq u_p(\;\sigma[\sigma'_{s_1}/\sigma_{s_1},...,\sigma'_{s_j}/\sigma_{s_j}]). 
	\end{aligned}
	\end{equation}
\end{definition}	

\NEWR{
\begin{example}
In $\Gamma_C$ from \Cref{tbl:oldclosing}, all NE are also sNE. For the joint strategy $(C,C)$, this is easy to see. However, it is also the case for $(U,\mathfrak{I})$ and $(\mathfrak{I},U)$, since any deviation yields a utility of at most 1. Thus, at least one player's utility does not increase by deviating from $(U,\mathfrak{I})$, $(\mathfrak{I},U)$ respectively.
\end{example}
}

 A detailed comparison of the various concepts ensuring \ref{P2} including their strengths and weaknesses, is given in \Cref{sec:theory}.

	\section{\REPLACER{EFG Advancements for}{EFG-based Modeling of}  Off-Chain Protocols}\label{sec:theory}

\REPLACER{We now introduce novel EFG concepts 
by extending the NFG setting of  Section~\ref{sec:prelin:NFGSec}.  Such an extension is needed to overcome the restriction of NFGs in modeling only simultaneous actions, while ensuring practicality in the sequential setting of EFGs. 
Doing so, we introduce \emph{extended strategies in EFGs}, which allow us to capture deviations such as dishonest closing attempts in \Cref{sec:models}. Such closing attempts cannot be modeled in the NFG approach of~\cite{CITE}.}{So far we considered games in which each party takes only one action. We now extend our definitions to handle adaptive strategies, i.e., games in which parties take several actions and choose at each step which action to take based on the actions previously chosen by  other parties.  As we will see, this is necessary for faithfully modeling off-chain protocols and overcoming the limitations of previous work~\cite{CITE}. For that, we  overview the concept of extensive form games (EFGs) in \Cref{sec:efg-new}. We  show how to lift  NFG-based  security definitions  to EFGs in \Cref{sec:EFG:extended}. Finally, we show that these definitions do not yet suffice to yield an accurate security model of off-chain protocols, and  introduce a refined security definition based on the concept of collusion resilience in \Cref{sec:secstrat}.}

\NEWR{\subsection{Extensive Form Games (EFG)}\label{sec:efg-new}}
To formalize strategies where players make multiple choices one after the other, we \REPLACER{next consider}{advocate the usage of} Extensive Form Games (EFGs)~\cite{GameTheoryBook}, which extend NFGs as follows. 

\begin{definition}[Extensive Form Game -- EFG]
An \emph{Extensive Form Game (EFG)} is a tuple $\Gamma=(N,\mH,P,u)$, where $N$ and $u$  are as in NFGs. The set  $\mH$ captures \emph{game  histories}, $\T \subseteq \mH$ is the set of \emph{terminal histories}, and $P$ denotes the \emph{next player function}, satisfying the following properties.
\begin{itemize}
\item The set $\mH$ of histories is a set of sequences of actions with  \begin{enumerate} 
\item $\emptyset \in \mH$; 
\item if the action sequence $(a_k)_{k=1}^K \in \mH$ and $L<K$, then also $(a_k)_{k=1}^L \in \mH$; 
\item a history is terminal $(a_k)_{k=1}^K \in \T$, if there is no action $a_{K+1}$ with  $(a_k)_{k=1}^{K+1} \in \mH$.
\end{enumerate} 
\item The next player function $P$ 
\begin{enumerate}
    \item assigns the next player $p \in N$ to every non-terminal history  $(a_k)_{k=1}^K \in \mH\setminus \T$, that is  $P((a_k)_{k=1}^K) = p$; 
\item after a non-terminal history $h= (a_k)_{k=1}^K \in \mH$, it is player $P(h)$'s turn to choose an action from the action set $A(h)=\{a: (h,a) \in \mH\}$.
\end{enumerate}
\end{itemize}
A \emph{strategy} of  player $p$ is a function $\sigma_p$ mapping every $h \in \mH$ with $P(h)=p$ to an action from $A(h)$. \NEWR{Formally, \[\sigma_p:\{h \in \mH: P(h)=p\}\to \{a: (h,a)\in \mH, \forall h \in \mH\}\;,\] such that $\sigma_p(h)\in A(h)$. The set of all strategies of a player $p$ is $\mS_p$, and the set of all joint strategies is $\mS=\vartimes_{p\in N}\mS_p$.} \REMOVER{The utilities of all joint strategies with terminal history $h$ are the same.}
\end{definition}

Note that the set of terminal histories $\T$ is \REPLACER{well-defined}{uniquely determined} by $\mH$ and therefore does not  explicitly occur in the tuple $\Gamma$. \NEWR{Since histories $h$ are just sequences of actions $h=(a_k)_{k=1}^K=(a_1,...,a_K)$, we denote histories by the variable $h$, the abstract sequence  $(a_k)_{k=1}^K$, or the explicit sequence  $(a_1,...,a_K)$, depending on the context in which they are used.} We note that 
EFGs can  conveniently be represented as trees, as described below. 
\begin{definition}[EFG as Tree]
Considering an EFG $\Gamma=(N,\mH,P,u)$,  the following tree $G=(V,E)$ represents $\Gamma$.
\begin{itemize}
    \item For every history $h \in \mH$, there exists exactly one node $v_h \in V$. 
    This is labeled by $P(h)$, the next player, if $h$ is not terminal ($h \notin \T$), or by  $u(\sigma)$, the joint utility of playing a game with history $h$, if h is terminal ($h \in \T$) and the joint strategy $\sigma$ yields history $h$. 
    \item Two nodes $v_h$, $v_{h'}\in \mH$ are connected via an oriented edge $(v_h,v_{h'}) \in E$ iff $h'=(h,a)$. This edge is labeled $a$.
\end{itemize}
\REMOVER{The graph $G$ is a tree with root $v_{\emptyset}$ representing the empty history, which is labeled with the first player, and with one leaf $v_t$ per terminal history $t \in \T$.}
\end{definition}

\REMOVER{Note that (i) a history $h$ is a path starting from the root, and it is terminal if and only if it ends in a leaf; (ii) each internal tree node models the turn of an EFG player $p$ after a non-terminal history $h$. The outgoing edges represent the action set $A(h)$. Further, (iii) each leaf represents the joint utility of every joint strategy, whose  history (path) leads to it; (iv) an EFG strategy $\sigma$ is not a history (path) but a set of edges that contains precisely one outgoing edge per internal node.}
\NEWR{Let us illustrate EFGs and their tree-based representation through the following example. }

\begin{example}\label{ex:efg}
The game tree in \Cref{tbl:efg} results from the extensive form game $\Gamma_E=(N,\mH, P,u)$ with the two players $N=\{A,B\}$, where the set of histories is $\mH=\{\emptyset,(a), (b), (b,c), (b,d), (b,d,e), (b,d,f), (b,d,f,g),$  $(b,d,f,i)\}$. The next player function $P$ assigns player $A$ after histories $\emptyset$ and $(b,d)$, and player $B$ after $(b)$ and $(b,d,f)$. Finally, the utility function $u$ assigns  joint utility $(2,2)$ to strategies that yield history $(a)$, utility $(3,1)$ for strategies with history $(b,c)$, utility $(1,1)$ for strategies with history $(b,d,e)$, and $(0,2)$ for strategies resulting in $(b,d,f,i)$.
A strategy \NEWR{$\sigma=(\sigma_A,\sigma_B)$ in $\Gamma_E$ is for example: $A$ chooses $a$ after history $\emptyset$: $\sigma_A(\emptyset)=a$; and $f$ after $(b,d)$: $\sigma_A((b,d))=f$; $B$ takes $c$ after $(b)$: $\sigma_B((b))=c$, and $g$ after $(b,d,f)$: $\sigma_B((b,d,f))=g$. Following this strategy until we read a leaf yields history $(a)$. A different strategy $\sigma'=(\sigma'_A,\sigma'_B)$, which also yields history $(a)$, is for example $\sigma'_A(\emptyset)=a$,  $\sigma'_A((b,d))=e$ and  $\sigma'_B=\sigma_B$.}
 \begin{figure}
	\centering
\begin{tikzpicture}[scale= 0.9,->,>=stealth',auto,node distance=3cm, el/.style = {inner sep=4pt, align=center, sloped}]
	\node (1) at (0,0) {$A$};
	\node (2) at (1.5,-0.75) {$B$} ; 
	\node (3) at (-1.5,-0.75) {$(2,2)$};
	\node (4) at (3,-1.5) {$A$} ;
	\node (5) at (0, -1.5) {$(3,1)$};
	\node (6) at (4.5,-2.25) {$B$} ; 
	\node (7) at (1.5,-2.25) {$(1,1)$};
	\node (8) at (6,-3) {$(0,2)$} ;
	\node (9) at (3, -3) {$(0,1)$};

	\path (1) edge node[ above] { $b$} (2);
	\path (1) edge node[above, pos=0.7] {$a$} (3); 
	\path (2) edge node [above] {$d$} (4);
	\path (2) edge node[ above, pos=0.7] {$c$} (5);
	\path (4) edge node[above] {$f$} (6);
	\path (4) edge node[ above, pos=0.7] {$e$} (7); 
	\path (6) edge node[above] {$i$} (8);
	\path (6) edge node[above, pos=0.7] {$g$} (9);
\end{tikzpicture}
\caption{An EFG  $\Gamma_{E}$.}
\label{tbl:efg}
\end{figure}
\end{example} 

\NEWR{As depicted in the tree-based representation of \Cref{tbl:efg}, we note that the utility of joint strategies in an EFG is uniquely determined by their associated history (i.e., path). }  
In the context of EFGs, the concept  of    
\emph{Nash Equilibria}  remains as given in \Cref{def:NE}. In addition to Nash Equilibria, another useful concept for EFGs is the \emph{Subgame Perfect Equilibrium}, which we will use to characterize the strategies played in practice by rational parties. To this end, we first introduce the notion of subgames of EFGs\NEWR{. A subgame of an EFGs  can be seen as a  subtree determined by a certain history (i.e., whose root note is the last history node), and is formalized below}.

\begin{definition}[Subgame of  EFG]
The \emph{subgame} of an EFG $\Gamma=(N,\mH,P,u)$ associated to  history $h\in \mH$ is the  EFG $\Gamma(h)=(N,\mH_{|h}, P_{|h}, u_{|h})$ defined as follows: $\mH_{|h}:=\{h'~|~(h,h') \in \mH\}$,  $P_{|h}(h'):= P(h,h')$, and $u_{|h}(h'):=u(h,h')$.
\end{definition}

\NEWR{
\begin{example}
Consider the EFG $\Gamma_{E}$ from  \Cref{tbl:efg} . 
The subgame of $\Gamma_{E}$  associated to history $(b,d)$ is the subtree rooted in $A$. 
\end{example}
}

\NEWR{By adjusting the concept of Nash Equilibrium to subgames, we derive the following property of joint strategies.}

\begin{definition}[Subgame Perfect Equilibrium]
A \emph{subgame perfect equilibrium} is a joint strategy $\sigma = (\sigma_1,...,\sigma_{n}) \in \mS$, s.t. $\sigma_{|h}= (\sigma_{1|h},...,\sigma_{n|h})$ is a Nash Equilibrium of the subgame $\Gamma(h)$, for every $h \in \mH$. \NEWR{The strategies $\sigma_{i|h}$ are functions that map every $h'\in \mH_{|h}$ with $P_{|h}(h')=i$ to an action from $A_{|h}(h')$.}
\end{definition}

\subsection{EFG Extensions \NEWR{for Security Properties}}\label{sec:EFG:extended}
While  EFGs enable us to incorporate choices made at different times yielding different options for the next player, they come with the following limitation.
\NEWR{The intended (i.e., honest) behaviors in off-chain protocols only specify a terminal history (i.e., a path from root to leaf), rather than a strategy. For instance, an honest history may specify to close the channel collaboratively, but it does not capture a  player's behavior once a player deviated.}
To address this limitation, we introduce the following notion of an extended strategy in EFGs.
\begin{definition}[Extended Strategy]\label{def:extstrat}
	Let $\beta$ be a terminal history in an EFG $\Gamma$. Then, all strategies $\sigma_\beta$ that result in history $\beta$ are \emph{extended strategies} of $\beta$. 
\end{definition} 

\NEWR{
\begin{example}
Recall \Cref{tbl:efg}. In \Cref{ex:efg}, we consider the terminal history $(a)$ and provide two extended strategies of $(a)$, they are $\sigma$ and $\sigma'$. A strategy, which is not an extended strategy of $(a)$ is for instance $\sigma''=(\sigma''_A,\sigma''_B)$, where $\sigma_A''(\emptyset)=b$,   $\sigma''_A((b,d))=e$ and  $\sigma''_B=\sigma_B$. This is the case because by following the choices of $A$ and $B$ in $\sigma''$, we end up in $(b,c)$. 
\end{example}
}

\REMOVER{Recall that the game-theoretic properties in~\Cref{sec:prelim} are defined on strategies, but not on terminal histories.
In our work, however, we are interested in analyzing whether a protocol together with its honest behavior $\beta$ satisfies  the security properties \ref{P1} and \ref{P2}. We therefore use extended strategies $\sigma_\beta$ to formalize properties of  $\beta$. } 


 While  EFGs  can in principle be translated to NFGs, as explained in~\cite{GameTheoryBook},  
 analyzing the security properties \ref{P1}-\ref{P2}  over the translated NFGs may yield unexpected results. We shortly exemplify this point in \Cref{ex:pract}, but similar issues occur also in larger games. 
 We thus lift NFG-based definitions to EFGs,  enabling the  analysis of \ref{P1} and \ref{P2}. 
 Since EFGs have a utility function just as NFGs do, which assigns  values after the game, the NFG concepts of weak immunity, strong resilience and \sNE{} remain the same for EFGs.
 
 \NEWR{
 \begin{definition}[EFG Properties]
 A joint strategy $\sigma\in\mH$ of an EFG $\Gamma$ is called weak immune, strongly resilient, or a strong Nash Equilibrium, if it satisfies the formulae of \Cref{def:wi}, \Cref{def:sr} or \Cref{def:sne} respectively.
 \end{definition}
 }
 
 Practicality in NFGs, however, relies  on IDWDS, which fails to incorporate the sequential nature of EFGs,  and hence must  be adjusted for EFGs. This is because NFG actions  happen simultaneously, while EFG players  choose their actions sequentially. We first  present an example to showcase that  applying the NFG definition of practicality to an EFG, by using its translation to an NFG, leads to overlooking rational strategies.

\begin{example} \label{ex:pract}
 Let us consider the EFG $\Gamma_E$ from \Cref{tbl:efg}, with two players $A$ and $B$. The compact translation of $\Gamma_E$ to an NFG  $\Gamma_N$ is given in \Cref{tbl:nfg}. \NEWR{Histories of \Cref{tbl:efg}, where players choose twice, such as $(b,d,f)$, are translated to \Cref{tbl:nfg} as the joint strategy $(b;f\;,d)$. Hence, the NFG strategy $b;f$ of player $A$ means choosing action $b$ first, and, if $A$ gets to choose again, $A$ takes $f$.} Player $A$'s strategies are displayed in the rows, whereas player $B$'s are shown in the columns of \Cref{tbl:nfg}. Strategy $d;g$ for example denotes choosing $d$ in the first turn and $g$ in the second turn, unless the game ends before. For readability, strategies with identical utilities in any case are merged together, e.g., having only $a$ instead of both $a;e$ and $a;f$. 
 
According to definition of practicality for NFGs (see \Cref{def:nfg:practical}), the only practical strategy in $\Gamma_N$ is $(a,\;d;i)$, which results in a utility of $(2,2)$. This is because for $A$ strategy $b;e$ weakly dominates $b;f$ and for $B$ strategy $d;i$ weakly dominates both $c$ and $d;g$. After deleting those (in blue), the red strategy $b;e$ of $A$ becomes weakly dominated by $a$. Thus, after removing $b;e$ only the joint strategy $(a,\;d;i)$ remains and is therefore a Nash Equilibrium of the resulting game. 

\begin{table}	
	\centering	
	\caption{Compact View of $\Gamma_E$, Translated to an NFG $\Gamma_N$.}
	
	    	\begin{tabular}{|c||r|c|c|c|}
	    \hline 
	\diagbox[width=5em,height=2em]{$~~~~A$}{$~~B$}	
		&\textcolor{cyan}{$c$} & \textcolor{cyan}{$d;g$} & $d;i$ \\
		\hline\hline
		$a$  & $(2,{2})$ & $(2,{2})$ & $(2,2)$ \\
		\hline
		\textcolor{red}{$b;e$} & $(3,{1})$ & $(1,{1})$ & $(1,{1})$ \\
		\hline
		\textcolor{cyan}{$b;f$} & $({3},{1})$ & $({0},{1})$ & $({0},2)$ \\
		\hline
	\end{tabular}
	\label{tbl:nfg}
\end{table}
However, in the EFG $\Gamma_E$ the comparison of strategies has a certain order, as not all choices are made simultaneously. Thus, when it comes to $B$ choosing between  option $c$ and $d$, choosing $c$ is also a rational action because in any case $B$ gets utility 1. This is the case, since the subgame following after $d$, will end in the subgame perfect and practical $(1,1)$, if played by rational players. Following this argumentation, we claim that $(b;e,c)$, yielding history $(b,c)$ should also be considered rational and thus practical. 
\end{example}

Example~\ref{ex:pract} demonstrates that it is advisable to adapt the \REPLACER{introduced}{NFG} concept of practicality \NEWR{for EFGs}, and that a n\"aive application can be problematic since information may be lost during the transformation from EFG to NFG~\cite{GameTheoryBook}.  We therefore propose to use subgame perfect equilibria for comparing EFG strategies, and define \emph{practicality for EFGs} as follows.

\begin{definition}[Practicality for EFG]\label{def:practicalEFG}
		A strategy of an EFG $\Gamma$ is \emph{practical} if it is a subgame perfect equilibrium of $\Gamma$. 
\end{definition}  

\subsection{Security Strategies for Off-Chain Protocols}\label{sec:secstrat}
We now leverage the previously introduced EFG-based definitions  (\Cref{sec:EFG:extended}) to faithfully model the security  of off-chain protocols. \NEWR{ In particular, we propose the novel concept of \emph{collusion resilience} for addressing \ref{P2}, and compare it to existing formalizations of property \ref{P2}.}

In \cite{CITE},  strong resilience and practicality were used to model the no deviation property of \ref{P2}:  We identify unwanted properties of strong resilience and we thus investigate variations of it. Specifically, we show  that strong Nash Equilibria do not imply strong resilience nor vice-versa (\Cref{lemma:impl}), and therefore define the \emph{collusion resilience} property of a joint strategy. Intuitively, collusion resilience considers the sum of the utilities of the deviating parties, since rational players may collude or be controlled by the same entity.

\begin{definition}[Collusion Resilience -- \srp] \label{def:CR}
	A joint strategy $\sigma \in \mS$ in an EFG/NFG $\Gamma$ is called \emph{collusion resilient (\srp)} if no strict subgroup of players $S:=\{s_1,...,s_j\}$ has a joint incentive in deviating from $\sigma$. That is, 
	\begin{equation}
\begin{aligned}
	\forall {\color{teal}S \subset N} \quad &  \forall \sigma'_{s_i} \in \mS_{s_i}: \\
	&{\color{teal} \sum_{p \in S}} u_p(\sigma) \geq {\color{teal} \sum_{p \in S}} u_p(\;\sigma[\sigma'_{s_1}/\sigma_{s_1},...,\sigma'_{s_j}/\sigma_{s_j}]\;).
	\end{aligned}
	\end{equation}
\end{definition}

In addition, we also consider a slight adaption of strong resilience, \srs{}, where the deviation of the entire set of players $N$ is also allowed, as it is for \sNE.

\begin{definition}[Strong Subset Resilience -- \srs]
	A joint strategy  $\sigma \in \mS$ is called \emph{strongly subset resilient (\srs)}, if no player of any subgroup $S \subseteq N$, $S:=\{s_1,...,s_j\}$ has an incentive to deviate from $\sigma$: 
	\begin{equation}
\begin{aligned}
\forall {\color{teal}S \subseteq N}\quad &
\forall \sigma'_{s_i} \in \mS_{s_i}\quad
{\color{teal}  \forall p \in S}: \\
&u_p(\sigma) \geq u_p(\;\sigma[\sigma'_{s_1}/\sigma_{s_1},...,\sigma'_{s_j}/\sigma_{s_j}]\;) \;.	
\end{aligned}
\end{equation}
\end{definition}

We now  formalize how the resilience properties relate to each other\NEWR{, which motivates our definition of  \ref{P2}.} \REMOVER{It is not hard to prove that the following holds.}

\begin{lemma}[Resilience Properties] \label{lemma:impl}
	 Let $\sigma \in \mS$ be a joint strategy. The following and only the following implications hold.
	 
\begin{minipage}{0.34\textwidth}
	\begin{enumerate}
		\item $\sigma \text{ is \srs}\; \Rightarrow \; \sigma \text{ is \SR, \srp{}, \sNE}$.
		\item $\sigma \text{ is \SR} \; \Rightarrow \; \sigma \text{ is \srp}$.
	\end{enumerate}
	\end{minipage}
	\begin{minipage}{0.12\textwidth}
\begin{tikzpicture}[scale=0.8, baseline, ->]
		\node (1) at (0,0) {\srs};
		\node (2) at (-1.5,0) {\SR} ; 
		\node (3) at (0,-1.5) {\sNE};
		\node (4) at (-1.5,-1.5) {\srp} ;

		\path (1) edge  (2);
		\path (1) edge (3); 
		\path (1) edge  (4);
		\path (2) edge (4);
	\end{tikzpicture}
\end{minipage}
\end{lemma}

The next example further motivates why we decided to formalize \ref{P2} in terms of \emph{collusion resilience}. 

\begin{example} \label{ex:cr} Consider the games $\Gamma_1$ and $\Gamma_2$, respectively defined in Tables~\ref{tbl:G1}-\ref{tbl:G2}. The games $\Gamma_1$ and $\Gamma_2$  show that there exist cases where both  strong resilience and strong Nash Equilibria fail to correctly state whether rational players will deviate,  while collusion resilience does not. 

Let us study $\Gamma_1$ first. There are three players $P_1$ on the left, $P_2$ in the ``3rd dimension'' who only has one possible strategy, and $P_3$ at the top. Let us consider the joint strategy $\sigma = (H_1,H_2, H_3)$. Since $P_2$ does not have another choice, $P_2$ can never deviate. Player $P_1$ deviating alone yields the same utility as $\sigma$ and is thus irrelevant. The same holds for $P_3$. The only deviation that makes a difference, is if $P_1$ and $P_3$ change strategy together to $(D_1,H_2, D_3)$. By doing so, $P_1$ profits and receives 5 instead of 1, but $P_3$ looses by getting $-2$ instead of 1. Thus, $P_3$ does not have an incentive to do so, unless the two players collude for their mutual benefit and share their payoffs. This way they receive 1.5 each instead of 1 each, which poses a serious thread to $\sigma$ and should thus not be considered satisfying \ref{P2}. However, $(H_1, H_2, H_3)$ is \sNE{}, since $P_3$ has no incentive in deviating with $P_1$, if their utilities are not shared, but it is not \srp{}, since 
\begin{align}
   2&= u_{P_1}(H_1, H_2, H_3)+u_{P_3}(H_1, H_2, H_3)  \\ 
   &< u_{P_1}(D_1, H_2, D_3)+u_{P_3}(D_1, H_2, D_3)=3 \; . 
\end{align}

In the similar game $\Gamma_2$, on the contrary, $P_3$ has no incentive in deviating from $\sigma=(H_1,H_2,H_3)$ together with $P_1$, also  if their utilities are shared. Such a deviation yields 0.5 each, instead of 1 each in $\sigma$.
Hence, there is no incentive to change strategy for one or more players and therefore $(H_1, H_2, H_3)$ should be considered satisfying \ref{P2}.
    Nevertheless, according to \Cref{def:sr}, $(H_1, H_2, H_3)$ is not \SR, since at least one of the deviating parties $P_1$, $P_3$ profits from choosing $(D_1, H_2, D_3)$, although $P_3$ has no reason to play along. However, in $\Gamma_2$, $(H_1, H_2, H_3)$ is \srp{} as 
\begin{align}
       2&= u_{P_1}(H_1, H_2, H_3)+u_{P_3}(H_1, H_2, H_3)  \\
   &\geq u_{P_1}(D_1, H_2, D_3)+u_{P_3}(D_1, H_2, D_3)=1 \; .
\end{align}
\end{example}

\begin{table}
\centering
\begin{minipage}{0.47\linewidth}\centering
\caption{Three Player Game $\Gamma_1$.}
		\begin{tabular}{|r|c|c|}
		\hline
		{\tiny \rotatebox{45}{$H_2$}}& $H_3$ & $D_3$ \\
		\hline
		$H_1$ & ${\color{red}(1,1,1)}$ & $(1,1,1)$ \\
		\hline
		$D_1$ & $(1,1,1)$& $(\textcolor{cyan}{5},0,\textcolor{cyan}{-2})$\\
		\hline
	\end{tabular}
\label{tbl:G1}
\end{minipage}
\begin{minipage}{0.47\linewidth}\centering
\caption{Three Player Game $\Gamma_2$.}
	\begin{tabular}{|r|c|c|}
	\hline
	{\tiny \rotatebox{45}{$H_2$}} & $H_3$ & $D_3$ \\
	\hline
	$H_1$ & {\color{red}$(1,1,1)$} & $(1,1,1)$ \\
	\hline
	$D_1 $& $(1,1,1)$& $(\textcolor{cyan}{3},0,\textcolor{cyan}{-2})$\\
	\hline
\end{tabular}
\label{tbl:G2}
\end{minipage}
\end{table}

\begin{remark}[Formalizing (\ref{P1} and \ref{P2}]
Based on the resilience properties of \Cref{lemma:impl},  we say {\emph{(P2) is satisfied by a joint strategy $\sigma$}, if $\sigma$ is \srp{} and practical}.
In addition, a joint strategy \emph{$\sigma$ satisfies \ref{P1}, if $\sigma$ is weak immune}, as in \cite{CITE}.
\end{remark}

\NEWR{We conclude this section by defining secure game strategies/histories, as follows. }

\begin{definition}[Secure Strategy] 
A strategy $\sigma$ of an NFG/EFG  is \emph{secure} if it is weak immune, practical and \srp.
\end{definition}

\NEWR{When  discussing security in the setting of EFGs, we are interested mainly in assessing whether a  \emph{history} is  secure, as the protocol only defines an honest history instead of a full strategy. By applying  \Cref{def:extstrat}, we state the following security characterization.

\begin{definition}[Secure History] \label{def:secure}
A terminal history $\beta$ of an EFG is \emph{secure}  if there exist extended strategies $\sigma_1$, $\sigma_2$, and $\sigma_3$ of $\beta$, such that $\sigma_1$ is weak immune, $\sigma_2$ is practical and $\sigma_3$ is \srp. 
\end{definition}

We note that we do not have to find a secure extended strategy for the history to be secure, as aiming for one joint secure strategy in an EFG would be unnecessarily restrictive. Instead, our goal is to make sure that rational parties follow the honest history, no matter what their actual strategy is. In particular, an honest player follows the honest history by default, a rational player does so because of practicality and collusion resistance. Weak immunity further ensures that honest players as well as rational one cannot be damaged by Byzantine players while following the honest history. 
Hence, the strategy each player has in mind does not matter, since in a secure protocol weak immune, practical, and collusion resistant, strategies are overlapping along the honest history. This is the case because in \Cref{def:secure} we require $\sigma_1$, $\sigma_2$, and $\sigma_3$ to all yield the same history, namely $\beta$. We can therefore admit  that an  honest player has a weak immune strategy in mind, while a rational player has a practical one, as long as these overlap on the honest history. 
}

	\section{Closing Games \NEWR{of Off-Chain Protocols}} \label{sec:models}
We now define a new two-player EFG, called the \emph{Closing Game $G_c$}, 
in order to model closing phase properties of off-chain protocols, in particular of the Lightning Network. \NEWR{As explained in \Cref{sec:pcn}, to close a channel a party can unilaterally publish a channel state on-chain, which does not necessarily have to be the latest one. The one who closes, however, has to wait a certain amount of time until the money can be used. Meanwhile, the other party can steal all the money from the channel in case the state published on-chain is not the latest one: this ensures that rational players   close their channel only with the latest  state.  Alternatively, the parties can collaboratively sign a new transaction to split the money. In this case no one has to wait.}

Our closing game overcomes the limitations of previous work~\cite{CITE} in representing dishonest closing attempts, \NEWR{by} modeling  how closing can be achieved after a failed collaborative closing attempt and by also considering the additional fee $f$ to be paid in a revocation transaction.

To the best of our knowledge, our closing game $G_c$ is the most accurate model for the security analysis of off-chain protocols, notably of the Lightning Network.
 In our model of the closing phase we make the following assumptions for a channel between $A$ and $B$ \NEWR{at the moment where the closing phase is initiated.}
\begin{itemize}
	\item The fair split of the channel's funds is $a \to A$, $b \to B$ and $a>0$, $b>0$.
	\item The benefit of closing the channel is $\alpha$. Closing a channel yields a benefit, since it unlocks assets.
	\item The opportunity cost of having to wait for one's funds upon closing is $\epsilon$.
	\item When both players agree to update the channel we assume a fair deal in the background which yields a profit of $\rho$ for both parties.
	\item Publishing a revocation transaction on-chain costs a fee $f>0$.
	\end{itemize}
 
	Further, to properly model utilities in the closing game $G_c$, we define the following total order, which is crucial for \NEWR{analyzing} security properties of $G_c$. \NEWR{For capturing total order properties in the setting of $G_c$, we  extend the set  $\R$ of real numbers by the infinitesimal numbers $\alpha$, $\epsilon$ and $\rho$.}

\begin{definition}[Utility Order\label{def:utility:order}]
We consider the total order $(\U,\preccurlyeq)$, where $\U$ is the group resulting from closing $\R \;\dot{\cup}\; \{\alpha,\epsilon,\rho \}$ under addition. The total ordering $\preccurlyeq$ is uniquely defined by the following conditions.
\begin{enumerate}
\item On $\R$, the relation $\preccurlyeq$ is the usual less than or equal relation  $\preccurlyeq|_\R\: := \:\leq$.
\item The values $\alpha$, $\epsilon$ and $\rho$ are greater than 0, 
\begin{align}
\forall \xi \in \{\alpha,\epsilon,\rho\}: \; -\xi\prec 0 \prec \xi \;. \end{align}
\item The values $\alpha$, $\epsilon$ and $\rho$ are closer to 0 than any real number, 
\begin{align} \forall x \in \R, \xi \in \{\alpha,\epsilon,\rho\}, x>0: \; \xi \prec x,\; -x \prec -\xi . \end{align}
\item Additionally, $\alpha$, $\epsilon$ and $\rho$ have the order $\rho \prec \epsilon \prec \alpha$.
\end{enumerate}
\end{definition}

In general, unlocking funds gives additional financial freedom even if there is some processing delay; therefore, we choose \REPLACER{$\alpha-\epsilon \succ 0$}{$\epsilon\prec\alpha$ in \Cref{def:utility:order}}. Additionally, once the parties initiate the closing phase, it is reasonable to assume that 
no potential update significantly benefits both parties. 
In contrast, both parties are interested in avoiding the opportunity cost, i.e., the cost of having to wait for their funds upon closing,
therefore, we set $\rho \prec \epsilon$ \NEWR{in \Cref{def:utility:order}}.

\NEWR{
\begin{remark}
While the ordering conditions of \Cref{def:utility:order} may seem to be  restrictive, lifting them comes with the burden of  considering a high number of possible variable orderings.  In particular, one would need to consider
 $ (\text{number of variables})!$ orderings, which would highly complicate the  formal analysis task. Approximating or clustering the number of orderings, while weakening conditions in \Cref{def:utility:order}, is an interesting venue for future work. 
\end{remark}
}

Based on the utility ordering of \Cref{def:utility:order}, we introduce our \emph{Closing Game for Player $A$} below.

\begin{table}[t]
    \centering
        \caption{Possible Actions in $G_c(A)$.} 
    \begin{tabularx}{\linewidth}{|@{\hspace{.5em}}>{\bfseries}l@{\hspace{.5em}}X@{\hspace{.5em}}|}
    \hline
        $H$ & Close unilaterally and \emph{honestly} without reacting to a previous move, such as a collaborative closing attempt. \\
    \hline
        $D$ & Close unilaterally but \emph{dishonestly} (without reacting to a previous move) with a profit of $d_A \in (0,b]$ in $A$'s case, $d_B \in (0,a]$ in $B$'s case.\\
    \hline
        $C_h$ & Try to close \emph{collaboratively} and \emph{honestly}, that is proposing a fair split.\\
    \hline
        $C_c$ & Try to close \emph{collaboratively} but by \emph{cheating} the other party by $c \in (0,b]$, that means proposing an unfair split.\\
    \hline
        $S$ & \emph{Signing} the collaborative closing attempt of the other player.\\
    \hline 
     $\mathfrak{I}$ & \emph{Ignore} the previous action and do nothing.\\
    \hline
    $P$ & \emph{Prove} other party tried to close dishonestly. That means stating a revocation transaction. We assume \NEWR{its publication requires a fee of $f>0$ and that} the attempt to do so is always successful, that is that the miners behave honestly.\\
    \hline
    $U^+$ & Propose an \emph{update} of the channel where player $A$'s balance is \emph{increased} by $p_A \in (0,b]$.\\
    \hline
    $U^-$ & Propose an \emph{update} where player $A$'s balance is \emph{decreased} by $p_B \in (0,a]$.\\
    \hline 
    $\mathfrak{A}$ & Agree to a proposed update.\\
    \hline
    \end{tabularx}
    \label{tbl:actions}
\end{table}

\begin{definition}[Closing Game $G_c(A)$ of Player $A$] \label{def:closing}
	The \emph{Closing Game} $G_c(A)=(N, \mH, P, u)$ is an EFG with two players $N=\{A,B\}$. The tree representation of $G_c(A)$ in \Cref{tbl:Gc} defines $\mH$, $P$ and $u$\footnote{The subgames $S_i$, $S'_i$ are given in the appendix.}, 
	with the actions of the game being summarized in \Cref{tbl:actions}.
\end{definition}

Note that 
	the \emph{utility  function} $u$ of $G_c(A)$ in \Cref{tbl:Gc} assigns player $p\in N$ the money player $p$ received minus the money player $p$ deserved based on the latest channel state. The  values of closing ($\alpha$), updating ($\rho$) and waiting ($-\epsilon$) are also considered in \Cref{tbl:Gc}. As discussed in \Cref{sec:pcn}, the fee needed for the closing transaction is assumed to have been reserved among the locked funds in the channel all the time and is spent upon closing, therefore not affecting the players' channel balance.

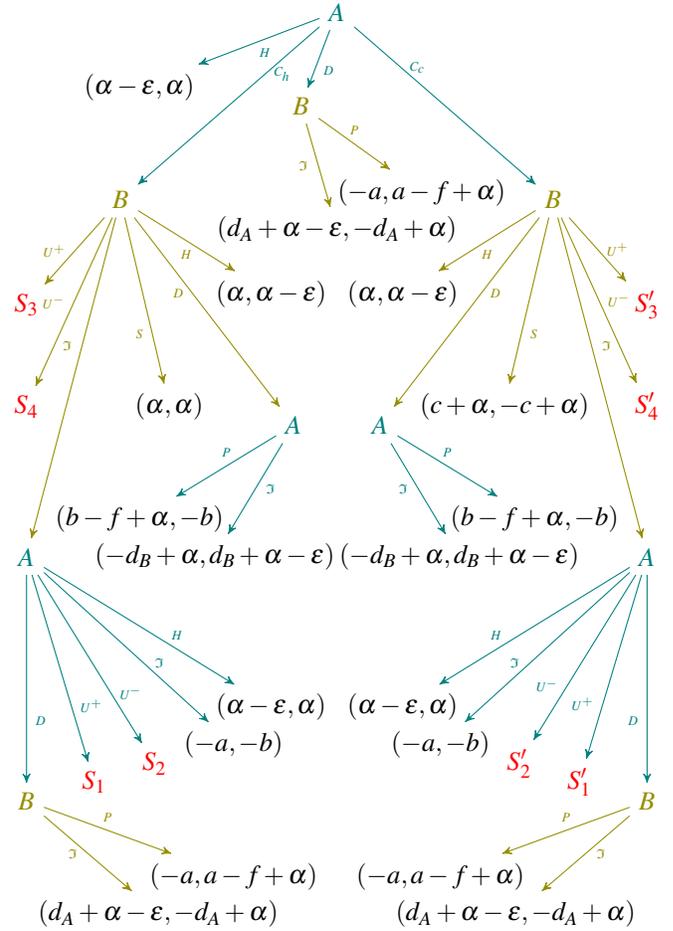
\begin{figure}[t]
	\centering
	\begin{tikzpicture}[scale=1,->,>=stealth',auto,node distance=2cm, el/.style = {inner sep=2pt, align=left, sloped}]
		\node (1) at (-0.375,0.5) {\color{teal}$A$};
		
		\node (2) at (-3.25,-2) {\color{olive} $B$} ; 
		\node (3) at (-3,-0.5) { $(\alpha-\epsilon,\alpha)$};
		\node (4) at (-0.85,-0.75) {\color{olive}$B$};
		\node (5) at (2.5,-2) {\color{olive}$B$};
		
		\node (8) at (0.75,-1.9) { $(-a,a-f+\alpha)$};
		\node (9) at (-0.375,-2.4) { $(d_A+\alpha-\epsilon, -d_A+\alpha)$};
		
		\node (6) at (-4.5,-6.75) {\color{teal}$A$};
		\node (7) at (-2.6,-4.75) { $(\alpha,\alpha)$};
		\node (27) at (-1.25,-3.25) {$(\alpha,\alpha-\epsilon)$};
		\node (28) at (-4.5,-3.4) {\color{red}$S_3$};
		\node (29) at (-4.5,-4.75) {\color{red}$S_4$};
		\node (30) at (-0.95,-5) {\color{teal}$A$};
		
		\node (10) at (1.85,-4.75) {$(c+\alpha,-c+\alpha)$};
		\node (11) at (3.75,-6.75) {\color{teal}$A$};
		\node (31) at (0.2,-5) {\color{teal}$A$};
		\node (32) at (3.75,-4.75) {\color{red}$S'_4$};
		\node (33) at (3.75,-3.4) {\color{red}$S'_3$};
		\node (34) at (0.5,-3.25) {$(\alpha,\alpha-\epsilon)$};
		
		\node (36) at (-3,-6.25) { $(b-f+\alpha,-b)$} ; 
		\node (37) at (-2,-6.75) { $(-d_B+\alpha,d_B+\alpha-\epsilon)$};
		\node (38) at (2.25,-6.25) { $(b-f+\alpha,-b)$};
		\node (39) at (1.25,-6.75) {$(-d_B+\alpha,d_B+\alpha-\epsilon)$};
		
		\node (12) at (-4.5,-10) {\color{olive}$B$};
		\node (15) at (-3.6,-9.75) {\color{red}$S_1$};
		\node (16) at (-2.8,-9.5) {\color{red}$S_2$};
		\node (13) at (-1.25,-8.75) {$(\alpha-\epsilon,\alpha)$};
		\node (14) at (-1.75,-9.25) {$(-a,-b)$};

		\node (21) at (3.75,-10) {\color{olive}$B$};
		\node (18) at (2.85,-9.75) {\color{red}$S'_1$};
		\node (17) at (2.05,-9.5) {\color{red}$S'_2$};
		\node (20) at (0.5,-8.75) { $(\alpha-\epsilon,\alpha)$};
		\node (19) at (1.,-9.25) {$(-a,-b)$};
		
		\node (22) at (-1.75,-11) {$(-a,a-f+\alpha)$} ; 
		\node (23) at (-2.75,-11.5) {$(d_A+\alpha-\epsilon,-d_A+\alpha)$};
		\node (24) at (1,-11) { $(-a,a-f+\alpha)$};
		\node (25) at (2,-11.5) {$(d_A+\alpha-\epsilon,-d_A+\alpha)$};	
		
		{\color{teal}
		\path (1) edge node [right, pos=0.3] {\tiny $C_h$}  (2);
		\path (1) edge node [below, pos=0.45] {\tiny $H$}  (3); 
		\path (1) edge node [right, pos=0.7] {\tiny $D$}  (4);
		\path (1) edge node [above, pos=0.35] {\tiny $C_c$}   (5);}
		
		{\color{olive}	
		\path (4) edge node [above] {\tiny $P$}   (8);
		\path (4) edge node [left] {\tiny $\mathfrak{I}$}   (9);
		}
		
		{\color{teal}
		\path (6) edge node [right, pos=0.7] {\tiny $D$}   (12);
		\path (6) edge node [above, pos=0.7] {\tiny $H$}   (13); 
		\path (6) edge  node [above, pos=0.7] {\tiny $\mathfrak{I}$}  (14);
		\path (6) edge node [right, pos=0.7] {\tiny $U^+$}   (15);
		\path (6) edge node [right, pos=0.7] {\tiny $U^-$}   (16);
		\path (11) edge node [left, pos=0.67] {\tiny $U^-$}  (17);
		\path (11) edge node [left, pos=0.7] {\tiny $U^+$}   (18);
		\path (11) edge node [above, pos=0.7] {\tiny $\mathfrak{I}$}  (19); 
		\path (11) edge node [above, pos=0.7] {\tiny $H$}   (20);
		\path (11) edge node [left, pos=0.7] {\tiny $D$}   (21);}
		
		{\color{olive} 
		\path (12) edge node [above] {\tiny $P$}   (22);
		\path (12) edge node [left] {\tiny $\mathfrak{I}$}   (23); 
		\path (21) edge node [above] {\tiny $P$}   (24);
		\path (21) edge node [right] {\tiny $\mathfrak{I}$}   (25); }
		
		{\color{olive}
		\path (2) edge node [below, pos=0.5] {\tiny $H$}   (27); 
		\path (2) edge node [left, pos=0.5] {\tiny $U^+$}   (28);
		\path (2) edge node [left, pos=0.5] {\tiny $U^-$}   (29);
		\path (2) edge node [left, pos=0.4] {\tiny $D$} (30);
		\path (2) edge node [left, pos=0.4] {\tiny $\mathfrak{I}$}  (6);
		\path (2) edge node [left, pos=0.7] {\tiny $S$}   (7); 
	
		\path (5) edge node [right, pos=0.7] {\tiny $S$}   (10);
		\path (5) edge node [right, pos=0.4] {\tiny $\mathfrak{I}$}  (11);
		\path (5) edge node [right, pos=0.4] {\tiny $D$}   (31);
		\path (5) edge node [right, pos=0.5] {\tiny $U^-$}   (32);
		\path (5) edge node [right,pos=0.5] {\tiny $U^+$}   (33); 
		\path (5) edge node [below, pos=0.5] {\tiny $H$}   (34);}
		
		{\color{teal}
		\path (30) edge node [above] {\tiny $P$}   (36);
		\path (30) edge node [right] {\tiny $\mathfrak{I}$}   (37); 
		\path (31) edge node [above] {\tiny $P$}   (38);
		\path (31) edge node [left] {\tiny $\mathfrak{I}$}  (39);}
	\end{tikzpicture}
	\vspace{-0.2cm}
\caption{Closing Game $G_c(A)$.}
\label{tbl:Gc}
\end{figure}

The closing game for player $B$, $G_c(B)$ is defined similarly to  $G_c(A)$, with the roles of $A$ and $B$ being swapped in \Cref{def:closing}. 
Based on the closing games $G_c(A)$ and $G_c(B)$, we consider  the closing phase in an off-chain channel as given in \Cref{fig:C} and defined below.


\begin{definition}[Closing Phase]
	The  \emph{closing phase} of an off-chain channel modeled by a closing game $G_c(A)$ is initiated in one of three ways: (i)  $A$ starts with a closing action $C$,  and thus triggers the {closing game} $G_c(A)$; 
	(ii) $A$ does not start a closing action, thus performing action ignore $\mathfrak{I}$, but  $B$ starts with a closing action $C$ and triggers $G_c(B)$; 
	or (iii) none of the players $A$ and $B$ ever start closing, that is $B$ also choosing action $\mathfrak{I}$, in which case the money stays locked in the channel. Then, we get the EFG $\Gamma_C$ from \Cref{fig:C} \NEWR{modeling the closing phase of $G_c(A)$ and $G_c(B)$}.
	\end{definition}

	\section{Closing Games for \REPLACER{the Security Analysis of Lightning Channels}{Secure  Lightning Channels}}\label{sec:sec}

We  now show that the closing games from \Cref{def:closing}  precisely capture \NEWR{secure} closing \NEWR{phases} in Lightning channels~\cite{lightning}. Namely, the following two terminal histories of closing games   model the honest behavior of Lightning: 
(i) history $(H)$ from \Cref{tbl:Gc} represents unilateral honest closing of $A$, yielding utility $(\alpha-\epsilon,\alpha)$;  
and (ii)  history $(C_h,S)$ captures the attempt of $A$ to close collaboratively and honestly, while $B$ signs, with a utility of $(\alpha,\alpha)$. Our security analysis  focuses on these two honest histories  of Lightning channels. 

\NEWR{
\begin{definition}[Honest Closing]\label{def:honest:closing}
The only \emph{honest histories} in the closing game $G_c(A)$ are the terminal histories $(H)$ honest unilateral closing and $(C_h,S)$ honest collaborative closing. All strategies yielding one of the two histories are considered \emph{honest strategies}.
\end{definition}
}

In the following, the values $d_{A}$ (resp. $d_{B}$)  defined in \Cref{tbl:actions} (line $D$) represent the difference of funds between the latest state and the old one that is dishonestly posted on chain by $A$ (resp. $B$). In other words, if $(a,b)$ is the latest state, the one posted on chain is   $(a+d_A,b-d_A)$ (resp.  $(a-d_B,b+d_B)$), thus enabling dishonest closing attempts of profit $d_A$ for $A$ (resp.  $d_B$ for $B$).
   The values $p_{A,B}$ (\Cref{tbl:actions}, lines $U^+$, $U^-$) and $c$ (\Cref{tbl:actions}, line $C_c$)  can respectively be chosen by $A$ and $B$ at the time of the action and do not depend on previous distribution states.  \NEWR{Based on this setting, we derive the security properties~\ref{P1} and \ref{P2} of Lightning channels as given below. The omitted proofs are given in the appendix. 

} 

\begin{theorem}[Weak Immunity of Honest Behavior -- \ref{P1}]
 \label{thm:wi}
 The terminal histories $(H)$ \NEWR{of honest unilateral closing}, and $(C_h,S)$ \NEWR{ of honest collaborative closing} of $G_c(A)$ are weak immune, if \NEWR{the channel balances are higher than the fee required in a revocation transaction, that is} if $a\geq f$ and $b\geq f$.
\end{theorem}

\begin{figure}
    \centering
	\begin{tikzpicture}[->,>=stealth',auto,node distance=2cm, el/.style = {inner sep=2pt, align=left, sloped}]
			\node (1) at (0,0) {$A$};
			\node (2) at (0,-1) {$G_c(A)$} ; 
			\node (3) at (2,0) {$B$};
			\node (4) at (2,-1) {$G_c(B)$};
			\node (5) at (4.5,0) {$(-a,-b)$};
	
			\path (1) edge  node [left, pos=0.5] {\tiny $C$} (2);
			\path (1) edge  node [above, pos=0.5] {\tiny $\mathfrak{I}$} (3); 
			\path (3) edge  node [left, pos=0.5] {\tiny $C$} (4);
			\path (3) edge  node [above, pos=0.5] {\tiny $\mathfrak{I}$} (5);
	\end{tikzpicture} 
    \caption{Closing Phase $\Gamma_C$.}
    \label{fig:C}
\end{figure}
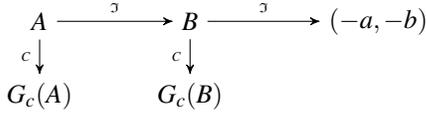



\Cref{thm:wi} implies that as long as both players have a minimal balance of $f$ in the channel, no honest player can lose money. As such, \Cref{thm:wi} establishes the security property \ref{P1} ensuring ``no honest loss" \NEWR{in the channel}. 

Further, to ensure the security property \ref{P2} of ``no deviation",  we require that 
\begin{align}
&a-p_B+d_A \geq f \quad \text{and} \label{eq:ineqA} \\ &b-p_A+d_B \geq f\;\label{eq:ineqB} .
\end{align}

To understand the inequations~\eqref{eq:ineqA}-\eqref{eq:ineqB} consider
the history $(C_h,\mathfrak{I},U^-,\mathfrak{A},D)$ in \Cref{tbl:Gc}, respectively $S'_1$. This history formalizes the case where $A$ attempts honest collaborative closing (action $C_h$) and $B$ ignores it (action $\mathfrak{I}$). Then $A$ proposes an update (action $U^-$) from state $(a,b)$ to state $(a-p_B, b+p_B)$ and $B$ agrees (action $\mathfrak{A}$). Finally, $A$ closes dishonestly (action $D$) using the old distribution state $(a-p_B+d_A, b+p_B-d_A)$.  
Let us also study the options $B$ has.  By ignoring $A$'s behavior (action $\mathfrak{I}$), $B$ receives $b+p_B-d_A$ instead of the fair amount $b+p_B$, leaving $B$ with a loss of $d_A$. By publishing the revocation transaction (action $P$), $B$ receives $a+b$ but has to pay the fee $f$ for pushing it on the blockchain, which leads to a win of $a-p_B-f$. Therefore, the win should be greater than the loss; hence 
\begin{align}
a-p_B-f \geq -d_A \quad \Leftrightarrow \quad a-p_B+d_A\geq f \;,
\end{align}
in order for a rational $B$ to publish the revocation transaction. This in turn yields a loss for $A$ and hence  discourages $A$ from closing dishonestly, which is necessary for the incentive compatibility \ref{P2} of Lightning's closing phase.
By swapping $A$'s and $B$'s roles, we get the  prerequisite formulated in \eqref{eq:ineqB}. These extreme cases of dishonest closing subsume the others. Thus, the only preconditions we need in
the following \Cref{thm:incentcomp} are \eqref{eq:ineqA}--\eqref{eq:ineqB}.
 In summary, formulas~\eqref{eq:ineqA}--\eqref{eq:ineqB}  ensure that ignoring any dishonest closing attempt is worse than publishing the revocation transaction.  Property \ref{P2} is then established by the following theorem. 

\begin{theorem}[Incentive-Compatibility -- \ref{P2}]
\label{thm:incentcomp} 
\noindent If $a-p_B+d_A \geq f$ and $b-p_A+d_B \geq f$, then \begin{enumerate}
    \item honest unilateral closing $(H)$ is \srp, but \emph{not} practical.
    \item honest collaborative closing $(C_h,S)$ is \srp. It is practical iff $c \neq p_A$.
\end{enumerate}
\end{theorem}

\begin{remark}[Explanation of $c \neq p_A$] \label{rmk:ceqp}
The condition $c \neq p_A$ in \Cref{thm:incentcomp} has the following  relevance. Player $A$ can in principle choose to propose dishonest collaborative closing (action $C_c$), providing $A$ an unfair advantage of value $c$. Then, either $B$ (action $U^+$) or $A$ ($B$ choosing action $\mathfrak{I}$ to ignore  first, then $A$ taking action $U^+$) can propose a channel update $(a,b) \mapsto (a+c,b-c)$. The value of the update $p_A$ is now equal to the amount player $A$ cheated with in $C_c$: $p_A=c$. 
In this special case, the closing game behaves differently.  The described histories $(C_c, \mathfrak{I}, U^+)$ and $(C_c,U^+)$ lead to the subgames $S_1'$  and $S_3'$  respectively. Let us consider $S_3'$ with $p_A=c$.

Assume $A$ agrees to the update, action $\mathfrak{A}$, and player $B$ signs the initially unfair collaborative closing attempt of $A$. Since in the meantime the channel was updated by the exact amount that $A$ tried to cheat with, the pending collaborative closing now contains the fair split. Therefore, both players profit from this course of action, yielding utility $(\rho+\alpha,\rho+\alpha)$. The analog can be achieved in subgame $S_1'$ with the history $(\mathfrak{A},\mathfrak{I},S)$. In fact, for $p_A=c$, those histories are the only practical ones and provide the mutually best outcome possible.

However,  updating to $(a+c,b-c)$ first and then closing honestly and collaboratively yields the exact same result. This is why we study the closing game without the possibility of updating after a closing attempt in the next section \Cref{sec:noup}.
\end{remark}

We  now state our first main security theorem. Since  $(H)$ is not practical, a rational player will not play it. Hence, the terminal history $(H)$ is not  secure. We get the following security result instead for $(C_h,S)$.

\begin{theorem}[Security of $G_c(A)$]\label{thm:sec:closing}
If $a\geq f$, $b\geq f$, $a-p_B+d_A \geq f$, $b-p_A+d_B \geq f$, and $c \neq p_A$, then the closing game $G_c(A)$ together with the honest behavior $(C_h,S)$ is \emph{secure}.
\end{theorem}
\begin{proof}
As $a\geq f$ and $b\geq f$, we have that $(C_h,S)$ is weak immune (\Cref{thm:wi}). Since  $a-p_B+d_A\geq f$, we derive  $b-p_A+d_B \geq f$ and $c \neq p_A$, \NEWR{we have that $(C_h,S)$ is also practical} and \srp{} (\Cref{thm:incentcomp}). Hence, by \Cref{def:secure},  $(C_h,S)$ is secure.
\end{proof}

\NEWR{\Cref{thm:sec:closing} implies that for  honest and rational players the action of  collaborative closing followed by signing $(C_h,S)$ is the best way to close an off-chain channel. It also implies, that rational adversaries will cooperate. Further, Byzantine players represent no threat as long as their channel balances are high enough and they do not engage in special cases of  channel updates after a collaborative closing attempt. }

\NEWR{We note that for proving our security properties~\ref{P1}-\ref{P2} in \Cref{thm:wi}--\Cref{thm:sec:closing},
we rely  on a succinct analysis of the finite graph properties of the closing game $G_C(A)$ from \Cref{tbl:Gc}. While automated  approaches analyzing a finite number of graph properties exist, see e.g.~\cite{GRAPHtoolone,graphtooltwo}, these approaches cannot handle (game) graphs where graph leaves contain variables, instead of specific numerical values, which is the case of $G_C(A)$. For such cases, automated reasoning tools, such as theorem provers, need to be combined with graph-theoretic manipulations of $G_C(A)$, an approach we aim to investigate as a future work towards automating the security analysis (and proofs) of closing games. }

\subsection{Closing Games without Updates}\label{sec:noup}
We will now consider a variation of closing games without updates, as updating is not beneficial for at least one player  upon closing.
Furthermore, we avoid special cases such as the one described in \Cref{rmk:ceqp}, which should be  equivalent to updating before initiating $G_c(A)$, and then closing honestly and collaboratively. As such, 
the {\it closing game  $G_c(A)$ without updates}  results from removing all actions $U^+$ and $U^-$ in   \Cref{tbl:Gc}.    For the resulting closing game  $G_c(A)$ without updates we get the following security result similar to \Cref{thm:sec:closing}. 


\begin{theorem}[Security of $G_c(A)$ without Updates]\label{thm:noup}
If $a\geq f$ and $b\geq f$, then the closing game $G_c(A)$ without updates and together with  both honest histories $(H)$ and $(C_h,S)$ is \emph{secure}.
\end{theorem}

\begin{proof} 
\NEWR{We  respectively fix honest strategies $\sigma$  and $\sigma'$  for histories  $(H)$ and  $(C_h,S)$; let $\sigma'$  have $A$ choosing $C_h$ initially, $P$ after $(C_h,D)$ and $H$ after $(C_h,\mathfrak{I})$, and then $B$ choosing  $S$ after $(C_h)$, $P$ after $(D)$ and $H$ after $(C_c)$. 
 Infer that the deviation of $A$   causes negative utility for $B$, whereas the deviation of $B$ leads to non-negative utility for $A$ as  $b-f\geq 0$. By \Cref{thm:wi} we thus have that  $\sigma'$, and therefore $(C_h,S)$, are weak immune. In addition, \Cref{thm:wi} implies that also $(H)$ is weak immune.}
 
To show practicality, we compute all subgame perfect terminal histories. From $a\geq f$ and  $b\geq f$ we have  $a+d_A\geq f$ and $b+d_B \geq f$.  
\NEWR{Since closing with a  dishonest behavior yields utility $a-f+\alpha$, $b-f+\alpha$ respectively, whereas ignoring a dishonest behavior   leads to $-d_A+\alpha$ and $-d_B+\alpha$, we conclude that  the best choice after action $D$ is always $P$.}
%
Thus, $A$'s best choice after $(C_h,\mathfrak{I})$ and $(C_c,\mathfrak{I})$ is $H$. Therefore, $B$ has the two subgame perfect options $\mathfrak{I}$ and $S$ after $(C_h)$, and only $\mathfrak{I}$ after $(C_c)$, yielding thus  the following practical histories:   history $(C_h,S)$ with utility $(\alpha,\alpha)$; and $(C_h,\mathfrak{I},H)$, $(C_c,\mathfrak{I},H)$, and $(H)$ each with utility $(\alpha-\epsilon,\alpha)$. Therefore, both $(H)$ and $(C_h,S)$ are practical.

Note that every practical terminal history is a Nash Equilibrium, since if a deviation could benefit a player, the player would have chosen differently already. As  \srp{} is equivalent to Nash Equilibria in two-player games \NEWR{(by \Cref{def:sr}
and 
\Cref{lemma:impl}), we use \Cref{def:practicalEFG} and \Cref{lemma:impl}} to conclude that  practicality of $(H)$ and $(C_h,S)$ implies collusion resistance \srp{} of $(H)$ and $(C_h,S)$. 
\NEWR{As $(H)$ and $(C_h,S)$ are both weak immune, practical and \srp{}, by \Cref{def:secure} we  infer   that they are also secure. }
\end{proof}

\NEWR{
\begin{remark}
Note that the analysis of utilities in  the closing game $G_c(A)$  crucially depends on constraints of the underlining ordering that we set in \Cref{def:utility:order}, and thus on the values of variables $a,b,c,d_{A,B},f$ in \Cref{tbl:actions}.
%
In general, the bigger $\epsilon$ gets in \Cref{def:utility:order}, the more discouraged is closing unilaterally in \Cref{tbl:actions}, and hence in \Cref{tbl:Gc}. Further, $B$  is more likely to accept a dishonest collaborative closing attempt $C_c$, as it is better to lose $c$ than to lose $\epsilon$.
\end{remark}
}

We further study  what happens if a player has almost no funds left in a channel. In particular, we show that security properties, \NEWR{in particular weak immunity and practicality},  are  violated in this  case, thereby formalizing  the following folklore   in the community. 

\begin{theorem}[Little Funds] \label{thm:secflaw}
If $a<f$, then  only 
terminal histories that involve an explicit cheating attempt are weak immune \NEWR{in the closing game $G_c(A)$  without updates}. A terminal history involves an explicit cheating attempts  if one of its actions is $C_c$ or $D$.
\end{theorem}
\begin{proof}
Let $\sigma$ be any strategy, yielding a history that does not involve an explicit cheating attempt. Then $A$ can deviate to a strategy where $A$ chooses $D$ as its first action. In this case, the honest $B$ gets negative utility, \NEWR{no matter whether} $B$ chooses $P$ or $\mathfrak{I}$, since $a<f$. Hence, only histories that involve explicit cheating attempts can be weak immune.
\end{proof}

We next derive the following results on security properties. 
\begin{corollary} \label{cor:cor2}
If there exists an old channel state $(a+d_A,b-d_A)$, with $a+d_A <f$, then neither history $(H)$ nor $(C_h,S)$ is  weak immune nor practical, but \srp.
\end{corollary}

\begin{corollary}\label{cor:cor3}
A rational party should \emph{never}, in any channel, let the opponent's balance fall below $f$, because at that point the other party can always cause financial loss by closing dishonestly and unilaterally\footnote{The special edge cases $a=0$ or $b=0$ are considered in the appendix.}. 
\end{corollary}
\begin{proof}
Once the opponent's balance is below $f$, that party can start the closing game, therefore the opponent becoming $A$. Thus, by applying \Cref{thm:secflaw}, it follows that the opponent can make the rational player lose money by closing unilaterally and dishonestly. If it is not the first time that $A$'s balance is below $f$ and the respective old state contains a higher balance for $A$ than the latest one, then we are even in the situation of \Cref{cor:cor2}. It is thus  rational of $A$ (practical) to close dishonestly. 
\end{proof}

\subsection{Optimal Strategy for Closing Off-Chain}
\NEWR{To summarize, our security analysis based on  closing games for Lightning channels yields the following  results.} 
 \Cref{thm:noup}-\Cref{thm:secflaw}, together with  with \Cref{cor:cor2}-\Cref{cor:cor3}, allow us to derive the optimal strategy for closing an off-chain channel for a rational and suspicious player. \NEWR{ We next describe and illustrate this optimal strategy, highlighting the main steps of our security analysis based on \Cref{thm:noup}-\Cref{thm:secflaw}.}
 
 \NEWR{Without loss of generality,} we  assume the current state of the channel is $(a,b)$.\vspace{0.2cm}
 
 \noindent The player, \NEWR{assumed to be player $A$}, who initiated the closing phase shall:
\begin{itemize}
    \item try to close honestly and collaboratively (action $C_h$), if there does not exist an old state $(a+d_A, b-d_A)$, where $d_A>0$ and $a+d_A<0$. In case the other player, that is player $B$,  does not sign (action $S$), \NEWR{player $A$}  shall close honestly and unilaterally (action $H$).
    
    If player $B$ closed dishonestly and unilaterally (action $D$), \NEWR{player $A$} shall:
    \begin{itemize}
        \item state the revocation transaction \NEWR{(action $P$)}, if the state used for cheating was $(a-d_B,b+d_B)$, where $d_B>0$ and $b+d_B \geq f$. 
        \item ignore the cheating otherwise \NEWR{(action $\mathfrak{I}$)}, as it yields less loss.
    \end{itemize}
    \item close dishonestly and unilaterally (action $D$), if there exists an old state $(a+d_A,b-d_A)$, where $d_A>0$ and $a+d_A<f$. In this case, \NEWR{player $A$} shall use the old distribution state $(a+d_A',b-d_A)$, with the highest $d_A'>0$ that still satisfies $a+d_A'<f$.
\end{itemize}
\vspace{0.2cm}
The reacting player, \NEWR{in this case assumed to be player $B$}, shall:
\begin{itemize}
    \item sign the collaborative honest closing attempt (action $S$), if applicable, if there is no old state  $(a-d_B,b+d_B)$, $d_B>0$ in which the funds of \NEWR{player $B$}  are less then $f$, that is if $b+d_B<f$. 
    \item  close honestly and unilaterally (action $H$), in case of a dishonest collaborative closing attempt (action $C_c$). This holds, if there is no old state  $(a-d_B,b+d_B)$, $d_B>0$ in which the player \NEWR{$B$'s} funds are less then $f$, that is  $b+d_B<f$. 
    \item otherwise ignore \NEWR{(action $\mathfrak{I}$)} the collaborative and honest/dishonest closing attempt, if applicable, and close dishonestly and unilaterally (action $D$), using the old state $(a-d_B',b+d_B')$, with the highest $d_B'>0$ that still satisfies $b+d_B'<f$.
    \item state the revocation transaction (action $P$), if player $A$ tried to close dishonestly and unilaterally (action $D$) with state $(a+d_A,b-d_A)$, where $d_A>0$ and $a+d_A \geq f$.
    \item ignore (action $\mathfrak{I}$) if player $A$ closed dishonestly (action $D$), in the case where $a+d_A < f $, as it yields less loss.
\end{itemize}

\begin{example}
Let players $A$ and $B$ share a channel with initial balance $(5,5)$ and let us assume the fee for publishing a revocation transaction $f=2$. After the first update let their state be $(3,7)$. The optimal way for $A$ to close now is $C_h$ and for $B$ to sign. Dishonest closing would cause $B$ to publish the revocation transaction, yielding a loss of $3$ for $A$ and a profit of $3-2=1$ for $B$.

The next update could be $(1.8,8.2)$. The best way to close for $A$ is still $C_h$. Dishonest closing using $(3,7)$, for example, would still cause $B$ to publish the revocation transaction. Player $B$ would in this case lose $1.8-2=-0.2$, but he would lose more, $7-8.2=-1.2$, by ignoring it. 

Another update could be $(1,9)$. Now the optimal strategy for $A$ to close is $D$, using the old state $(1.8, 8.2)$. Ignoring the dishonest closing (action $\mathfrak{I}$) brings $B$ $-0.8$, but proving $A$'s cheating (action $P$) leads to $1-2=-1$. Hence, a rational $B$ will choose to ignore (action $\mathfrak{I}$), that means $B$ does not publish the revocation transaction.
\end{example}

	\section{\REPLACER{Further Games for Security Analysis}{Beyond Closing Games for Off-Chain Security} }\label{sec:refine}

\NEWR{Our game-theoretic  analysis so far focused on using closing games to capture  security properties of off-chain channels (\Cref{sec:models}), and in particular of  Lightning channels (\Cref{sec:sec}). In this section, we show that our game-theoretic formalism from \Cref{sec:theory} is  expressive enough to analyse more complex protocols than just closing phases in Lightning channels. 
In particular, we introduce a new EFG, called the {\it Routing Game} in \Cref{sec:routingGame}, and use this game in
\Cref{sec:routingGameAnalysis} 
to disprove security of Lightning's routing mechanism amid the Wormhole and Griefing  attacks~\cite{AMHL,griefing}.
 We also discuss a natural extension of our analysis to model  other off-chain protocols in \Cref{sec:otherProtocols}. 
}

\begin{figure*}[t]
	\centering
	\begin{tikzpicture}[scale=0.975,->,>=stealth',auto,node distance=2cm, el/.style = {inner sep=2pt, align=left, sloped}]
		\node (1) at (-3,0) {$B$};
		\node (7) at (-4.25,-0.75) {$(0,0,0,0,0)$} ; 
		\node (12) at (-1.5,-0.5) {$A$};
		\node (13) at (-3.25,-1.75) {$(0,0,0,0,0)$};
		\node (2) at (0,-1) {$E_1$};
		\node (8) at (-2,-2.5) {$(-\epsilon,0,0,0,0)$} ;
		\node (14) at (1.5, -1.5) {$I$};
		\node (15) at (-0.5, -3) {$(-\epsilon,-\epsilon,0,0,0)$};
		\node (3) at (3,-2) {$E_2$};
		\node (9) at (1,-3.5) {$(-\epsilon,-\epsilon,-\epsilon,0,0)$} ;
		\node (16) at (4.5,-2.5) {$B$};
		\node (17) at (2.5, -4) {$(-\epsilon,-\epsilon,-\epsilon,-\epsilon,0)$};
		\node (4) at (6,-3) {$E_2$};
			\node (20) at (3.5,0.5) {$E_1$};
			
			\node (21) at (4.75,0.) {$I$};
			\node (22) at (6.5,2) {\color{cyan}$(\rho,\;m+3f-\epsilon,\;-\epsilon,\;-m,\;\rho)$};		
			\node (46) at (0.75,1.5) {$(m+3f+\rho-\epsilon,-\epsilon,-\epsilon,-m,\rho)$};
			
			\node (77) at (3.5,1.75) {\color{olive} $B$};
			\node (78) at (2,2.25) {\color{olive} $I$};
			\node (79) at (-1.5,2.75)  {\color{olive}$(\rho,\;m+3f-\epsilon,\;-\epsilon,\;-m,\;\rho)$}; 
			\node (80) at (5.5,2.75) {\color{olive}$(\rho, f, m+2f-\epsilon, -m, \rho)$};

			\node (69) at (6,0.5) {$E_1$};
			\node (70) at (6,-0.5) {$E_1$};

			\node (73) at (8.5,1.25) {$(m+3f+\rho-\epsilon,-\epsilon,-\epsilon,-m,\rho)$};
			\node (74) at (9.5,0.5) {$(\rho,m+3f-\epsilon,-\epsilon,-m,\rho)$};
			
			\node (75) at (8.75,-1.25) {$(m+3f+\rho-\epsilon,-m-2f,m+2f-\epsilon,-m,\rho)$};
			\node (76) at (9.5,-.5) {$(\rho,f,m+2f-\epsilon,-m,\rho)$};

		\node (10) at (4,-4.6) {$B$} ;
		
		    \node (56) at (5,-5) {$\mathfrak{S}_5$} ;
		    \node (57) at (2.75,-5.5) {$I$} ;
		    \node (58) at (0,-4.9) {$(m+3f+\rho-\epsilon,-\epsilon,-\epsilon,-m,\rho)$} ;
		    
		    \node (59) at (1,-5.75) {$E_1$} ;
		    \node (61) at (3.75,-6) {$\mathfrak{S}_6$} ;
		    \node (60) at (2,-7.25) {$B$} ;
		    
		    \node (62) at (-1.75,-6.75) {$(m+3f+\rho-\epsilon,-m-2f,m+2f-\epsilon,-m,\rho)$} ;
		    \node (63) at (-3.,-6.25) {$(\rho,f,m+2f-\epsilon,-m,\rho)$} ;
		    
		    \node (64) at (-2,-7.75) {$(m+3f+\rho-\epsilon,-\epsilon,-\epsilon,-m,\rho)$} ;
		    \node (65) at (2,-8.25) {$E_1$} ;
		    
		    \node (66) at (-2,-8.75) {$(m+3f+\rho-\epsilon,-\epsilon,-\epsilon,-m,\rho)$} ;
		    \node (67) at (-0,-9.25) {$(\rho,m+3f-\epsilon,-\epsilon,-m,\rho)$} ;
		
		\node (18) at (7.5, -3.5) {$I$};
		\node (19) at (6.25, -4.75) {$E_2$};
		    \node (48) at (7.375,-5.125) {$E_1$};
			\node (51) at (4.25,-6.75) {$B$};	    
		    
		    \node (49) at (7.75,-6.25) {$(m+3f+\rho-\epsilon,-\epsilon,-m-f,f,\rho)$};
		    \node (50) at (10,-5.75) {$(\rho,m+3f-\epsilon,-m-f,f,\rho)$};

		    \node (52) at (7.5,-7.25) {$(m+3f+\rho-\epsilon,-\epsilon,-m-f,f,\rho)$};
		    \node (53) at (3.25,-7.75) {$E_1$};
		    
		    \node (54) at (7.25,-8.5) {$(m+3f+\rho-\epsilon,-\epsilon,-m-f,f,\rho)$};
		    \node (55) at (4.5,-9) {$(\rho,m+3f-\epsilon,-m-f,f,\rho)$};
		
		\node (5) at (9,-4) {$E_1$};
		\node (6) at (11,-4.625) {\textcolor{red}{$(\rho,f,f,f,\rho)$}} ;
		\node (11) at (9.,-2.65) {$(m+3f+\rho-\epsilon,-m-2f,f,f,\rho)$};

			\node (32) at (0.5,0.1) {$\mathsf{S}_2$};
			\node (33) at (0,-2) {$\mathbb{S}_2$};
			\node (81) at (-.5,0.1) {$\mathbf{S}_2$};
			
			\node (34) at (2,-0.4) {$\mathsf{S}_3$};
			\node (35) at (1.5,-2.5) {$\mathbb{S}_3$};
			\node (82) at (1,-0.4) {$\mathbf{S}_3$};
			
			\node (36) at (3.5,-0.9) {$\mathsf{S}_4$};
			\node (37) at (3,-3) {$\mathbb{S}_4$};
			\node (83) at (2.5,-0.9) {$\mathbf{S}_4$};
			
			\node (39) at (-3.,-1) {$\mathfrak{S}_1$};
			
			\node (40) at (-1.,0.6) {$\mathsf{S}_1$};
			\node (41) at (-1.5,-1.5) {$\mathbb{S}_1$};	
			\node (84) at (-2,0.6) {$\mathbf{S}_1$};
			
			\node (42) at (4.5,-3.5) {$\mathfrak{S}_2$};
			
			\node (44) at (6.,-4) {$\mathfrak{S}_3$};
			
			\node (45) at (7.5,-4.5) {$\mathfrak{S}_4$};
		
		\path (1) edge node[above] {\tiny $S_H$} (12);
		\path (1) edge node[ above, pos=0.7] {\tiny $\mathfrak{I}$} (7); 
		\path (12) edge node[above] {\tiny $L$} (2);
		\path (12) edge node[above, pos=0.5] {\tiny $\mathfrak{I}$}  (13);
		\path (2) edge node[above] {\tiny $L$} (14);
		\path (2) edge node[above, pos=0.5] {\tiny $\mathfrak{I}$}  (8);
		\path (14) edge node[ above] {\tiny $L$} (3);
		\path (14) edge node[ above, pos=0.5] {\tiny $\mathfrak{I}$}  (15);
		\path (3) edge node[above] {\tiny $L$} (16);
		\path (3) edge node[above, pos=0.5] {\tiny $\mathfrak{I}$}  (9);
		\path (16) edge node[ above] {\tiny $U$} (4);
		\path (16) edge node[ above, pos=0.5] {\tiny $\mathfrak{I}$}  (17);
		\path (4) edge node[above] {\tiny $U$} (18);
		\path (4) edge node[above, pos=0.5] {\tiny $\mathfrak{I}$}  (10); 
		\path (18) edge node[ above] {\tiny $U$} (5);
		\path (18) edge node[ above, pos=0.5] {\tiny $\mathfrak{I}$}  (19);
		\path (5) edge node[above] {\tiny $U$} (6);   
		\path (5) edge node[right, pos=0.5] {\tiny $\mathfrak{I}$}  (11);

			\path (2) edge node[right, pos=.3] {\tiny $L_H$} (32);
			\path (2) edge node[right] {\tiny $L_A$} (33);
			\path (2) edge node[right] {\tiny $L_T$} (81);
			
			\path (14) edge node[right, pos=.3] {\tiny $L_H$} (34);
			\path (14) edge node[right] {\tiny $L_A$} (35);
			\path (14) edge node[right] {\tiny $L_T$} (82);
			
			\path (3) edge node[right, pos=.3] {\tiny $L_H$} (36);
			\path (3) edge node[right] {\tiny $L_A$} (37);
			\path (3) edge node[right] {\tiny $L_T$} (83);
			
			\path (1) edge node[right] {\tiny $S_S$} (39);
			
			\path (12) edge node[right, pos=.3] {\tiny $L_H$} (40);
			\path (12) edge node[right] {\tiny $L_A$} (41);
			\path (12) edge node[right] {\tiny $L_T$} (84);
			
			\path (16) edge node[right] {\tiny $S_S$} (42);
			
			\path (4) edge node[right] {\tiny $S_S$} (44);
			
			\path (18) edge node[right] {\tiny $S_S$} (45);
			
			\path (10) edge node[above, pos=0.9] {\tiny $S_{S_{E_1}}$} (56);
			\path (10) edge node[above, pos=0.6] {\tiny $S_{S_{I}}$} (57);
			\path (10) edge node[above] {\tiny $\mathfrak{I}$}  (58);
			
			\path (57) edge node[above] {\tiny $U$} (59);
			\path (57) edge node[right] {\tiny $\mathfrak{I}$}  (60);
			\path (57) edge node[above, pos=0.9] {\tiny $S_{S_{E_1}}$} (61);
			
			\path (59) edge node[below] {\tiny $\mathfrak{I}$}  (62);
			\path (59) edge node[above] {\tiny $U$} (63);
			
			\path (60) edge node[above] {\tiny $\mathfrak{I}$}  (64);
			\path (60) edge node[left] {\tiny $S_{S_{E_1}}$} (65);
			
			\path (65) edge node[below] {\tiny $U$} (67);
			\path (65) edge node[above] {\tiny $\mathfrak{I}$}  (66);
			\path (19) edge node[below] {\tiny $S_{S_{E_1}}$} (48);
			\path (19) edge node[above] {\tiny $\mathfrak{I}$}  (51);
			
			\path (48) edge node[left] {\tiny $\mathfrak{I}$}  (49);
			\path (48) edge node[above] {\tiny $U$} (50);
			
			\path (51) edge node[above] {\tiny $\mathfrak{I}$}  (52);
			\path (51) edge node[above] {\tiny $S_{S_{E_1}}$} (53);
			
			\path (53) edge node[above] {\tiny $\mathfrak{I}$}  (54);
			\path (53) edge node[left] {\tiny $U$} (55);
			
			\path (4) edge node[left, pos=0.5] {\tiny $S_{S_{E_1}}$} (20);
			
			{\color{olive}\path (20) edge node[right] {\tiny $U$} (77);}
			\path (20) edge node[above] {\tiny $S_{S_I}$} (21);
			\path (20) edge node[above] {\tiny $\mathfrak{I}$}  (46);
			
			\path (21) edge node[above] {\tiny $\mathfrak{I}$}  (69);
			\path (21) edge node[below] {\tiny $U$} (70);
			
			\path (69) edge node[above] {\tiny $\mathfrak{I}$}  (73);
			\path (69) edge node[below] {\tiny $U$} (74);
			
			\path (70) edge node[below] {\tiny $\mathfrak{I}$}  (75);
			\path (70) edge node[above] {\tiny $U$} (76);
			
		{\color{olive}	\path (77) edge node[above] {\tiny $\mathfrak{I}$}  (22);
			\path (77) edge node[above, pos=0.2] {\color{olive}\tiny $S_{S_I}$} (78);
			
			\path (78) edge node[above] {\color{olive}\tiny $\mathfrak{I}$}  (79);
			\path (78) edge node[above] {\color{olive}\tiny $U$} (80);}
			
	\end{tikzpicture}
	\caption{\NEWR{Partial Definition of the Lightning's Routing $G_{\text{rout}}$ and the Fulgor Model $G_{\text{Ful}}$. The olive colored subtree only applies for $G_{\text{rout}}$.}}
	\label{tbl:mymodel}
\end{figure*}
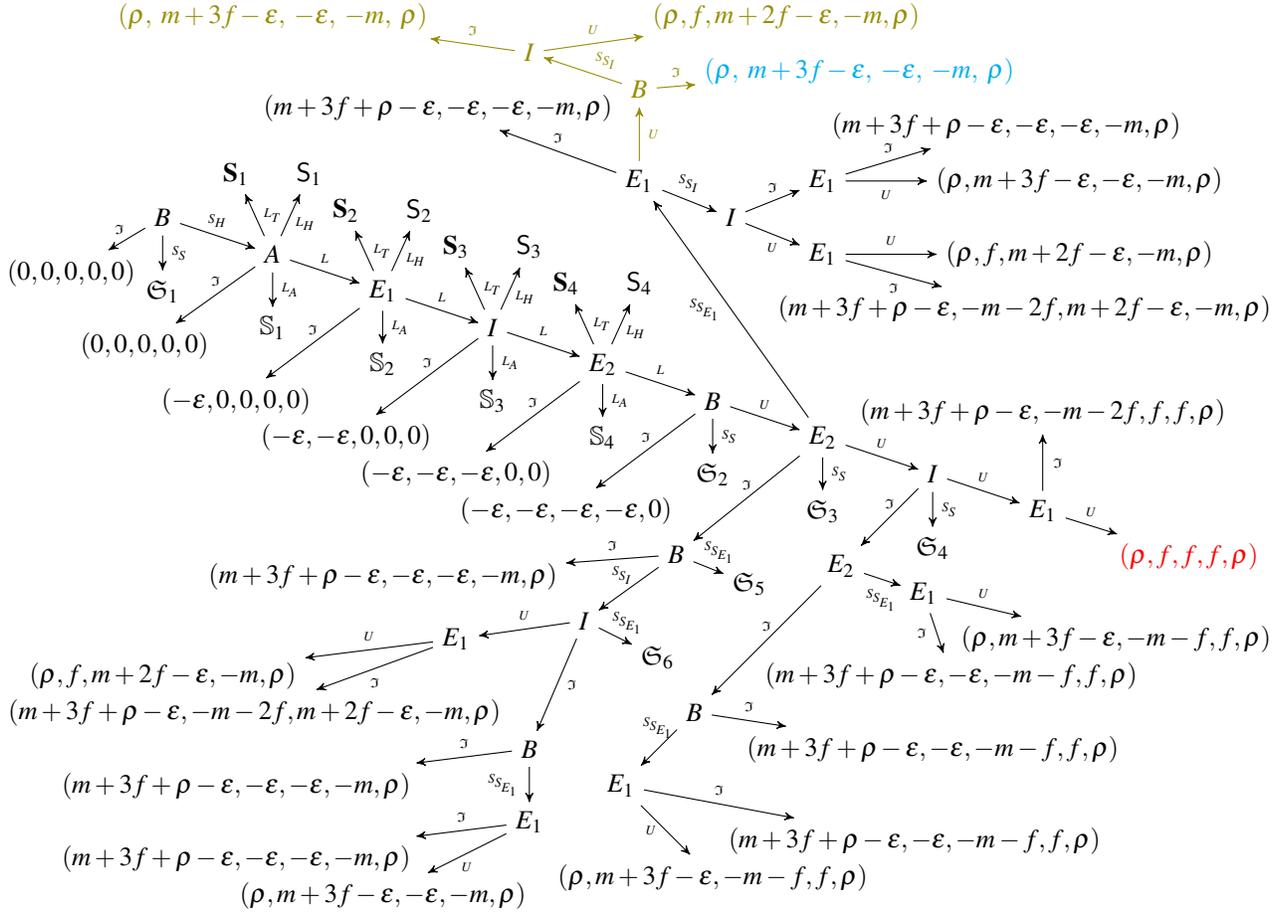

\subsection{\NEWR{Routing Games for Lightning's Routing Module}}\label{sec:routingGame}
%
%
We first  propose a new EFG, 
called the \emph{Routing Game},  showing that EFGs can capture actual attacks, in this case the Wormhole attack \cite{AMHL} and the Griefing attack~\cite{griefing}, which were overlooked for example in~\cite{CITE}. Specifically,  the below defined routing game considers fees $f$, and  \NEWR{supports  actions} allowing  the intermediaries  to choose not to claim their money using the secret $x$  but instead to forward it to another intermediary (as explained in \Cref{sec:pcn}). \NEWR{Additionally, other deviations such as creating a conditional payment (i.e. HTLC) with a different hash value, a different amount, or a different time-out than expected are also considered.} \NEWR{For simplicity, we chose to model our routing game below with five players;  however, an arbitrary number of intermediaries can be modeled.}

\begin{table}[b]
    \centering
        \caption{\NEWR{Possible Actions in $G_{\text{rout}}$.}}
    \begin{tabularx}{\linewidth}{|@{\hspace{.5em}}>{\bfseries}l@{\hspace{.5em}}X@{\hspace{.5em}}|}
    \hline
        $S_H$ &  \emph{Sharing} the secret's \emph{Hash} to enable the others to create HTLCs (action 1 in \Cref{tbl:honestrouting}, \Cref{sec:pcn}). \\
    \hline
        $L$ & \emph{Lock} money, as defined in actions 2--5 in \Cref{tbl:honestrouting}, in an HTLC.\\
    \hline
       $U$ & \emph{Unlocking} the money from an HTLC (actions 6--9 in \Cref{tbl:honestrouting}). Thereby the secret is revealed to the HTLC's creator. \\ 
    \hline
       $\mathfrak{I}$ & \emph{Ignoring} all the previous actions and do nothing. If applicable, until the unlockable HTLC has timed out.\\
    \hline 
     $S_S$ & \emph{Sending the Secret} to another player. If it is sent to a specific player (not leading to $\mathfrak{S}_i$) this player is indicated by another subscript. \\
    \hline 
    $L_H$ & \emph{Locking} money in an HTLC, that uses a different \emph{Hash-lock} than described in \Cref{tbl:honestrouting}.\\
    \hline
   $L_A$ & \emph{Locking} a different \emph{Amount} of money in an HTLC, than described in \Cref{tbl:honestrouting}.\\
    \hline
   $L_T$ & \emph{Locking} money in an HTLC, whose \emph{Time-out} is different from the values described in \Cref{tbl:honestrouting}.  \\
    \hline
    \end{tabularx}
    \label{tbl:actionsrout}
\end{table}

\begin{definition}[Routing Game $G_{\text{rout}}$]
	The \emph{routing game} $G_{\text{rout}}=(N_r,\mH_r, P_r, u_r)$ is an EFG with five players $N=\{A,E_1,I,E_2,B\}$, where 
	\begin{itemize}
	    \item 
the histories $\mH_r$, the next player function $P_r$, and the utility function $u_r$ are defined via the  tree representation of \Cref{tbl:mymodel}.
 The utility tuples in \Cref{tbl:mymodel} assign the first value to $A$, the second to $E_1$, the third to $I$, the fourth to $E_2$, and the last to $B$; 
\item  the  actions of $G_{\text{rout}}$  are as listed  in \Cref{tbl:actionsrout}.
\end{itemize}
\end{definition}

\NEWR{We note that our Routing Game $G_{\text{rout}}$  has four types of subgames, as modeled in \Cref{tbl:mymodel} and described next: (i) subgames that  result from sending the secret to another player $\mathfrak{S}_i$; 
(ii) subgames that result from locking a wrong amount of money in the HTLC $\mathbb{S}_i$; 
(iii) subgames that  result from using a wrong time-out in an HTLC $\mathbf{S}_i$; 
and (iv) subgames that result from using a wrong hash value as lock in the HTLC $\mathsf{S}_i$.
We further note that \Cref{tbl:mymodel} only gives a partial model, 
 as not all subgames are presented in \Cref{tbl:mymodel}. However, within one type of subgame, the game trees are similar. Therefore, we provide only one instance of each type in the appendix, 
 which are the subgames $\mathbb{S}_1$, $\mathsf{S}_2$,  and $\mathbf{S}_3$. An instance of a secret forwarding subgame capturing the Wormhole attack  can be seen in $G_{\text{rout}}$, as  the subtree after history $(S_H, L, L,L, L, U, S_{S_{E_1}})$.}

Let us  emphasize that the  utility function $u_r$ of  $G_{\text{rout}}$ assigns each player $p \in N$ the relative profit of their routing actions and does not mirror the individual channel balances. It also takes the value $\rho$ of a successful payment \NEWR{and the opportunity cost $\epsilon$} into account.

%

 As in the closing games $G_c(A)$ and $G_c(B)$, we aim to align utility and monetary outcome as tight as possible. 
 We adjust the ordering $(\U,\preccurlyeq)$ of \Cref{def:utility:order} \NEWR{by not assuming that $\rho \prec \epsilon$, since achieving an update is the ultimate goal of the routing protocol.} We also consider the utility relative to the amount due to each party.

\subsection{\NEWR{Security Analysis of Lightning's Routing Module}}\label{sec:routingGameAnalysis}

\NEWR{Let us recall \Cref{tbl:honestrouting} and \Cref{tbl:wormhole}, where   player $A$  wants to pay another player $B$ money of value $m$. Since, $A$ and $B$ do not share a channel, the three intermediaries $E_1$, $I$,and $E_2$ support the payment, with each receiving a fee $f>0$  for their collaboration if the payment is successful. Each player who creates an HTLC locks her money for a given time, yielding an opportunity cost of $\epsilon$ if the money is returned. }
If the transaction fails, before anyone has unlocked an HTLC, all parties get utility 0 or $-\epsilon$, depending on whether they created an HTLC or not. Otherwise, the intermediaries' utilities are according to their financial win/loss.
The parties $A$ and $B$ both receive $\rho$ once $B$ is paid. Should the transaction fail after $B$ is paid, but before $A$ has paid, she has utility $m+3f+\rho-\epsilon$; once $E_1$ collects the money, $A$'s  utility is $\rho$. 

\NEWR{In the sequel, we consider the behavior from \Cref{tbl:honestrouting} as the only \emph{honest} history in $G_\text{rout}$, as also formalized next. 

\begin{definition}[Honest Routing]
The only \emph{honest history} in the routing game $G_\text{rout}$ is the history $(S_H, L, L, L,L, U,U,U,U)$. All strategies yielding this history are considered \emph{honest strategies}.
\end{definition}
}

Using our model $G_\text{rout}$ and its honest behavior, we derive the following result. 

\NEWR{
\begin{theorem}[Vulnerability of $G_\text{rout}$ to Wormhole Attacks] \label{thm:wormhole}
	The honest behavior $(S_H, L, L,L,L,U,U,U,U)$ of the  Routing Game $G_{\text{rout}}$ is not \srp{}.
\end{theorem}
\begin{proof}
 The utility of the honest behavior of the routing module  $(S_H,L,L,L,L,U,U,U,U)$  is $(\rho,f,f,f,\rho)$ (as indicated in red in \Cref{tbl:mymodel}). Let us compare this behavior and utility to the deviating terminal history $(S_H, L,L,L,L,U,S_{S_{E_1}},U,\mathfrak{I})$ with a utility of $(\rho,\;m+3f-\epsilon,\;-\epsilon,\;-m,\;\rho)$ (given in blue in \Cref{tbl:mymodel}). It is not hard to argue that the collusion of $E_1$ and $E_2$ (and $B$ by not sending the secret to $I$) strictly profits from the deviation, which yields a joint utility of $3f-\epsilon+\rho$, whereas the honest behavior only yields a joint utility of $2f+\rho$. As such,  collusion resistance \srp{} is violated, since no honest player can prevent the Wormhole attack from happening by following any honest strategy (that is,  a strategy $\sigma$ whose history is the honest behavior $(S_H,L,L,L,L,U,U,U,U)$).
\end{proof}
 
 In conclusion, \Cref{thm:wormhole} formally proves that Lightning's routing module is susceptible to the Wormhole attack.\NEWR{ We further extend this result by noting that not only can    $G_{\text{rout}}$  capture the Wormhole attack, but also the Griefing attack, as stated below.
}

\begin{theorem}[Vulnerability of $G_{\text{rout}}$ to Griefing Attack] \label{thm:griefing}
	The honest behavior $(S_H, L, L,L,L,U,U,U,U)$ of the  Routing Game $G_{\text{rout}}$ is not weak immune.
\end{theorem}
\begin{proof}
For showing that history $(S_H, L, L,L,L,U,U,U,U)$ is not weak immune, we  prove that no strategy which yields this history is weak immune. Let us consider any such strategy $\sigma$. Then, player $A$ has to choose action $L$ after $B$ sent her the secret, that is history $(S_H)$. Assume now $E_1$ deviates and chooses to ignore (action $\mathfrak{I}$). Then $A$'s utility is $-\epsilon \prec 0$. Hence, history $(S_H, L, L,L,L,U,U,U,U)$ is not weak immune.
 \end{proof}

We also obtain the following result as an immediate consequence of  \Cref{thm:wormhole} and \Cref{thm:griefing}.

\begin{corollary}[Security of Routing Module]
	The honest behavior $(S_H, L, L,L,L,U,U,U,U)$ of the  Routing Game  $G_{\text{rout}}$ is not secure. Hence, the Routing Game $G_{\text{rout}}$ is not secure. 
\end{corollary}
}


\NEWR{
\subsection{Further Routing Protocols Beyond the Lightning Network}\label{sec:otherProtocols}

We conclude this paper by arguing that our EFG games, either closing or routing games, are not restricted to Lightning networks but can be used for other protocols as well. In the remaining of this section,  we illustrate how to model  Fulgor~\cite{10.1145/3133956.3134096}, a payment channel network protocol that fixes the Wormhole attack, but not the Griefing attack.

The routing mechanisms used in  Fulgor  is similar to Lightning's routing, and is similarly based on  HTLCs. The main difference lies in the  structure of the secrets and their hashes. Indeed,  while Lightning uses the same secret $x$ for every HTLC, Fulgor provides a different secret and hash lock for each player.

\begin{figure}[t!]
	\centering
\begin{tikzpicture}[scale= 0.8, ->]
	\node (1) at (0,0) {$A$}; 
	\node (2) at (2.5,0) {$E_1$};
	\node (3) at (5,0) {$I$};
	\node (4) at (7.5,0) {$E_2$};
	\node (5) at (10,0) {$B$};
	
\draw[-] (0,2) -- (0, 0.3) node [left, pos=0.6] {1.};
\draw[->] (2.5,2) -- (2.5, 0.3) ;
\draw[->] (5,2) -- (5, 0.3) ;
\draw[->] (7.5,2) -- (7.5, 0.3) ;
\draw[->] (10,2) -- (10, 0.3) ;
\draw[-] (0,2) -- (10, 2) node [above, pos=0.125] { \footnotesize $y_2, x_2,\text{ZKP}_2$} node [above, pos=0.375] { \footnotesize $y_3,x_3,\text{ZKP}_3$}  node [above, pos=0.625] { \footnotesize $y_4,x_4, \text{ZKP}_4$} node [above, pos=.9] { \footnotesize $x_1+x_2+x_3+x_4$}  ;

	\path (1) edge node [above] {\footnotesize $(m+3f,\textcolor{red}{y_1},t_1)$} node [ below, pos=0.3] {2.}(2); 
	\path (2) edge node [above] {\footnotesize $(m+2f,\textcolor{red}{y_2},t_2)$} node [ below, pos=0.3] {3.} (3); 
	\path (3) edge node [above] {\footnotesize $(m+f,\textcolor{red}{y_3},t_3)$} node [ below, pos=0.3] {4.} (4); 
	\path (4) edge node [above] {\footnotesize $(m,\textcolor{red}{y_4},t_4)$} node [ below, pos=0.3] {5.} (5); 
	\path (5) edge[out=240, in=300, distance= 0.8cm] node [below] {\color{cyan} \footnotesize $x_1+x_2+x_3+x_4$}    node [ above, pos=0.3] {6.} (4);
	\path (4) edge[out=240, in=300, distance= 0.8cm] node [below] {\color{cyan} \footnotesize$x_1+x_2+x_3$} node [ above, pos=0.3] {7.} (3); 
	\path (3) edge[out=240, in=300, distance= 0.8cm] node [below] {\color{cyan}  \footnotesize $x_1+x_2$}    node [ above, pos=0.3] {8.} (2); 
	\path (2) edge[out=240, in=300, distance= 0.8cm] node [below] {\color{cyan}  \footnotesize $x_1$}    node [ above, pos=0.3] {9.} (1); 
\end{tikzpicture}
\vspace{-0.2cm}
\caption{Routing in Fulgor.}
\label{tbl:fulgor}
\end{figure}
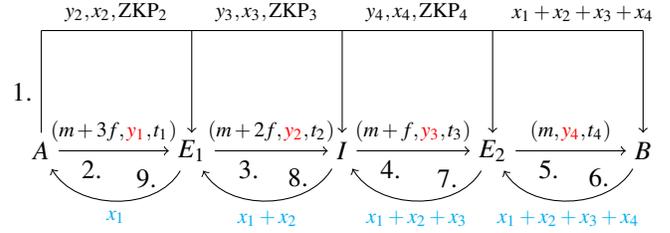

Fulgor's routing mechanism is illustrated in \Cref{tbl:fulgor}, where player $A$ generates
different secrets and hash locks at the beginning.
The secrets and the hashes relate in the following way: $h(x_1)=y_1$, $h(x_1+x_2)=y_2$, $h(x_1+x_2+x_3)=y_3$ and $h(x_1+x_2+x_3+x_4)=y_4$. Therefore, a player only gets to know a sum of secrets when the right-hand party unlocks and subtracts the secret value received from $A$ to unlock their HTLC. $A$ also provides a zero-knowledge-proof $\text{ZKP}_i$ for each intermediary~\cite{CITE-ZKP-PAPER} to prove that the secrets and hashes  constructed this  way guarantee successful unlocking of the left HTLC, which is essential to not lose funds.

The game-theoretical (EFG) model of Fulgor $G_{Ful}$ reported in  \Cref{tbl:mymodel} looks similar to the routing game $G_{\text{rout}}$, yet, with one significant difference. Consider the history $(S_H,L,L,L,L,U,S_{S_{E_1}})$ in \Cref{tbl:mymodel}, which enables player $E_1$ to unlock (action $U$) the HTLC created by $A$. In Fulgor, the same history does not enable $E_1$ to unlock the HTLC. As \Cref{tbl:fulgor} shows, the secrets that $E_2$ can share after action 6 are $x_4$ and $x_1+x_2+x_3+x_4$. Further, $E_1$ only knows $x_2$, thus there is no way to compute $x_1$. This is however the value needed to unlock the HTLC created by $A$.  \REPLACER{Therefore, intuitively Fulgor seems Wormhole proof}{
Indeed, Fulgor is not affected by the Wormhole attack. Nevertheless, similarly to \Cref{thm:griefing}, the honest behavior of Fulgor is not weak immune, as it is vulnerable to the Griefing attack. 
}
}


 	\section{Conclusions}\label{sec:conclusions}
Our work advocates the use of Extensive Form Games (EFGs) for the game-theoretic security analysis of off-chain protocols. In particular, we introduce two instances of EFGs to model the closing and the routing of the Lightning Network. By doing so, we take the first step towards closing the gap existing security proof techniques have due to using informal arguments about rationality. \NEWR{
We express security properties as formal requirements over  joint strategies in EFGs, allowing us to establish optimal strategies for closing off-chain and capture security vulnerabilities amid attacks.  Given the theoretical expressiveness of our EFGs, 
future work includes the definitions of new games to  capture 
a wider range of  off-chain protocols. To overcome the burden of tedious manual analysis, we also plan to leverage
SMT solving and/or automated theorem proving in order to provide automated security proofs. }

\REMOVER{close it further in future work by generalizing and extending our models to various off-chain protocols and by automating the security analysis thereof. To this end, we are interested in leveraging 
SMT solving and/or automated theorem proving for enforcing security of (more) complex protocols. }

\medskip 
\noindent
\textbf{Acknowledgments.} We thank our anonymous reviewers for their valuable feedback. The work was partially supported  by the European Research Council (ERC) under the ERC CoG  ARTIST 101002685 and the ERC CoG BROWSEC 71527; by the TU Wien Doctoral College SecInt;  by the Austrian Science Fund (FWF) projects PROFET P31621 and LogiCS W1255-N23;  by the Austrian Research Promotion Agency (FFG) (COMET K1 SBA, COMET K1 ABC); by the Vienna Business Agency project Vienna Cybersecurity and Privacy Research Center (VISP); by the Austrian Federal Ministry for Digital and Economic Affairs;  the National Foundation for Research, Technology and Development; and the Christian Doppler Research Association through CDL-BOT.

%
%

	
	%
	\bibliographystyle{IEEEtran}
	\bibliography{includes/references}
	
	\newpage
	
	\appendix		
	\subsection{Proof of the Resilience Properties}
We restate the results for better readability.
\newtheorem*{RP}{Lemma~III.1}
\begin{RP}[Resilience Properties] 
	 Let $\sigma \in \mS$ be a joint strategy. The following and only the following implications hold.
	 
\begin{minipage}{0.34\textwidth}
	\begin{enumerate}
		\item $\sigma \text{ is \srs}\; \Rightarrow \; \sigma \text{ is \SR, \srp{}, \sNE}$.
		\item $\sigma \text{ is \SR} \; \Rightarrow \; \sigma \text{ is \srp}$.
	\end{enumerate}
	\end{minipage}
	\begin{minipage}{0.12\textwidth}
\begin{tikzpicture}[scale=0.8, baseline, ->]
		\node (1) at (0,0) {\srs};
		\node (2) at (-1.5,0) {\SR} ; 
		\node (3) at (0,-1.5) {\sNE};
		\node (4) at (-1.5,-1.5) {\srp} ;

		\path (1) edge  (2);
		\path (1) edge (3); 
		\path (1) edge  (4);
		\path (2) edge (4);
	\end{tikzpicture}
\end{minipage}
\end{RP}

\begin{proof}
    We start by showing property (2). Let $\sigma$ be \SR{} and let $S=\{s_1,...,s_j\} \subset N$, $\sigma'_S\in \mS_S$ be arbitrary but fixed. Then, for all $p \in S$ we have  $u_p(\sigma) \geq u_p(\sigma[\sigma'_{s_1}/\sigma_{s_1},...,\sigma'_{s_j}/\sigma_{s_j}])$ and thus also $ \sum_{p \in S} u_p(\sigma) \geq \sum_{p \in S} u_p(\sigma[\sigma'_{s_1}/\sigma_{s_1},...,\sigma'_{s_j}/\sigma_{s_j}])$. Hence $\sigma$ is \srp{} and the implication is proven.
    	For implication (1) we see that \srs{} $\Rightarrow$ \SR{} is trivial. 
	If the property is satisfied for every $S \subseteq N$, then it is also satisfied for every $S \subset N$.
	By (2) and the transitivity of implication we also get \srs{} $\Rightarrow$ \srp.
	For the last implication let $\sigma$ be \srs{} and let $S=\{s_1,...,s_j\} \subseteq N, \, S \neq \emptyset$ and $\sigma'_S \in \mS_S$ be arbitrary but fixed. Then there exists some $p \in S$ and by definition all $p \in S$ satisfy $u_p(\sigma) \geq u_p(\sigma[\sigma'_{s_1}/\sigma_{s_1},...,\sigma'_{s_j}/\sigma_{s_j}])$. Therefore, $\sigma$ is \sNE{}. 
	
	To prove that no other implication holds between those four concepts, we provide three counterexamples. An overview of which game disproves which implication is given in \Cref{tbl:counterex}.
	
		  \begin{table}[b]
    \centering    
    \caption{Game $\Gamma_3$.}
    	\begin{tabular}{|r|c|c|}
    	\hline
    	 & $H_2$ & $D_2$ \\
    	\hline
    	$H_1$ & {\color{red}$(1,1)$} & $(1,1)$ \\
    	\hline  
    	$D_1 $& $(1,1)$& $(2,2)$\\
    	\hline
    \end{tabular}
    \label{tbl:G3}
    \end{table}
	
	The three-player NFG $\Gamma_1$ in Table IV shows a joint strategy $(H_1,H_2,H_3)$ that is \sNE{}, but not \srp{} (refer to Example III.5
	). Using the just proven (1) and (2), we get that $(H_1,H_2,H_3)$ is also not \SR{} nor \srs{}. 
	
	The three-player game $\Gamma_2$ (Table V) shows a strategy $(H_1,H_2,H_3)$ that is \srp{}, but not \SR{} (see Example III.5) and thus also not \srs{} (property (1)).
	It is, however, \sNE{}: The only relevant deviation from  $(H_1,H_2,H_3)$ is $(D_1, H_2, D_3)$, as it yields a different utility $(3,0,-2)$ instead of $(1,1,1)$. While player $P_1$ profits in this case, player $P_3$ does not. One deviating player not profiting suffices for a strong Nash Equilibrium, thus $(H_1,H_2,H_3)$ is \sNE{}.
	
	To prove the remaining implications incorrect, we consider the two-player game $\Gamma_3$ in \Cref{tbl:G3}. We can easily see that $(H_1,H_2)$ is not \srs, nor \sNE{}. This is the case, as all players $\{P_1,P_2\}$ can deviate to play $(D_1,D_2)$ which yields a strict increase for both. However, since no player profits from deviating alone, $(H_1,H_2)$ is still \SR{} and \srp.

\end{proof}

\subsection{Results of Security Analysis and Their Proofs} \label{app:sec}
In this section all omitted proofs of the results from Section V are provided. Additionally, the  results  \Cref{thm:bleqf}, \Cref{thm:a0} and \Cref{thm:b0} about edge cases are stated and proven.\\

\newtheorem*{WI}{Theorem~V.1}
\begin{WI}[Weak Immunity of Honest Behavior -- (P1)]
 The terminal histories $(H)$ of honest unilateral closing, and $(C_h,S)$ of honest collaborative closing of $G_c(A)$ are weak immune, if the channel balances are higher than the fee required in a revocation transaction, that is if $a\geq f$ and $b\geq f$.
\end{WI}
\begin{table}[t]
	\renewcommand*{\arraystretch}{1.2}
	\caption{Overview of Implications and Counterexamples.}
	\centering
	\begin{tabular}{| c | c |c |c | c |}
		\hline
			$\to$ & \SR{}& \srs & \sNE{} & \srp \\
		\hline
			\SR{}   &  \diagbox{\quad}{\quad} & $\Gamma_3$ & $\Gamma_3$ & $\checkmark$ \\
		\hline
			\srs  & $\checkmark$ & \diagbox{\quad}{\quad} & $\checkmark$ & $\checkmark$ \\
		\hline
			\sNE{} & $\Gamma_1$ & $\Gamma_1$&\diagbox{\quad}{\quad} & $\Gamma_1$ \\
		\hline
			\srp   & $\Gamma_2$ & $\Gamma_2$ & $\Gamma_3$&  \diagbox{\quad}{\quad}\\
		\hline
	\end{tabular}
	\label{tbl:counterex}
\end{table}

\begin{proof}
Let $a,b\geq f$. For history $(H)$, we consider any strategy $\sigma$, where $A$ chooses $H$ after the empty history $\emptyset$, $B$ chooses $S$ after $(C_h)$, $P$ after $(D)$ and $H$ after $(C_c)$. Such a strategy $\sigma$ yields terminal history $(H)$. If we can show that $\sigma$ is weak immune, also history $(H)$ is weak immune by Definition III.11.
Assume, player $A$ honestly follows $\sigma$, i.e., choosing $(H)$, then $B$'s deviation from $\sigma$ cannot affect the outcome. Thus, $A$'s utility remains non-negative. The other way around, if $B$ follows $\sigma$, $A$ can deviate to any initial action $C_h$, $D$ or $C_c$, player $B$'s utility never drops below 0, by following strategy $\sigma$, as $a\geq f$. Since honest players cannot get negative utility, $\sigma$ is weak immune.

    Similarly, for  $(C_h,S)$, we consider any strategy $\sigma'$, where $A$ chooses $C_h$ initially, player $B$ chooses $S$ after $(C_h)$, $P$ after $(D)$ and $H$ after $(C_c)$. Further, player $A$ takes $P$ after $(C_h,D)$ and $H$ after $(C_h,\mathfrak{I})$, $(C_h,U^+)$ and $(C_h,U^-)$. This strategy  $\sigma'$ yields terminal history $(C_h,S)$.  Deviation of $A$ has the same effects as before, never causing the honest $B$, who follows $\sigma$, negative utility. If $B$ deviates now to one of $U^+$, $U^-$, $\mathfrak{I}$, $D$, or $H$, honest $A$, following $\sigma$,  also never gets negative utility, since $b \geq f$. Therefore, $\sigma'$ and hence history $(C_h,S)$ are weak immune.  
\end{proof}

\newtheorem*{IC}{Theorem~V.2}
\begin{IC}[Incentive-Compatibility -- (P2)]
\noindent If $a-p_B+d_A \geq f$ and $b-p_A+d_B \geq f$, then \begin{enumerate}
    \item honest unilateral closing $(H)$ is \srp, but \emph{not} practical.
    \item honest collaborative closing $(C_h,S)$ is \srp. It is practical iff $c \neq p_A$.
\end{enumerate}
\end{IC}

\begin{proof} Let us first prove collusion resilience \srp{} of $(H)$ and $(C_h,S)$.
As in the previous proof, we only have to find a strategy $\sigma$ that yields history $(H)$ and another strategy $\sigma'$ yielding history $(C_h,S)$, that are \srp{}, to prove $(H)$ and $(C_h,S)$ \srp{}, as defined in Definition III.11.
Additionally, collusion resilience is defined on strict subsets of players. Thus, in a two-player game, it considers only deviations of single players and since the summation over one value is the value itself, \srp{} is equivalent to being a Nash Equilibrium in this case.
We therefore only have to check whether $\sigma$ and $\sigma'$ are Nash Equilibria. 

For $(H)$, we consider a strategy $\sigma$, where player $A$ chooses $H$ initially, player $B$ chooses $\mathfrak{I}$ after $(C_h)$, and $\mathfrak{I}$ after $(C_c)$. Additionally, player $B$ always chooses $P$ after a history $(...,D)$, where the last action was $D$. Player $A$ takes action $H$ after $(C_h,\mathfrak{I})$ and $(C_c,\mathfrak{I})$. Further, $B$ takes action $\mathfrak{I}$ after $(C_{h/c},\mathfrak{I},U^{+/-})$ (subgames $S_1$, $S_2$, $S'_1$, $S'_2$ in \Cref{tbl:S1p}). For $A$, we finally assume she takes action $H$ after $(C_{h/c},\mathfrak{I}, U^{+/-},\mathfrak{I})$. This strategy yields history $(H)$. Deviations from $\sigma$ of player $B$ cannot change the utility, hence in particular cannot increase his utility.
Let us consider deviations of player $A$. A deviation to $D$ at any point in the game, leads to $A$ losing all her funds $a$, which is a strict decrease in utility. This is the case because in $\sigma$ player $B$ always chooses $P$ after $D$. Therefore this option is not a threat. If $A$ deviates to $C_h$ or $C_c$ initially, we end up in $(C_h,\mathfrak{I})$, $(C_c,\mathfrak{I})$ respectively. Closing honestly (action $H$) here leads to the same utility as not deviating. Also a deviation to $\mathfrak{I}$ does not lead to a better utility. The options she has left is taking $U^+$ or $U^-$. Either way, $B$ takes $\mathfrak{I}$ and leaves $A$ similar choices to before: action $H$ or action $\mathfrak{I}$, both of which do not yield a better utility for her. 
Since no player can increase their utility by deviating from $\sigma$, it is a Nash Equilibrium, and hence $(H)$ is too.

To show that $(C_h,S)$ is a Nash Equilibrium, we consider a strategy $\sigma'$, where $A$ picks $C_h$ initially, $B$ chooses $S$ after $(C_h)$, $P$ after $(D)$ and $H$ after $(C_c)$. Further, let $A$ pick $P$ after $(C_h,D)$, $H$ after $(C_h,\mathfrak{I})$ and $(C_h,U^{+/-})$ (subgames $S_3$, $S_4$ in \Cref{tbl:S4p}). This strategy $\sigma'$ has terminal history $(C_h,S)$.
A deviation of player $B$, results in either the same utility (choosing action $\mathfrak{I}$, $U^+$, or $U^-$ after $(C_h)$ and having $A$ taking $H$) or in strictly worse utility (choosing $H$ or $D$, where $A$ takes $P$). Every other deviation has no impact on the resulting history. Similarly, player $A$ cannot profit from deviating. Choosing action $H$ or $D$ initially, leads to a strict loss, as $B$ plays $P$, whereas taking action $C_c$ yields the same utility for $A$ (as $B$ will take action $H$). Every other deviation has no impact on the history. Hence, no player can increase their utility by deviating, which makes $\sigma'$ and therefore $(C_h,S)$ a Nash Equilibrium.

To prove the practicality properties, we compute all subgame perfect equilibria of $G_c(A)$. We compute subgame perfect equilibria bottom-up. That is, we start comparing the utility of subtrees with leaves only. In $G_c(A)$, these are for example the subgames after history $(C_h,\mathfrak{I},D)$ or $(C_c,D)$. For the latter, $A$ is the player to choose the action. To compute the subgame perfect equilibrium, we have to compare all possible utilities for $A$ after $(C_c,D)$. We then replace this internal node labelled $A$, by the utility that yields the best value for $A$ and proceed until we reach the root. If there is no single best choice for a player, then all actions resulting in best utility have to be considered. Applying this procedure to the subgames $S_1$-$S_4$ and $S'_1$-$S'_4$ we get subgame perfect terminal history  $(\mathfrak{A},H)$ with utility $(\rho+\alpha-\epsilon,\rho+\alpha)$ for $S_1$. For $S_2$ we get terminal history $(S)$ yielding $(\alpha,\alpha)$ and $(\mathfrak{I},H)$, yielding $(\alpha-\epsilon,\alpha)$. For $S_3$ and $S_4$ it is $(\mathfrak{I},S)$ with $(\alpha,\alpha)$. The subgame $S'_1$ has practical history $(\mathfrak{I},H)$, with $(\alpha-\epsilon,\alpha)$ if $c>p_A$, $(\mathfrak{A},\mathfrak{I},S)$ with $(\rho+\alpha,\rho+\alpha)$ if $c=p_A$ and $(\mathfrak{A},H)$ with $(\rho+\alpha-\epsilon,\rho+\alpha)$ if $c<p$. The subgame $S'_2$ has practical history $(\mathfrak{I},H)$, yielding $(\alpha-\epsilon,\alpha)$. For $S'_3$ and $S'_4$ in \Cref{tbl:S4p} we get $(\mathfrak{I},H)$ with $(\alpha,\alpha-\epsilon)$ and additionally for $S'_3$, if $c=p$, we also have $(\mathfrak{A},S)$ yielding $(\rho+\alpha,\rho+\alpha)$. All of these results are based on the facts $a-p_B+d_A\geq f$ and $b-p_A+d_B \geq f$, since this causes the revocation transaction always to be better than ignoring the dishonest unilateral closing attempt.

Based on these preliminary results, we can now compute the subgame perfect equilibria for $G_c(A)$ considering multiple practical histories and case splits as stated:
If $c=p_A$, then $(C_c,U^+,\mathfrak{A},S)$ and $(C_c,\mathfrak{I},U^+,\mathfrak{A},\mathfrak{I},S)$ are practical, both yielding $(\rho+\alpha,\rho+\alpha)$. If $c>p_A$, then the histories $(C_h,S)$, $(C_h,U^+,\mathfrak{I},S)$, $(C_h,U^-,\mathfrak{I},S)$ and $(C_h,\mathfrak{I},U^-,S)$ all leading to $(\alpha,\alpha)$ are practical, as well as terminal history $(C_h,\mathfrak{I},U^+,\mathfrak{I},H)$, yielding $(\rho+\alpha-\epsilon,\rho+\alpha)$.
For $c<p_A$, all the histories and their utilities from $c>p_A$ are practical. Additionally $(C_c,\mathfrak{I},U^+,\mathfrak{A},H)$ is subgame perfect in this case 
and also results in utility $(\rho+\alpha-\epsilon,\rho+\alpha)$.

This shows, that $(H)$ is never practical and $(C_h,S)$ is practical if and only if $c\neq p_A$. 
\end{proof}

\subsubsection{Results without Updates}

\newtheorem*{COR}{Corollary~1}
\begin{COR}
If there exists an old channel state $(a+d_A,b-d_A)$, with $a+d_A <f$, then neither history $(H)$ nor $(C_h,S)$ is  weak immune nor practical, but \srp.
\end{COR}

\begin{proof}
We fix the old distribution state such that the difference $d_A$ to the latest state is the value of $A$'s dishonest closing attempt in the closing game. As $a+d_A<f$ implies $a<f$, Theorem V.5 applies. Therefore, neither $(H)$ nor $(C_h,S)$ are weak immune.

In order to show that they are also not practical, we prove instead, that the only practical history is $(D,\mathfrak{I})$. Since $a+d_A<f$, $\mathfrak{I}$ is the best choice for $B$ after $(D)$, $(C_h,\mathfrak{I},D)$ and $(C_c,\mathfrak{I},D)$. Consequently, $A$ will choose $D$ after $(C_h,\mathfrak{I})$ and $(C_c,\mathfrak{I})$. If now $b+d_B\geq f$, then $A$'s best choice is $P$ after $(C_h,D)$ and $(C_c,D)$. Thus, $B$ will take $S$ after $(C_h)$ and $H$ after $(C_c)$. In the other case, $b+d_b<f$, $A$'s best option is $\mathfrak{I}$ after $(C_h,D)$ and $(C_c,D)$, thus $B$'s best choice after $(C_h)$ and $(C_c)$ is $D$, which yields a negative utility for $A$. Therefore, in both cases $A$'s only subgame perfect action is $D$. Hence, $(D,\mathfrak{I})$ is the unique subgame perfect history.

For \srp, we show instead that there exist extensions (Definition III.11) $\sigma$ of $(H)$ and $\sigma'$ of $(C_h,S)$ that are Nash Equilibria. Let $\sigma$ be the strategy, where $A$ chooses $H$, everyone chooses $P$ after a dishonest closing attempt $D$, $B$ chooses $\mathfrak{I}$ after $(C_h)$ and  $(C_c)$ and $A$ chooses $H$ after $(C_h,\mathfrak{I})$ and $(C_c,\mathfrak{I})$.
Then, player $B$'s deviations have no impact, thus cannot not increase his utility, and player $A$'s deviations either lead to the same utility as $\sigma$, or to the strictly worse utility $-a$. Anyway, no player can  deviate to increase their utility and therefore $\sigma$ and thus $(H)$ is \srp{}.
To prove $(C_h,S)$ is \srp, we consider the strategy $\sigma'$, which is the same as $\sigma$, except $A$ initially chooses $C_h$ and $B$ chooses $S$ after $(C_h)$. A deviation of $A$ either leads to utility $\alpha-\epsilon$ for her, which is worse than $\sigma$'s utility, or to utility $-a$, which is even worse. For $B$, a deviation either leads to the same utility $\alpha$ (taking $\mathfrak{I}$ after $(C_h)$), to a slightly worse $\alpha-\epsilon$ (choosing $H$ after $(C_h)$) or to the way worse $-b$ ($D$ after $(C_h)$). Every other deviation has no impact on the history. Hence, as nobody profits from deviating, $\sigma'$ is also a Nash Equilibrium.
\end{proof}

We present an additional theorem, discussing the case where player $B$ has little funds left in the channel. Since the roles of player $A$ and $B$ are arbitrary, it is of little importance because the results give stronger security guarantees as for the case where $A$ has a low balance. Nevertheless, we state it for the sake of completeness.
 
\begin{theorem} \label{thm:bleqf}
If there exists an old state with $b+d_B<f$, but $a \geq f$, then \begin{enumerate}
        \item $(H)$ is secure.
        \item $(C_h,S)$ is not practical, not weak immune, but \srp.
    \end{enumerate}
\end{theorem}
\begin{proof}
To prove (1), we start by showing weak immunity for a strategy $\sigma$ with history $(H)$. Consider $\sigma$, where $A$ takes action $H$ initially, player $B$ chooses $P$ after $(D)$, $S$ after $(C_h)$ and $H$ after $(C_c)$. Then $\sigma$ and thus $(H)$ is weak immune, because $B$'s deviations have no impact on the history and $A$'s deviations can never bring $B$'s utility below zero.

Next, we prove the practicality of $(H)$ by computing all subgame perfect equilibria. Since $a \geq f$, the subgame perfect choice after $(D)$, $(C_h,\mathfrak{I},D)$ and $(C_c,\mathfrak{I},D)$ is $P$. Thus, $A$ chooses $H$ after $(C_h,\mathfrak{I})$ and $(C_c,\mathfrak{I})$. Due to $b+d_B<f$, $A$'s best option after $(C_h,D)$ and $(C_c,D)$ is $\mathfrak{I}$. Hence $B$'s unique subgame perfect choice after $(C_c)$ and $(C_h)$ is $D$. Thus, $A$'s only best response is $H$. Therefore, $(H)$ is the only practical history.
As practicality implies \srp{} in our case, $(H)$ is secure.

For (2), we just showed that $(C_h,S)$ cannot be practical. Additionally, $(C_h,S)$ is not weak immune, since $B$ could deviate to $D$ after $(C_h)$, in which case $A$ gets negative utility for sure, because of $b+d_B<f$.

Finally, we consider the strategy $\sigma'$, with history $(C_h,S)$, where $A$ initially  chooses $C_h$, $B$ chooses $S$ after $(C_h)$, both take $P$ in case of a dishonest unilateral closing attempt $D$, $B$ takes $H$ after $(C_c)$, similarly $A$ takes $H$ after $(C_h,\mathfrak{I})$ and $(C_c,\mathfrak{I})$. Using similar argumentation as before, we conclude that any deviation of a player leads a utility as most as good as $\sigma'$ for them, but never better. Hence, $\sigma'$ is a Nash Equilibrium yielding terminal history $(C_h,S)$.
\end{proof}

\subsubsection{Results for Edge Cases}
So far, we only considered cases where both balances $a$ and $b$ were strictly greater than zero. This is not necessarily the case. Therefore, we consider these cases here. 
In the first case, $a=0$, $B$ cannot close dishonestly, as there is no old state that increases his balance. The corresponding simplified game is presented in \Cref{tbl:a0}.

\begin{figure}
	\centering
	\begin{tikzpicture}[scale=0.9,->,>=stealth',auto,node distance=2cm, el/.style = {inner sep=2pt, align=left, sloped}]
		\node (1) at (0.125,2) {\color{teal}$A$};
		
		\node (2) at (-3.5,-0.5) {\color{olive} $B$} ; 
		\node (3) at (-2.5,1.25) { $(0,\alpha)$};
		\node (4) at (0,0.75) {\color{olive}$B$};
		\node (5) at (3.5,-0.5) {\color{olive}$B$};
		
		\node (6) at (-3.5,-3.25) {\color{teal}$A$};
		\node (7) at (-2.5,-2.5) { $(0,\alpha)$};
		\node (8) at (1.1,-0.25) { $(0,-f+\alpha)$};
		\node (9) at (-0.25,-0.75) { $(d_A+\alpha-\epsilon, -d_A+\alpha)$};
		\node (10) at (2,-2.5) {$(c+\alpha,-c+\alpha)$};
		\node (11) at (3.5,-3.25) {\color{teal}$A$};
		
		\node (12) at (-3.5,-5.25) {\color{olive}$B$};
		\node (13) at (-1,-4) {$(0,\alpha)$};
		\node (14) at (-2.25,-4.5) {$(0,-b)$};

		\node (19) at (2.25,-4.5) {$(0,-b)$};
		\node (20) at (1,-4) { $(0,\alpha)$};
		\node (21) at (3.5,-5.25) {\color{olive}$B$};
		
		\node (22) at (-1,-6) {$(0,-f+\alpha)$} ; 
		\node (23) at (-2,-6.5) {$(d_A+\alpha-\epsilon,-d_A+\alpha)$};
		\node (24) at (1,-6) { $(0,-f+\alpha)$};
		\node (25) at (2,-6.5) {$(d_A+\alpha-\epsilon,-d_A+\alpha)$};	
		
		\node (27) at (-1.25,-2) {$(0,\alpha-\epsilon)$};
		\node (34) at (1.25,-2) {$(0,\alpha-\epsilon)$};

		{\color{teal}
		\path (1) edge node [below, pos=0.15] {\tiny $C_h$}  (2);
		\path (1) edge node [below, pos=0.5] {\tiny $H$}  (3); 
		\path (1) edge node [right, pos=0.6] {\tiny $D$}  (4);
		\path (1) edge node [above, pos=0.25] {\tiny $C_c$}   (5);}
		
		{\color{olive}
		\path (2) edge node [right] {\tiny $\mathfrak{I}$}   (6);
		\path (2) edge node [right, pos=0.55] {\tiny $S$}   (7); 
		\path (4) edge node [above] {\tiny $P$}   (8);
		\path (4) edge node [left] {\tiny $\mathfrak{I}$}   (9);
		\path (5) edge node [left] {\tiny $S$}   (10);
		\path (5) edge node [left] {\tiny $\mathfrak{I}$}   (11); }
		
		{\color{teal}
		\path (6) edge node [right, pos=0.5] {\tiny $D$}   (12);
		\path (6) edge node [above, pos=0.5] {\tiny $H$}   (13); 
		\path (6) edge  node [above, pos=0.7] {\tiny $\mathfrak{I}$}  (14);
		\path (11) edge node [above, pos=0.7] {\tiny $\mathfrak{I}$}   (19); 
		\path (11) edge node [above, pos=0.5] {\tiny $H$}   (20);
		\path (11) edge node [left, pos=0.5] {\tiny $D$}   (21);}
		
		{\color{olive}
		\path (12) edge node [above] {\tiny $P$}   (22);
		\path (12) edge node [left] {\tiny $\mathfrak{I}$}   (23); 
		\path (21) edge node [above] {\tiny $P$}   (24);
		\path (21) edge node [right] {\tiny $\mathfrak{I}$}   (25); }
		
		{\color{olive}
		\path (2) edge node [right, pos=0.5] {\tiny $H$}   (27); 
		\path (5) edge node [left, pos=0.5] {\tiny $H$}   (34);}
		
	\end{tikzpicture}
\caption{Closing game $G_c(A)$ with $a=0$.}
\label{tbl:a0}
\end{figure}
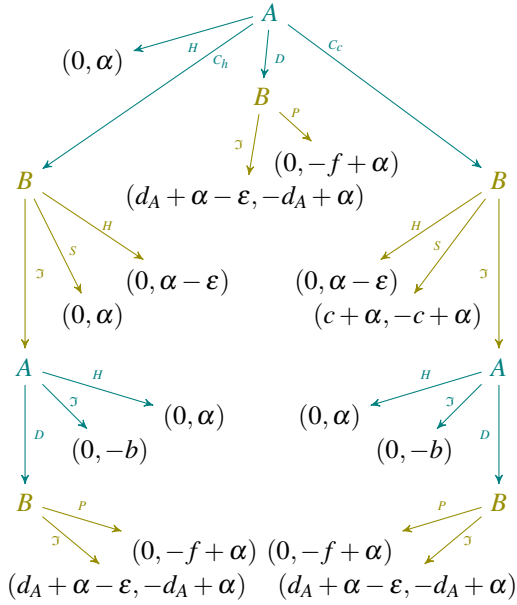

If $b=0$ (\Cref{tbl:b0}), player $A$ cannot close dishonestly, as she cannot take any money from $B$. Thus, both dishonest unilateral closing $D$ and proposing an unfair split in a collaborative closing attempt $C_c$ are not possible.

\begin{figure}
	\centering
	\begin{tikzpicture}[scale=0.8,->,>=stealth',auto,node distance=2cm, el/.style = {inner sep=2pt, align=left, sloped}]
		\node (1) at (-1.75,0) {\color{teal}$A$};
		
		\node (2) at (-3.5,-1) {\color{olive} $B$} ; 
		\node (3) at (-0,-1) { $(\alpha-\epsilon,0)$};
		
		\node (6) at (-5.75,-2.5) {\color{teal}$A$};
		\node (7) at (-2.75,-2.5) { $(\alpha,0)$};
	
		\node (13) at (-5.1,-3.75) {$(\alpha-\epsilon,0)$};
		\node (14) at (-7,-3.75) {$(-a,0)$};
		
		\node (27) at (-4,-2.5) {$(\alpha,0)$};
		\node (30) at (-1.25,-2.5) {\color{teal}$A$};
		
		\node (36) at (-2.75,-3.75) { $(-f+\alpha,0)$} ; 
		\node (37) at (0.75,-3.75) { $(-d_B+\alpha,d_B+\alpha-\epsilon)$};
		
		{\color{teal}
		\path (1) edge node [above, pos=0.5] {\tiny $C_h$}  (2);
		\path (1) edge node [above, pos=0.5] {\tiny $H$}  (3); }
		
		{\color{olive}
		\path (2) edge node [above] {\tiny $\mathfrak{I}$}   (6);
		\path (2) edge node [right] {\tiny $S$}   (7);  }
		
		{\color{teal}
		\path (6) edge node [right, pos=0.3] {\tiny $H$}   (13); 
		\path (6) edge  node [left, pos=0.3] {\tiny $\mathfrak{I}$}  (14);}
		
		{\color{olive}
		\path (2) edge node [left, pos=0.5] {\tiny $H$}   (27); 
		\path (2) edge node [above] {\tiny $D$}   (30);}
		
		{\color{teal}
		\path (30) edge node [above] {\tiny $P$}   (36);
		\path (30) edge node [above] {\tiny $\mathfrak{I}$}   (37); }
	\end{tikzpicture}
\caption{Closing game $G_c(A)$ with $b=0$.}
\label{tbl:b0}
\end{figure}
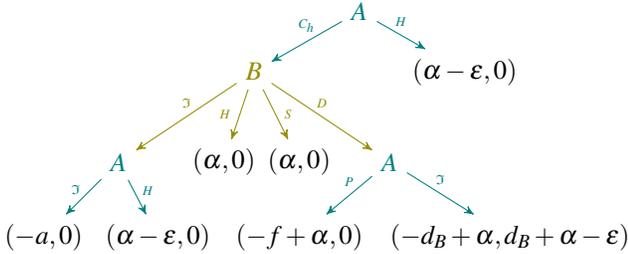

Finally, we present results about the two edge cases.

\begin{theorem} \label{thm:a0}
 If  $a=0$ and $b>0$ (\Cref{tbl:a0}), then only histories that involve an explicit cheating attempt are weak immune. Additionally,  $(H)$ and $(C_h,S)$ are practical if and only if $d_A\geq f$ in every previous state $(d_A,b-d_A)$. In any case they are \srp.
  \end{theorem}  
    \begin{proof}
    We first show that only histories that involve an explicit cheating attempt can be weak immune. Let us consider a history $h$ without a $D$ or $C_c$ action, then $A$ does not initially choose $D$ in $h$. However, if $A$ deviates to $D$, then $B$'s utility is negative. Thus, any such history $h$ is not weak immune.

To show both $(H)$ and $(C_h,S)$ are \srp, it suffices to show they are Nash Equilibria as before. We therefore consider any strategy $\sigma$, where $A$ initially chooses $H$, player $B$ chooses $P$ after $(D)$, $S$ after $(C_h)$ and $H$ after $(C_c)$. Further, player $A$ takes action $H$ after $(C_h,\mathfrak{I})$. The strategy $\sigma$ yields history $(H)$. 
No matter how player $A$ deviates, she always gets utility $0$, as she does in $\sigma$. Thus, she has no incentive to deviate. Since player $B$'s deviations cannot change the history, also he has no incentive to do so. Therefore, $\sigma$ and hence $(H)$ is a Nash Equilibrium.
Adapting $\sigma$, by making $A$ first choice $C_h$ we get strategy $\sigma'$ which leads to history $(C_h,S)$. As before, $A$'s utility stays 0 no matter how she deviates from $\sigma'$. Also player $B$ cannot improve his utility by changing strategy. Hence, also $\sigma'$ and therefore $(C_h,S)$ is a Nash Equilibrium.

Towards practicality, we now compute all subgame perfect equilibria. Let $d_A\geq f$. In which case $P$ is the subgame best choice for $B$ after $(D)$, $(C_h,\mathfrak{I},D)$ and $(C_c,\mathfrak{I},D)$. Further, after history $(C_c,\mathfrak{I})$, $S$ it is never a best option for $B$, because it is strictly dominated by $H$. Therefore, $A$ will get utility zero in any case. This makes $(H)$ a practical history. Similarly for $(C_h,S)$, since $S$ is subgame perfect for $B$ after $(C_h)$.

If now $d_A<f$, then $\mathfrak{I}$ is subgame perfect for $B$ after $D$. Thus, with similar argumentation as before, $(D,\mathfrak{I})$ is the only practical history.
\end{proof}

\begin{theorem} \label{thm:b0}
  If  $a>0$ and $b=0$ (\Cref{tbl:b0}), then
    \begin{enumerate}
        \item $(H)$ is secure.
        \item $(C_h,S)$ is not weak immune, but \srp. It is practical iff $d_B \geq f$ in every previous state $(a-d_B,d_B)$.
    \end{enumerate}
\end{theorem}
\begin{proof}
We prove (1.) first. The history $(H)$ is weak immune, as $B$'s strategy does not effect the history, since $A$'s initial choice has to be $H$. Further, $A$'s deviation is irrelevant for $B$, as he can never get negative utility in this game.

Practicality of $(H)$. We compute subgame perfect equilibria. After history $(C_h,D)$ the subgame perfect choice of $A$ depends on whether $d_B \geq f$. In any case, $D$ is subgame perfect for $B$ after history $(C_h)$. If $A$ chose $P$, then it is as good as any other choice, yielding 0, otherwise it is the only best option resulting in a positive utility. Thus, $A$ either gets $-f+\alpha$ or $-d_B+\alpha$ if she chooses $C_h$, both of which is negative. Hence $A$'s subgame perfect and therefore practical choice is $H$, yielding the history $(H)$.

The fact that $(H)$ is \srp{} follows from practicality. This shows that $(H)$ is secure, if $b=0$.

For (2), we start showing $(C_h,S)$ is not weak immune. We consider any strategy $\sigma'$ yielding the history $(C_h,S)$. Assume now, $B$ deviates to $D$ after $(C_h)$, then no matter what $A$'s choice is, she will get a negative utility, thus $(C_h,S)$ is not weak immune.

The collusion resilience of $(C_h,S)$, can be shown by considering a strategy $\sigma'$ with history $(C_h,S)$, where we additionally fix that $A$ chooses $P$ after $(C_h,D)$. Then $B$ has no incentive to deviate as he always gets utility 0, and $A$ has no incentive as $\alpha$, which is her utility in $\sigma'$, is the best possible outcome for her.

To finally show that $(C_h,S)$ is practical iff $d_B\geq f$, we consider $A$'s choice after $(C_h,D)$. The option $P$ is subgame perfect iff $d_B \geq f$. Thus, $S$ is subgame perfect for $B$ iff $d_B \geq f$. For $d_B < f$, $D$ is the better option for $B$, yielding $(-d_B+\alpha,d_B+\alpha-\epsilon)$. Therefore $C_h$ is subgame perfect for $A$ iff $d_B \geq f$, in which case the resulting history is $(C_h,S)$. 
\end{proof}

The weak immunity result of $(H)$  might be misleading, as $B$ can actually close dishonestly immediately (before $A$ takes action). This is not represented here, but in $G_c(B)$, which is analog to $G_c(A)$ but with swapped roles.

    \subsection{Subgames of the Closing Game} \label{app:subgames}
\begin{figure}
	\centering
	\begin{tikzpicture}[->,>=stealth',auto,node distance=2cm, el/.style = {inner sep=2pt, align=left, sloped}]
		\node (1) at (1.75,0.75) {\color{olive}$B$};
		
		\node (20) at (1,-2.25) {\color{teal} $A$};
		\node (2) at (1.75,-2.75) {\color{teal}$A$} ; 
		\node (3) at (-1.,-0.25) {\color{teal}$A$};
		\node (21) at (-1,0.5) {$(\alpha,\alpha-\epsilon)$};
		\node (24) at (-0.15,-1.25) {$({\color{red}y}+\alpha,{\color{red}-y}+\alpha)$};
		
		\node (22) at (-2.,-3.25) {$(b-f+\alpha,-b)$};
		\node (23) at (-3.5,-2.75) {$(-d_B+\alpha,d_B+\alpha-\epsilon)$};
				
		\node (8) at (-3.25,-2) {$(-a,-b)$};
		\node (9) at (-1.75,-2.) {$(\alpha-\epsilon,\alpha)$};
		\node (10) at (-2.25,-0.5) {\color{olive}$B$};
		
		\node (5) at (1.75,-5.75) {\color{olive}$B$};
		\node (6) at (-1,-4) {$(\rho+\alpha-\epsilon,\rho+\alpha)$};
		\node (7) at (1,-4.25) {\color{olive}$B$};
		
		\node (11) at (-1.25,-5.5) {$(-a{\color{red}-x}+\rho,a{\color{red}+x}-f+\rho+\alpha)$};
		\node (12) at (-3.25,-5) {$(d_A+\rho+\alpha-\epsilon,-d_A+\rho+\alpha)$};
		
		\node (13) at (-5,-0.75) {$(-a,a-f+\alpha)$};
		\node (14) at (-4.35,-1.25) {$(d_A+\alpha-\epsilon,-d_A+\alpha)$};
		
		\node (15) at (1.75,-8.25) {\color{teal}$A$};
		\node (16) at (-0.5,-8) {$(\rho+\alpha,\rho+\alpha-\epsilon)$};
		\node (17) at (-2.5,-7.5) {$(-a{\color{red}-x}+\rho,-b{\color{red}+x}+\rho)$};
		\node (4) at (-3.75,-6.75) {$({\color{red}y-x}+\rho+\alpha,{\color{red}-y+x}+\rho+\alpha)$};
		
		\node (18) at (-3,-9) {$(b{\color{red}-x}-f+\rho+\alpha,-b{\color{red}+x}+\rho)$};
		\node (19) at (-0.5,-9.5) {$(-d_B+\rho+\alpha,d_B+\rho+\alpha-\epsilon)$};
		
		{\color{olive}
		\path (1) edge node [left, pos=0.5] {\tiny $A$}  (2);
		\path (1) edge node [above, pos=0.5] {\tiny $\mathfrak{I}$}  (3); 
		\path (1) edge node [left]  {\tiny $D$} (20);
		\path (1) edge node [above]  {\tiny $H$} (21);
		\path (1) edge node [left] {\tiny $S$} (24); }
		
		{\color{teal}
		\path (2) edge node [left, pos=0.3] {\tiny $\mathfrak{I}$}  (5);
		\path (2) edge node [above] {\tiny $H$}  (6);
		\path (2) edge node [left] {\tiny $D$}  (7);
		\path (3) edge node [above, pos=0.5] {\tiny $\mathfrak{I}$}  (8);

		\path (3) edge node [left] {\tiny $H$}  (9);
		\path (3) edge node [above] {\tiny $D$}  (10); 
		
		\path (20) edge node [below] {\tiny $\mathfrak{I}$} (23); 
		\path (20) edge node [below] {\tiny $P$} (22);}
		
		{\color{olive}
		\path (7) edge node [below] {\tiny $P$}  (11); 
		\path (7) edge node [below] {\tiny $\mathfrak{I}$}  (12);
		
		\path (10) edge node [above, pos=0.7] {\tiny $P$}  (13);
		\path (10) edge node [below] {\tiny $\mathfrak{I}$}  (14);
		
		\path (5) edge node [left] {\tiny $D$}  (15);
		\path (5) edge node [above] {\tiny $H$}  (16); 
		\path (5) edge node [above] {\tiny $\mathfrak{I}$}  (17);
		\path (5) edge node [above] {\tiny $S$}  (4);}
		
		{\color{teal}
		\path (15) edge node [below] {\tiny $P$}  (18);
		\path (15) edge node [below] {\tiny $\mathfrak{I}$}  (19);}
	\end{tikzpicture}
	\caption{Subgames $S_{1,2}$, $S'_{1,2}$ with Update $(a,b)\mapsto (a+x,b-x)$.}
	\label{tbl:S1p}
\end{figure}
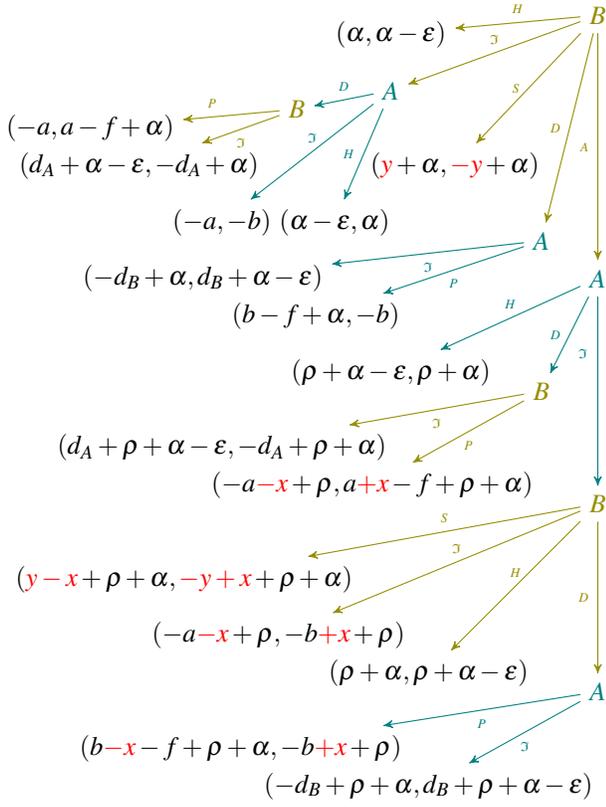
In the following all the subgames needed for the closing game $G_c(A)$ are defined. The subgames $S_{1,2}$ and $S'_{1,2}$ in \Cref{tbl:S1p} cover the case where a channel update is proposed by $A$, although $A$ has already signed a collaborative closing attempt. In $S_1$ the closing attempt was honest, hence $y=0$ and the update is from channel state $(a,b)$ to $(a+p_A,b-p_A)$, hence $x=p_A$. In $S_2$ also $y=0$ the suggested update is $(a-p_B, b+p_B)$, thus $x=-p_B$. In $S'_{1,2}$ the closing attempt was dishonest, therefore $y=c$. The channel updates are as before, thus $x=p_A$ for $S'_1$ and $x=-p_B$ for $S'_2$.
Similarly, subgames $S_{3,4}$ and $S'_{3,4}$ in \Cref{tbl:S4p} cover the case where a channel update is proposed by $B$, although $A$ has already signed a collaborative closing attempt. As in the first case, we have $y=0$ for the honest closing attempt in  $S_{3,4}$ and $y=c$ for dishonest collaborative closing in $S'_{3,4}$. Further in $S_3$ and $S'_3$, the proposed update is $(a+p_A,b-p_A)$, hence $x=p_A$, whereas in $S_4$ and $S'_4$ it is $(a-p_B, b+p_B)$, thus $x=-p_B$.

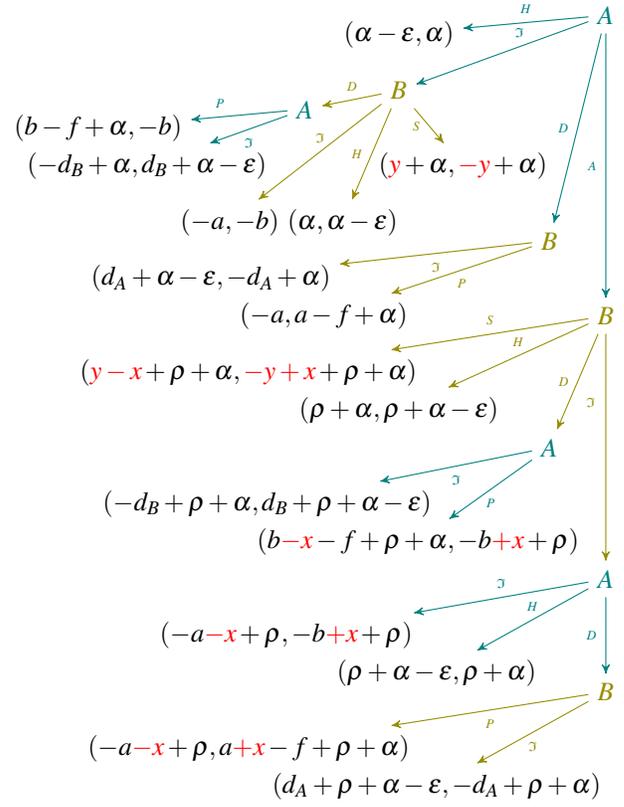
\begin{figure}
	\centering
	\begin{tikzpicture}[->,>=stealth',auto,node distance=2cm, el/.style = {inner sep=2pt, align=left, sloped}]
		\node (1) at (1.75,0.75) {\color{teal}$A$};
		
		\node (20) at (1,-2.25) {\color{olive} $B$};
		\node (2) at (1.75,-3.25) {\color{olive}$B$} ; 
		\node (3) at (-1.,-0.25) {\color{olive}$B$};
		\node (21) at (-1,0.5) {$(\alpha-\epsilon,\alpha)$};
		
		\node (22) at (-2.,-3.25) {$(-a,a-f+\alpha)$};
		\node (23) at (-3.5,-2.75) {$(d_A+\alpha-\epsilon,-d_A+\alpha)$};
				
		\node (8) at (-3.25,-2) {$(-a,-b)$};
		\node (24) at (-0.15,-1.25) {$({\color{red}y}+\alpha,{\color{red}-y}+\alpha)$};
		\node (9) at (-1.75,-2.) {$(\alpha,\alpha-\epsilon)$};
		\node (10) at (-2.25,-0.5) {\color{teal}$A$};
		
		\node (4) at (-3,-4) {$({\color{red}y-x}+\rho+\alpha,{\color{red}-y+x}+\rho+\alpha)$};
		\node (5) at (1.75,-6.75) {\color{teal}$A$};
		\node (6) at (-1,-4.5) {$(\rho+\alpha,\rho+\alpha-\epsilon)$};
		\node (7) at (1,-5.) {\color{teal}$A$};

		\node (11) at (-0.75,-6.25) {$(b{\color{red}-x}-f+\rho+\alpha,-b{\color{red}+x}+\rho)$};
		\node (12) at (-2.75,-5.75) {$(-d_B+\rho+\alpha,d_B+\rho+\alpha-\epsilon)$};
		
		\node (13) at (-5,-0.75) {$(b-f+\alpha,-b)$};
		\node (14) at (-4.35,-1.25) {$(-d_B+\alpha,d_B+\alpha-\epsilon)$};

		\node (15) at (1.75,-8.25) {\color{olive}$B$};
		\node (16) at (-0.5,-8) {$(\rho+\alpha-\epsilon,\rho+\alpha)$};
		\node (17) at (-2.5,-7.5) {$(-a{\color{red}-x}+\rho,-b{\color{red}+x}+\rho)$};

		\node (18) at (-3,-9) {$(-a{\color{red}-x}+\rho,a{\color{red}+x}-f+\rho+\alpha)$};
		\node (19) at (-0.5,-9.5) {$(d_A+\rho+\alpha-\epsilon,-d_A+\rho+\alpha)$};

		{\color{teal}
		\path (1) edge node [left, pos=0.5] {\tiny $A$}  (2);
		\path (1) edge node [above, pos=0.4] {\tiny $\mathfrak{I}$}  (3); 
		\path (1) edge node [left]  {\tiny $D$} (20);
		\path (1) edge node [above]  {\tiny $H$} (21);}
		
		{\color{olive}
		\path (2) edge node [above] {\tiny $S$}  (4);
		\path (2) edge node [left, pos=0.3] {\tiny $\mathfrak{I}$}  (5);
		\path (2) edge node [above] {\tiny $H$}  (6);
		\path (2) edge node [left] {\tiny $D$}  (7);
		 
		\path (3) edge node [above, pos=0.5] {\tiny $\mathfrak{I}$}  (8);
		\path (3) edge node [left] {\tiny $S$} (24); 
		\path (3) edge node [left] {\tiny $H$}  (9);
		\path (3) edge node [above] {\tiny $D$}  (10); 
		
		\path (20) edge node [below] {\tiny $\mathfrak{I}$} (23); 
		\path (20) edge node [below] {\tiny $P$} (22);}
		
		{\color{teal}
		\path (7) edge node [below] {\tiny $P$}  (11); 
		\path (7) edge node [below] {\tiny $\mathfrak{I}$}  (12);
		
		\path (10) edge node [above, pos=0.7] {\tiny $P$}  (13);
		\path (10) edge node [below] {\tiny $\mathfrak{I}$}  (14);
		
		\path (5) edge node [left] {\tiny $D$}  (15);
		\path (5) edge node [above] {\tiny $H$}  (16); 
		\path (5) edge node [above] {\tiny $\mathfrak{I}$}  (17);}
		
		{\color{olive}
		\path (15) edge node [below] {\tiny $P$}  (18);
		\path (15) edge node [below] {\tiny $\mathfrak{I}$}  (19);}
	\end{tikzpicture}
	\caption{Subgames $S_{3,4}$, $S'_{3,4}$ with Update $(a,b)\mapsto (a+x,b-x)$.}
	\label{tbl:S4p}
\end{figure}

\subsection{Subgames of the Routing Game} \label{app:rout}

In this section, one subgame of each type is detailed. First, subgame $\mathbb{S}_1$ in \Cref{fig:LA} describes the case where player $A$ locks an amount of money in the HTLC which deviates from the expected $m+3f$. The action $L_w$ means that the subsequent player follow along and forward the deviation of $-w$ to player $B$.
Subgame $\mathsf{S}_2$ in \Cref{fig:LH} illustrates the case that player $E_1$ creates her own secret and uses its hash $z$ as the lock of her HTLC. Action $L_z$ describes the reusing of hash lock $z$ in the next HTLCs. Lastly, subgame $\mathbf{S}_3$ in \Cref{fig:LT} handles the case, where player $I$ uses a time-out $t_3$ of the HTLC which is later than the previous ones $t_1$ and $t_2$, thereby neglecting the decreasing ordering of time outs. In \Cref{fig:LT}, the action $U_{t_2}$ means unlocking the HTLC before $t_2$ times out, thus enabling the other players to unlock too. The action $U_{>t_1}$ stands for unlocking after both $t_1$ and $t_2$ timed out, therefore $I$ and $E_1$ cannot unlock their respective HLTCs any more. Finally, action $U_{[t_2,t_1]}$ means unlocking after $t_2$ has timed out, but the HTLC with time-out $t_1$ can still be unlocked.

\begin{figure*}
	\centering
	\begin{tikzpicture}[scale=0.975,->,>=stealth',auto,node distance=2cm, el/.style = {inner sep=2pt, align=left, sloped}]

		\node (2) at (0,-1) {$E_1$};
		\node (8) at (-1,-2.5) {$(-\epsilon,0,0,0,0)$} ;
		\node (14) at (1.5, -1.5) {$I$};
		\node (15) at (0.5, -3) {$(-\epsilon,-\epsilon,0,0,0)$};
		\node (3) at (3,-2) {$E_2$};
		\node (9) at (2,-3.5) {$(-\epsilon,-\epsilon,-\epsilon,0,0)$} ;
		\node (16) at (4.5,-2.5) {$B$};
		\node (17) at (3.5, -4) {$(-\epsilon,-\epsilon,-\epsilon,-\epsilon,0)$};
		\node (4) at (6,-3) {$E_2$};

		    \node (10) at (6,-2) {$B$} ;
		
		    \node (56) at (4.5,-1.5) {$\mathfrak{S}_{10}$} ;
		    \node (57) at (6,-1) {$I$} ;
		    \node (58) at (9.75,-2.75) {$(m+3f+\rho-\epsilon,-\epsilon,-\epsilon,-m+w,-w+\rho)$} ;
		    
		    \node (59) at (6,0) {$E_1$} ;
		    \node (61) at (4.5,-0.5) {$\mathfrak{S}_{11}$} ;
		    \node (60) at (7.,-1.25) {$B$} ;
		    
		    \node (62) at (11.,0.6) {$(m+3f+\rho-\epsilon,-m+w-2f,m-w+2f-\epsilon,-m+w,-w+\rho)$} ;
		    \node (63) at (2.75,.6) {$(w+\rho,f,m-w+2f-\epsilon,-m+w,-w+\rho)$} ;
		    
		    \node (64) at (11,-2) {$(m+3f+\rho-\epsilon,-\epsilon,-\epsilon,-m+w,-w+\rho)$} ;
		    \node (65) at (8.5,-.75) {$E_1$} ;
		    
		    \node (66) at (12,-1.25) {$(m+3f+\rho-\epsilon,-\epsilon,-\epsilon,-m+w,-w+\rho)$} ;
		    \node (67) at (12,-.25) {$(w+\rho,m-w+3f-\epsilon,-\epsilon,-m+w,-w+\rho)$} ;
		
		\node (18) at (7.5, -3.5) {$I$};
		\node (19) at (6.75, -4.5) {$E_2$};
		    \node (48) at (7.25,-5.5) {$E_1$};
			\node (51) at (6,-5.125) {$B$};	    
		    
		    \node (49) at (11.5,-6) {$(m+3f+\rho-\epsilon,-\epsilon,-m+w-f,f,-w+\rho)$};
		    \node (50) at (12,-5.25) {$(w+\rho,m-w+3f-\epsilon,-m+w-f,f,-w+\rho)$};

		    \node (52) at (1.5,-4.75) {$(m+3f+\rho-\epsilon,-\epsilon,-m+w-f,f,-w+\rho)$};
		    \node (53) at (6,-6.) {$E_1$};
		    
		    \node (54) at (1.5,-6) {$(m+3f+\rho-\epsilon,-\epsilon,-m+w-f,f,-w+\rho)$};
		    \node (55) at (2,-5.4) {$(w+\rho,m-w+3f-\epsilon,-m+w-f,f,-w+\rho)$};
		
		\node (5) at (9,-4) {$E_1$};
		\node (6) at (11,-4.625) {$(w+\rho,f,f,f,-w+\rho)$} ;
		\node (11) at (12.,-3.5) {$(m+3f+\rho-\epsilon,-m+w-2f,f,f,-w+\rho)$};

			\node (32) at (0.5,0.1) {$\mathsf{S}_5$};
			\node (33) at (0,-2) {$\mathbb{S}_5$};
			\node (81) at (-.5,0.1) {$\mathbf{S}_5$};
			
			\node (34) at (2,-0.4) {$\mathsf{S}_6$};
			\node (35) at (1.5,-2.5) {$\mathbb{S}_6$};
			\node (82) at (1,-0.4) {$\mathbf{S}_6$};
			
			\node (36) at (3.5,-0.9) {$\mathsf{S}_7$};
			\node (37) at (3,-3) {$\mathbb{S}_7$};
			\node (83) at (2.5,-0.9) {$\mathbf{S}_7$};
			
			\node (42) at (4.5,-3.5) {$\mathfrak{S}_7$};
			
			\node (44) at (5.5,-4) {$\mathfrak{S}_8$};
			
			\node (45) at (8,-4.5) {$\mathfrak{S}_9$};

		\path (2) edge node[above] {\tiny $L_w$} (14);
		\path (2) edge node[above, pos=0.5] {\tiny $\mathfrak{I}$} (8);
		\path (14) edge node[ above] {\tiny $L_w$} (3);
		\path (14) edge node[ above, pos=0.5] {\tiny $\mathfrak{I}$} (15);
		\path (3) edge node[above] {\tiny $L_w$} (16);
		\path (3) edge node[above, pos=0.5] {\tiny $\mathfrak{I}$} (9);
		\path (16) edge node[ above] {\tiny $U$} (4);
		\path (16) edge node[ above, pos=0.5] {\tiny $\mathfrak{I}$} (17);
		\path (4) edge node[above] {\tiny $U$} (18);
		\path (4) edge node[left, pos=0.5] {\tiny $\mathfrak{I}$} (10); 
		\path (18) edge node[ above] {\tiny $U$} (5);
		\path (18) edge node[ above, pos=0.5] {\tiny $\mathfrak{I}$} (19);
		\path (5) edge node[right, pos=1] {\tiny $U$} (6);   
		\path (5) edge node[below, pos=1] {\tiny $\mathfrak{I}$} (11);

			\path (2) edge node[right, pos=.3] {\tiny $L_H$} (32);
			\path (2) edge node[right] {\tiny $L_A$} (33);
			\path (2) edge node[right] {\tiny $L_T$} (81);
			
			\path (14) edge node[right, pos=.3] {\tiny $L_H$} (34);
			\path (14) edge node[right] {\tiny $L_A$} (35);
			\path (14) edge node[right] {\tiny $L_T$} (82);
			
			\path (3) edge node[right, pos=.3] {\tiny $L_H$} (36);
			\path (3) edge node[right] {\tiny $L_A$} (37);
			\path (3) edge node[right] {\tiny $L_T$} (83);
			
			\path (16) edge node[right] {\tiny $S_S$} (42);
			
			\path (4) edge node[right] {\tiny $S_S$} (44);
			
			\path (18) edge node[right] {\tiny $S_S$} (45);
			
			\path (10) edge node[below, pos=0.5] {\tiny $S_{S_{E_1}}$} (56);
			\path (10) edge node[left, pos=0.5] {\tiny $S_{S_{I}}$} (57);
			\path (10) edge node[above] {\tiny $\mathfrak{I}$} (58);
			
			\path (57) edge node[left] {\tiny $U$} (59);
			\path (57) edge node[above] {\tiny $\mathfrak{I}$} (60);
			\path (57) edge node[below, pos=0.5] {\tiny $S_{S_{E_1}}$} (61);
			
			\path (59) edge node[below] {\tiny $\mathfrak{I}$} (62);
			\path (59) edge node[below] {\tiny $U$} (63);
			
			\path (60) edge node[above] {\tiny $\mathfrak{I}$} (64);
			\path (60) edge node[above] {\tiny $S_{S_{E_1}}$} (65);
			
			\path (65) edge node[right, pos=1] {\tiny $U$} (67);
			\path (65) edge node[right, pos=1] {\tiny $\mathfrak{I}$} (66);
			\path (19) edge node[right, pos=0.5] {\tiny $S_{S_{E_1}}$} (48);
			\path (19) edge node[above] {\tiny $\mathfrak{I}$} (51);
			
			\path (48) edge node[below, pos=0.3] {\tiny $\mathfrak{I}$} (49);
			\path (48) edge node[above] {\tiny $U$} (50);
			
			\path (51) edge node[above] {\tiny $\mathfrak{I}$} (52);
			\path (51) edge node[right] {\tiny $S_{S_{E_1}}$} (53);
			
			\path (53) edge node[below] {\tiny $\mathfrak{I}$} (54);
			\path (53) edge node[above, pos=0.1] {\tiny $U$} (55);

	\end{tikzpicture}
	\caption{Subgame $\mathbb{S}_1$ with locked amount $m-w+3f$.}
	\label{fig:LA}
\end{figure*}

\begin{figure*}
	\centering
	\begin{tikzpicture}[scale=0.975,->,>=stealth',auto,node distance=2cm, el/.style = {inner sep=2pt, align=left, sloped}]

		\node (2) at (0,-1) {$I$};
		\node (8) at (-1,-2.5) {$(-\epsilon,-\epsilon,0,0,0)$} ;
		\node (14) at (1.5, -1.5) {$E_2$};
		\node (15) at (0.5, -3) {$(-\epsilon,-\epsilon,-\epsilon,0,0)$};
		\node (3) at (3,-2) {$E_1$};
		\node (9) at (2,-3.5) {$(-\epsilon,-\epsilon,-\epsilon,-\epsilon,0)$} ;
		\node (16) at (4.5,-2.5) {$B$};
		\node (17) at (3.5, -4) {$(-\epsilon,-\epsilon,-\epsilon,-\epsilon,0)$};
		\node (4) at (6,-3) {$E_2$};

		    \node (10) at (6,-2) {$B$} ;
		
		    \node (56) at (4.5,-1.5) {$\mathfrak{S}_{16}$} ;
		    \node (58) at (6,-0.5) {$(m+3f+\rho-\epsilon,-\epsilon,-\epsilon,-m,\rho)$} ;
		    
		    
		    
		    
		
		\node (18) at (7.5, -3.5) {$I$};
		
		    \node (19) at (6.75, -4.5) {$B$};
		    \node (48) at (7.75,-5.125) {$\mathfrak{S}_{17}$};
			\node (51) at (4,-5.125) {$(m+3f+\rho-\epsilon,-\epsilon,-m+f,f,\rho)$};	    
		    

		    
		
		\node (5) at (9,-4) {$B$};
		\node (86) at (9,-5) {$\mathfrak{S}_{15}$};
		\node (6) at (10.5,-4.5) {$E_1$};
		\node (84) at (12.5,-5.125) {$(\rho,f,f,f,\rho)$} ;
		\node (11) at (9.,-2.5) {$(m+3f+\rho-\epsilon,-m-2f,f,f,\rho)$};
		\node (85) at (12.,-3.5) {$(m+3f+\rho-\epsilon,-m-2f,f,f,\rho)$};

			\node (32) at (0.5,0.1) {$\mathsf{S}_8$};
			\node (33) at (0,-2) {$\mathbb{S}_8$};
			\node (81) at (-.5,0.1) {$\mathbf{S}_8$};
			
			\node (34) at (2,-0.4) {$\mathsf{S}_9$};
			\node (35) at (1.5,-2.5) {$\mathbb{S}_9$};
			\node (82) at (1,-0.4) {$\mathbf{S}_9$};
			
			\node (36) at (3,-3) {$\mathfrak{S}_{12}$};
			
			\node (42) at (4.5,-3.5) {$\mathfrak{S}_{13}$};
			
			\node (44) at (6,-4) {$\mathfrak{S}_{14}$};
			

		\path (2) edge node[above] {\tiny $L_z$} (14);
		\path (2) edge node[above, pos=0.5] {\tiny $\mathfrak{I}$} (8);
		\path (14) edge node[ above] {\tiny $L_z$} (3);
		\path (14) edge node[ above, pos=0.5] {\tiny $\mathfrak{I}$} (15);
		\path (3) edge node[above] {\tiny $S_{S_{B}}$} (16);
		\path (3) edge node[above, pos=0.5] {\tiny $\mathfrak{I}$} (9);
		\path (16) edge node[ above] {\tiny $U$} (4);
		\path (16) edge node[ above, pos=0.5] {\tiny $\mathfrak{I}$} (17);
		\path (4) edge node[above] {\tiny $U$} (18);
		\path (4) edge node[left, pos=0.5] {\tiny $\mathfrak{I}$} (10); 
		\path (18) edge node[ above] {\tiny $U$} (5);
		\path (18) edge node[ above, pos=0.5] {\tiny $\mathfrak{I}$} (19);
		\path (5) edge node[above, pos=0.5] {\tiny $U$} (6);   
		\path (5) edge node[left, pos=0.5] {\tiny $\mathfrak{I}$} (11);
		\path (5) edge node[left, pos=0.5] {\tiny $S_S$} (86);
		\path (6) edge node[above, pos=0.5] {\tiny $U$} (84); 
		\path (6) edge node[above, pos=0.3] {\tiny $\mathfrak{I}$} (85); 

			\path (2) edge node[right, pos=.3] {\tiny $L_H$} (32);
			\path (2) edge node[right] {\tiny $L_A$} (33);
			\path (2) edge node[right] {\tiny $L_T$} (81);
			
			\path (14) edge node[right, pos=.3] {\tiny $L_H$} (34);
			\path (14) edge node[right] {\tiny $L_A$} (35);
			\path (14) edge node[right] {\tiny $L_T$} (82);
			
			\path (3) edge node[right, pos=.3] {\tiny $S_S$} (36);
			
			\path (16) edge node[right] {\tiny $S_S$} (42);
			
			\path (4) edge node[right] {\tiny $S_S$} (44);
			
			
			\path (10) edge node[below, pos=0.5] {\tiny $S_{S}$} (56);
			\path (10) edge node[left] {\tiny $\mathfrak{I}$} (58);
			
			
			
			
			\path (19) edge node[above, pos=0.5] {\tiny $S_S$} (48);
			\path (19) edge node[above] {\tiny $\mathfrak{I}$} (51);
			
			
			

	\end{tikzpicture}
	\caption{Subgame $\mathsf{S}_2$ with used hash lock $z$.}
	\label{fig:LH}
\end{figure*}

\begin{figure*} \label{fig:LT}
	\centering
	\begin{tikzpicture}[scale=0.975,->,>=stealth',auto,node distance=2cm, el/.style = {inner sep=2pt, align=left, sloped}]

		\node (3) at (3,-2) {$E_2$};
		\node (9) at (2,-3.5) {$(-\epsilon,-\epsilon,-\epsilon,0,0)$} ;
		\node (16) at (4.5,-2.5) {$B$};
		\node (17) at (3.5, -4) {$(-\epsilon,-\epsilon,-\epsilon,-\epsilon,0)$};
		\node (4) at (6,-3) {$E_2$};

		    \node (70) at (5.,-4.5) {$(m+3f+\rho-\epsilon,-\epsilon,-m-f,f,\rho)$};
		    \node (71) at (9, -2.25) {$(m+3f+\rho-\epsilon,-\epsilon,-\epsilon,-m,\rho)$};
		
		    \node (10) at (6,-2) {$B$} ;
		
		    \node (56) at (4.5,-1.5) {$\mathfrak{S}_{21}$} ;
		    \node (57) at (7.5,-1.5) {$E_2$} ;
		    
		    \node (61) at (5.5,-1) {$\mathfrak{S}_{22}$} ;
		    \node (60) at (11.5,-1.5) {$(m+3f+\rho-\epsilon,-\epsilon,-m-f,f,\rho)$} ;
		    
		    
		    
		
		\node (18) at (7.5, -3.5) {$I$};
		\node (19) at (8.5, -4.5) {$E_2$};
		    \node (48) at (7.,-5) {$\mathfrak{S}_{23}$};
			\node (51) at (9.5,-5.5) {$B$};	    
		    

		    \node (52) at (14,-6) {$(m+3f+\rho-\epsilon,-\epsilon,-m-f,f,\rho)$};
		    \node (53) at (8,-6.) {$\mathfrak{S}_{24}$};
		    
		
		\node (5) at (9,-4) {$E_1$};
		\node (6) at (11,-4.625) {$(\rho,f,f,f,\rho)$} ;
		\node (11) at (12.,-3.5) {$(m+3f+\rho-\epsilon,-m-2f,f,f,\rho)$};

			
			
			\node (36) at (3.5,-0.9) {$\mathsf{S}_{10}$};
			\node (37) at (3,-3) {$\mathbb{S}_{11}$};
			\node (83) at (2.5,-0.9) {$\mathbf{S}_{12}$};
			
			\node (42) at (4.5,-3.5) {$\mathfrak{S}_{18}$};
			
			\node (44) at (6,-4) {$\mathfrak{S}_{19}$};
			
			\node (45) at (9,-3) {$\mathfrak{S}_{20}$};

		\path (3) edge node[above] {\tiny $L_w$} (16);
		\path (3) edge node[above, pos=0.5] {\tiny $\mathfrak{I}$} (9);
		\path (16) edge node[ above] {\tiny $U$} (4);
		\path (16) edge node[ above, pos=0.5] {\tiny $\mathfrak{I}$} (17);
		\path (4) edge node[above, pos=0.7] {\tiny $U_{<t_2}$} (18);
		\path (4) edge node[left, pos=0.9] {\tiny $U_{[t_2,t_1]}$} (10); 
		\path (18) edge node[ above] {\tiny $U$} (5);
		\path (18) edge node[ above, pos=0.5] {\tiny $\mathfrak{I}$} (19);
		\path (5) edge node[right, pos=1] {\tiny $U$} (6);   
		\path (5) edge node[below, pos=1] {\tiny $\mathfrak{I}$} (11);

			
			
			\path (3) edge node[right, pos=.3] {\tiny $L_H$} (36);
			\path (3) edge node[right] {\tiny $L_A$} (37);
			\path (3) edge node[right] {\tiny $L_T$} (83);
			
			\path (16) edge node[right] {\tiny $S_S$} (42);
			
			\path (4) edge node[right] {\tiny $S_S$} (44);
			\path (4) edge node[left, pos=0.3] {\tiny $U_{>t_1}$} (70);
			\path (4) edge node[above] {\tiny $\mathfrak{I}$} (71);
			
			\path (18) edge node[above] {\tiny $S_S$} (45);
			
			\path (10) edge node[above, pos=0.5] {\tiny $S_S$} (56);
			\path (10) edge node[above, pos=0.5] {\tiny $S_{S_{I}}$} (57);
			
			\path (57) edge node[above] {\tiny $\mathfrak{I}$} (60);
			\path (57) edge node[above, pos=0.5] {\tiny $S_{S_{E_1}}$} (61);
			
			
			
			
			\path (19) edge node[below, pos=0.2] {\tiny $S_{S_{E_1}}$} (48);
			\path (19) edge node[above] {\tiny $\mathfrak{I}$} (51);
			
			
			\path (51) edge node[above] {\tiny $\mathfrak{I}$} (52);
			\path (51) edge node[below, pos=0.2] {\tiny $S_{S_{E_1}}$} (53);
			

	\end{tikzpicture}
	\caption{Subgame $\mathbf{S}_3$ with time-out ordering $t_3>t_1>t_2>t_4$.}
\end{figure*}

\end{document}